\documentclass[a4paper,USenglish,thm-restate,numberwithinsect,cleveref]{lipics-v2021}

\nolinenumbers
\pdfoutput=1
\hideLIPIcs

\title{Where Treewidth and Pathwidth Diverge: Towards a Uniform Kernel for \texorpdfstring{Pathwidth-$\eta$ Deletion}{Pathwidth-eta Deletion}}

\titlerunning{Towards a Uniform Kernel for \texorpdfstring{Pathwidth-$\eta$ Deletion}{Pathwidth-eta Deletion}}

\author{Ahmed Ghazy}{CISPA Helmholtz Center for Information Security, Saarbrücken, Germany \and Saarland University, Saarbrücken, Germany}{ahmed.ghazy@cispa.de}{https://orcid.org/0009-0009-7414-5871}{}
\author{Jakob Greilhuber}{CISPA Helmholtz Center for Information Security, Saarbrücken, Germany \and Saarland University, Saarbrücken, Germany}{jakob.greilhuber@cispa.de}{https://orcid.org/0009-0001-8796-6400}{}
\author{Tim A.\ Hartmann}{CISPA Helmholtz Center for Information Security, Saarbrücken, Germany}{hartmann@algo.rwth-aachen.de}{https://orcid.org/0000-0002-1028-6351}{}
\author{Roohani Sharma}{Discrete Mathematics Group, Institute for Basic Science, Daejeon, South Korea}{roohani@ibs.re.kr}{https://orcid.org/0000-0003-2212-1359}{Supported by the Young Scientist Fellowship of the Institute for Basic Science (IBS-R029-Y8).}

\authorrunning{A.~Ghazy, J.~Greilhuber, T.~A.~Hartmann, and R.~Sharma}

\Copyright{Ahmed Ghazy, Jakob Greilhuber, Tim A.\ Hartmann, and Roohani Sharma}



\category{}

\relatedversion{Full version of ESA 2026 contribution.}

\acknowledgements{A significant part of this work was done while A.\ Ghazy, J.\ Greilhuber, and T.~A.\ Hartmann were visiting R.\ Sharma at the Institute for Basic Science, Daejeon.}

\usepackage{mathtools}
\usepackage[noadjust]{cite}
\usepackage[most]{tcolorbox}
\usepackage{tabularx}
\usepackage{bm}
\usepackage{xparse}
\usepackage{xspace}

\usepackage[ruled,linesnumbered]{algorithm2e}
\DontPrintSemicolon
\SetKwInput{Input}{Input}
\SetKw{Continue}{continue}

\makeatletter
\let\original@algocf@latexcaption\algocf@latexcaption
\long\def\algocf@latexcaption#1[#2]{\@ifundefined{NR@gettitle}{\def\@currentlabelname{#2}}{\NR@gettitle{#2}}\original@algocf@latexcaption{#1}[{#2}]}
\makeatother

\crefname{claim}{Claim}{Claims}
\crefname{observation}{Observation}{Observations}
\crefname{mtheorem}{Main Theorem}{Main Theorems}
\Crefname{case}{Case}{Cases}

\newcommand{\refline}[1]{Line~\ref{#1}}
\newcommand{\reflines}[2]{Lines~\ref{#1} to~\ref{#2}}

\newcommand{\np}{\mathsf{NP}}
\newcommand{\conppoly}{\mathsf{coNP}/\mathsf{poly}}

\newcommand{\Oh}{{O}}

\newcommand{\range}[2][1]{\IfStrEq{#1}{1}{[#2]}{[#1,#2]}}
\newcommand{\nat}{\mathbb{N}}

\newcommand{\problemname}[1]{\textnormal{\textsc{#1}}\xspace}
\newcommand{\upsf}[1]{\textnormal{\textsf{#1}}}
\newcommand{\algofont}[1]{\textnormal{\textsf{#1}}}
\newcommand{\parameterfont}[1]{\textnormal{\texttt{#1}}}

\newcommand{\CC}{\mathcal{C}}
\newcommand{\FF}{\mathcal{F}}
\newcommand{\GG}{\mathcal{G}}
\newcommand{\HH}{\mathcal{H}}
\newcommand{\PP}{\mathcal{P}}
\newcommand{\XX}{\mathcal{X}}
\newcommand{\II}{\mathcal{I}}

\DeclareMathOperator{\pw}{pw}
\DeclareMathOperator{\vc}{\textnormal{\texttt{vc}}}
\DeclareMathOperator{\tw}{tw}
\DeclareMathOperator{\td}{td}
\DeclareMathOperator{\anc}{anc}
\DeclareMathOperator{\pl}{pl}
\DeclareMathOperator{\desc}{desc}

\DeclareMathOperator{\depth}{depth}

\DeclarePairedDelimiter\floor{\lfloor}{\rfloor}

\newcommand{\leftBagAny}{\ell_\exists}
\newcommand{\rightBagAny}{r_\exists}
\newcommand{\leftBagAll}{\ell_\forall}
\newcommand{\rightBagAll}{r_\forall}
\newcommand{\elimDist}[2][\HH]{\upsf{ed}_{#1}(#2)}
\newcommand{\pwElimDist}[1]{\elimDist[{\pwEtaBoundClass[1]{}}]{#1}}

\newcommand{\pendants}{\ensuremath{\mathsf{pendants}}}

\newcommand{\defparproblem}[4]{\begin{tcolorbox}[
            enhanced,
            colback=white,
            colframe=black,
            boxrule = 0.5pt,
            coltitle=black,
            title=#1,
            rounded corners,
            attach boxed title to top left={yshift=-10pt, xshift=6pt},
            bottom=-0.5mm,
            boxed title style={
                    interior style={fill=white},
                    frame hidden
                }
        ]\begin{tabularx}{12.5cm}{r X p {0.5cm}}
            Input:     & #2 \\
            Parameter: & #3 \\
            Question:  & #4
        \end{tabularx}
    \end{tcolorbox}
}

\newcommand{\pwEtaDeletion}{\problemname{Pathwidth-$\eta$ Deletion}}
\newcommand{\twEtaDeletion}{\problemname{Treewidth-$\eta$ Deletion}}
\newcommand{\tdEtaDeletion}{\problemname{Treedepth-$\eta$ Deletion}}

\newcommand{\fDeletion}{\problemname{$\mathcal{F}$-Minor-free Deletion}}

\newcommand{\pwEtaDelElimGeneral}{\hyperref[problem:pw_eta_delim_dist_to_g]{\problemname{Pathwidth-$\eta$ Deletion}/\parameterfont{dist-$\GG_{\pw \leq \eta}^\beta$}}}
\newcommand{\pwEtaDeletionEta}{\problemname{Pathwidth-$\eta$ Deletion/\parameterfont{dist-$\GG_{\pw \leq \eta}$}}}

\newcommand{\pwEtaDeletionDistToG}[1][\GG]{\hyperref[problem:pw_eta_delim_dist_to_g]{\problemname{Pathwidth-$\eta$ Deletion}/\parameterfont{dist-\ensuremath{#1}}}\xspace}
\newcommand{\vertexcover}{\problemname{Vertex Cover}}
\newcommand{\feedbackvertexset}{\problemname{Feedback Vertex Set}}
\newcommand{\outerplanardeletion}{\problemname{Outerplanar Deletion}}
\newcommand{\planarVertexDeletion}{\problemname{Planar Vertex Deletion}}

\newcommand{\degreeBound}[1][|M|,k]{\hyperref[def:important_bounds_degree_reduction]{\upsf{B}_{\upsf{deg}}(#1)}} \newcommand{\degreeBoundCC}[1][|M|]{\hyperref[def:important_bounds_degree_reduction]{\upsf{B}_{\upsf{deg in comp}}(#1)}} \newcommand{\degreeBoundTreedepthCase}{\hyperref[def:important_bounds]{\upsf{B}_{\FF}}} \newcommand{\degreeBoundCaterpillarCase}[1][|M|]{\hyperref[def:important_bounds_degree_reduction]{{\upsf{B}}_{\upsf{cat}}(#1)}} \newcommand{\ccBound}[1][|M|,k]{\hyperref[def:important_bounds_degree_reduction]{\upsf{B}_\upsf{\#comp}(#1)}} \newcommand{\polishingSetSize}[1][|M|,k]{\hyperref[def:important_bounds_degree_reduction]{\upsf{B}_{\upsf{polish}}(#1)}} \newcommand{\BoundCaterpillar}[1][|N_G(C)|]{\ensuremath{f_{\ref{thm:caterpillar_case}}(#1)}} \newcommand{\BoundCaterpillarOverall}[1][|M|]{\ensuremath{(8(\eta+1)(#1+2(\eta+1)+\beta+5))^{8}}} 

\newcommand{\numNeighborsInDepth}{{\hyperref[def:important_bounds]{\upsf{B}_3}}}\newcommand{\markedNumb}{{\hyperref[def:important_bounds]{\upsf{B}_\upsf{\#marked}}}}
\newcommand{\markedNumbChildren}{{\hyperref[def:important_bounds]{\upsf{B}_2}}}
\newcommand{\numbChildren}{{\hyperref[def:important_bounds]{\upsf{B}_1}}}
\newcommand{\boundOnS}{{\hyperref[def:important_bounds]{\upsf{B}_{0}}}}

\newcommand{\BoundProtrusionCaterpillarNeighborhood}{\hyperref[def:important_bounds_protrusion_replacement]{\upsf{B}_{9}}}
\newcommand{\BoundChildrenInForest}{\hyperref[def:important_bounds_protrusion_replacement]{\upsf{B}_{10}}}
\newcommand{\BoundProtrusionModulator}[1][|M|,k]{\hyperref[bound:protrusion_modulator]{\upsf{B}_{\upsf{PM}}(#1)}}
\newcommand{\BoundProtrusionSize}{{\ensuremath{\upsf{B}_{\ref{thm:reduce_protrusions}}}}}

\newcommand{\polishingSetAlgo}{\hyperref[thm:polish_instance]{\algofont{PolishingSet}}}
\newcommand{\protrusionAlgo}{\hyperref[thm:create_protrusions]{\algofont{CreateProtrusions}}}

\newtheoremstyle{casestyle}
	{\topsep}
	{2pt}
	{\itshape}
	{}
	{\color{lipicsGray}\sffamily\bfseries}
	{. }
	{0pt}{\thmname{#1}\thmnumber{ #2}\thmnote{ (#3)}}
\theoremstyle{casestyle}
\newtheorem{case}{Case}

\newcommand{\pwEtaBoundClass}[1][1]{\mathcal{G}_{\pw \leq #1}} 
\begin{document}

\maketitle

\begin{abstract}
  For a constant $\eta \geq 0$, \pwEtaDeletion{} is the problem of deciding whether, for a given graph $G$ and integer $k$, there is a set $S \subseteq V(G)$ of size at most $k$ such that the pathwidth of $G - S$ is at most $\eta$.
The problems \twEtaDeletion{} and \tdEtaDeletion{} are defined similarly for the parameters treewidth and treedepth, respectively.
A landmark result of Fomin et al.\ [{FOCS}, 2012] shows that, for any constant $\eta$, all three problems admit a kernel on $O(k^{c(\eta)})$ vertices,
where $c(\eta)$ is a constant depending on $\eta$.

Giannopoulou et al.\ [{ACM TALG}, 2017] 
show that, in some sense, this result is optimal for \twEtaDeletion{}: for $\eta \geq 2$ and even when parameterizing by the size of a vertex cover $M$ of the input graph, there is no kernel of size $O(|M|^{\frac{\eta+1}{2}-\varepsilon})$, for any $\varepsilon > 0$.
Contrasting this result, they prove that \tdEtaDeletion{} admits a uniform polynomial kernel, that is, a kernel of size $O(k^c)$ for a constant $c$ that is independent of $\eta$.

In comparison, the question whether \pwEtaDeletion{} admits a uniform polynomial kernel
has been neglected in the literature. As treewidth and pathwidth tend to behave similarly, it is natural to expect that no uniform kernel exists when parameterizing by the size of a vertex cover.
Surprisingly, we show this not to be the case.
More concretely, we prove the existence of a uniform polynomial kernel for \pwEtaDeletion{} when parameterizing by
\begin{bracketenumerate}
    \item the solution size $k$ plus the size of a set $M$ such that $G - M$ has bounded treedepth,
    \item the (vertex-deletion) distance to pathwidth-$1$ graphs,
    \item the distance to the class of graphs with treedepth at most $\eta + 1$.
\end{bracketenumerate}

This pinpoints a striking difference between \pwEtaDeletion{} and \twEtaDeletion{} and leads us to conjecture that \pwEtaDeletion{} admits a uniform kernel when parameterizing by the solution size $k$. \end{abstract}

\newpage

\section{Introduction}
\label{sec:intro}
The problem \fDeletion{} is defined for each finite family $\mathcal{F}$ of undirected graphs.
The input is an undirected graph $G$ and a non-negative integer $k$;
the question is whether there exists a set of at most $k$ vertices of $G$ whose deletion results in an $\FF$-minor free graph,
that is, a graph that has no graph from $\mathcal{F}$ as a minor.
A main objective in the field of parameterized complexity is to design
\emph{fixed-parameter tractable} (FPT) algorithms, which are algorithms with
running times of the form $f(p) \cdot n^{\Oh(1)}$ where $p$ is a parameter
given in the input and $n$ the input size.
One of the central goals in the field is finding the smallest parameter
for which an FPT algorithm exists for the considered problem.
The canonical starting point is to consider parameterizing the problem by the
solution size $k$.
Given the expressive power of the \fDeletion{} problem, determining the parameterized complexity
for every choice of $\FF$ is nothing short of a daunting task.

Indeed, the \fDeletion{} framework encompasses any vertex-deletion problem to a minor-closed graph class.\footnote{This is due to the seminal result of Robertson and Seymour~\cite{DBLP:journals/jct/RobertsonS04} who showed that minor-closed graph classes are well-quasi-ordered and hence admit a finite minor characterization~\cite[Corollary~2.1.1]{BIENSTOCK1995481}.
}
This includes classical $\np{}$-hard problems like \vertexcover{} (or equivalently vertex-deletion to treewidth $0$),
\feedbackvertexset{} (or equivalently vertex-deletion to treewidth at most $1$), \twEtaDeletion{}, \pwEtaDeletion{}, \tdEtaDeletion{} for any $\eta \geq 0$, \outerplanardeletion{} (vertex-deletion to outer-planarity), \planarVertexDeletion{} (vertex-deletion to planarity) and many more.

Parameterized by the solution size $k$,
a non-uniform\footnote{in the sense that the algorithm might be different for every $k$}
FPT algorithm for \fDeletion{} can be easily obtained via any minor-testing algorithm~\cite{DBLP:journals/jct/RobertsonS95b,DBLP:conf/focs/KorhonenPS24} because of the well-quasi-ordering of the graph class of positive instances.
A uniform FPT algorithm was later given by Fomin, Lokshtanov, Misra and Saurabh~\cite{fominPlanarFdeletionApproximation2012} for the case when $\mathcal{F}$ contains a planar graph.
This was later generalized by Sau, Stamoulis and Thilikos~\cite{DBLP:conf/icalp/SauST20} to every family~$\mathcal{F}$.

\paragraph*{Uniform Kernelization}
Another main objective in the field of parameterized complexity is to design
preprocessing algorithms (\emph{kernels}) that compress a given instance into an equivalent instance
of size bounded by some function (ideally polynomial) of the parameter $p$.
Again, it is desirable to find the smallest choice of parameter for which the
problem admits a polynomial kernel.
For the solution size parameterization, in contrast to the FPT landscape,
it is still a major open problem whether \fDeletion{} admits a polynomial kernel for
every $\FF$~\cite{fominPlanarFdeletionApproximation2012,giannopoulouUniformKernelizationComplexity2017,jansenLossyPlanarizationConstantfactor2025,bougeretKernelizationDichotomiesHitting2025a}.

On the upside, assuming that $\FF$ contains a planar graph,
a seminal work by Fomin, Lokshtanov, Misra and
Saurabh~\cite{fominPlanarFdeletionApproximation2012} gives a polynomial kernel for \fDeletion.
More concretely, they show that every \fDeletion{} problem admits a kernel of size
$\Oh(k^{f(\FF)})$ for some function $f$. 
Note that the kernel size is not explicit in the sense that not all involved constants are known.
Interestingly, the kernel size has a dependency on $\FF$ in the exponent.
This raises the question of which
\fDeletion{} problems admit a \emph{uniform} polynomial kernel;
that is, a kernel of size $g(\FF) \cdot k^c \in O(k^c)$ for some constant $c$ independent of $\FF$.

In this setting, due to the celebrated Grid Minor
Theorem~\cite{DBLP:journals/jct/RobertsonS86,DBLP:journals/jacm/ChekuriC16},
the graph obtained after removing a solution has treewidth bounded by a function
that depends on the size of the excluded planar graph,
making the family of \twEtaDeletion{} problems a prototypical subclass of \fDeletion{} when $\mathcal{F}$ contains a planar graph.

For any $\eta \geq 0$,
in the \twEtaDeletion{} problem, given an undirected graph $G$ and a non-negative integer $k$, the goal is to determine if there exists a set of at most $k$ vertices whose deletion results in a graph of treewidth at most $\eta$.
\tdEtaDeletion{} and \pwEtaDeletion{} are the same as \twEtaDeletion{} where treewidth is replaced with treedepth and pathwidth, respectively.
The aforementioned work of Fomin et al.~\cite{fominPlanarFdeletionApproximation2012} provides a (non-uniform) polynomial kernel, of size $g(\eta) \cdot k^{O(\eta)}$ for some function $g$, for \twEtaDeletion{}.

Does \textsc{Treewidth-$\eta$ Deletion} admit a uniform polynomial kernel?
For small $\eta$,
that misleadingly appears to be the case.
For $\eta=0$ (which is \vertexcover{}) as well as for $\eta=1$ (which is \feedbackvertexset{}),
there is a kernel with $O(k^2)$ edges, which is also essentially tight, unless $\np{} \subseteq \conppoly{}$~\cite{bussNondeterminism1993,DBLP:conf/wg/ChorFJ04,DBLP:journals/talg/Thomasse10,DellM2014}.
However, this does not extend to larger values of $\eta$.
Indeed, Giannopoulou et al.~\cite{giannopoulouUniformKernelizationComplexity2017} provide the following lower bound that rules out a uniform kernel even when parameterizing by the size of a minimum vertex cover $\vc$ of the input graph.

\begin{theorem}[{Giannopoulou et al.~\cite[Theorem 1.1]{giannopoulouUniformKernelizationComplexity2017}}]
  \label{theorem:treewidth:lb}
  For any function $g$ and any $\eta \geq 2$, $\varepsilon > 0$, \textnormal{\textsc{Treewidth-$\eta$ Deletion}} does not admit a kernel of size $g(\eta) \cdot \vc^{\frac{\eta+1}{2}-\varepsilon}$, unless $\np{} \subseteq \conppoly{}$.
\end{theorem}

Observe that $\vc \geq k$ for non-trivial problem instances, and therefore this lower bound extends to $k$ as the parameter.
In contrast, the same authors give a fully explicit uniform polynomial kernel for \textsc{Treedepth-$\eta$ Deletion}~\cite{giannopoulouUniformKernelizationComplexity2017}.
They obtain a polynomial kernel with $2^{O(\eta^2)} \cdot k^6 \in O(k^6)$ vertices for each $\eta$.
This shows a striking difference between treewidth and treedepth.

Given this contrast between these two width measures,
it is natural to explore whether similar problem families admit uniform polynomial kernels.
As the parameter {\em pathwidth} lies in between treewidth and treedepth,
our work addresses the following question:

\begin{quote}
  \centering
  Does \pwEtaDeletion{} admit a uniform polynomial kernel?
\end{quote}

When $\eta = 0$ the problem \pwEtaDeletion{} is the same as \vertexcover{}, which is well-known to admit kernels with a linear number of vertices, see e.g.,~\cite{DBLP:conf/wg/ChorFJ04,chenVertexCoverFurther2001,lampisKernelOrder22011,soleimanfallahKernelOrder2kc2011,fellowsWhatKnownVertex2018,abu-khzamCrownStructuresVertex2007}.
The case $\eta = 1$ has also been studied: a kernel with a quartic number of vertices was given by Philip et al.~\cite{philipQuarticKernelPathwidthone2010b}, which was subsequently improved to a kernel with a quadratic number of vertices~\cite{cyganImprovedFPTAlgorithm2012,tsurSmallerKernelsTwo2024}.
For each value of $\eta$, \pwEtaDeletion{} admits a (non-uniform) polynomial kernel simply because it is an \fDeletion{} problem
where $\FF$ contains a planar graph, namely a forest~\cite{DBLP:journals/jct/RobertsonS83}; hence, the result of Fomin et al.~\cite{fominPlanarFdeletionApproximation2012} applies.
As is well-known, treewidth and pathwidth behave quite similarly from an algorithmic perspective, and
there are only a few known scenarios where they behave differently (see also the discussion by Belmonte et al.~\cite{belmonteGrundyDistinguishesTreewidth2022}).
As such, it is tempting to conjecture that, just like \twEtaDeletion{}, also \pwEtaDeletion{} does not admit a uniform polynomial kernel, even when parameterizing by the vertex cover size.
In this work, contrary to this expectation, we show that \pwEtaDeletion{} \emph{does} admit a uniform polynomial kernel when parameterizing by vertex cover number.
In fact, we provide results that indicate that the problem might even admit a uniform polynomial kernel when parameterizing by the solution size.

\paragraph*{Towards a Uniform Polynomial Kernel}
Let us illustrate how our work makes progress on whether \pwEtaDeletion{} parameterized by the solution size $k$ admits a uniform kernel.
That is, the following algorithmic approach might help in eventually resolving this question.

We start with an approximate solution $M \subseteq V(G)$ of size $O(k)$
which can be obtained by using known approximation algorithms for \pwEtaDeletion{} \cite{guptaLosingTreewidthSeparating2019,fominPlanarFdeletionApproximation2012}.
Then, $G-M$ has a pathwidth bounded in $\eta$, but there is no bound on its size.
The goal is to design rules that let us reduce the following:
\begin{romanenumerate}
  \item the number of connected components of $G-M$, and
  \item the size of an arbitrary connected component of $G-M$.
\end{romanenumerate}
By bounding (i) and (ii) by polynomials in $k$, one immediately obtains a polynomial kernel.
This is also the general approach for the kernelization of \twEtaDeletion{} by Fomin et al.~\cite{fominPlanarFdeletionApproximation2012}.
They obtain \emph{non-uniform} polynomial bounds.

To obtain a \emph{uniform} polynomial kernel we need to bound (i) and (ii) by \emph{uniform} polynomials, that is by
$g(\eta) \cdot k^c$ for some function $g$ and universal constant $c$ independent of $\eta$.
This is something Giannopoulou et al.~\cite{giannopoulouUniformKernelizationComplexity2017} achieve for their uniform kernel for the \textsc{Treedepth-$\eta$ Deletion} problem.

Regarding (i) and (ii) for \pwEtaDeletion{} we show the following.
\begin{romanenumerate}
  \item We resolve (i).
  That is, we provide a preprocessing algorithm that reduces the number of connected components of $G - M$ to a \emph{uniform polynomial}.
  As a corollary, we obtain that \pwEtaDeletion{}, for $\eta \geq 0$, admits a kernel of size $g(\eta) \cdot {\vc}^{3} \in O({\vc}^3)$.
  This already shows a striking contrast between treewidth and pathwidth since, by \cref{theorem:treewidth:lb}, no such result is possible for \textsc{Treewidth-$\eta$ Deletion}.

  \item
  We make substantial progress towards (ii).
  That is, we provide a preprocessing algorithm that, in many albeit not all cases, bounds the size of a connected component of $G-M$ by a uniform polynomial.
  To explain this, let $C$ be a connected component of $G - M$.
  We give reduction rules that can reduce the size of $C$ assuming that
  (a) $C$ contains a long path in which every vertex has degree $2$
or (b) the treedepth of $C$ is bounded by \emph{any function} of $\eta$.
  Therefore, a component that we cannot deal with has unbounded treedepth and does not contain a long path in which every vertex has degree $2$.
\end{romanenumerate}

These findings lead us to conjecture that \pwEtaDeletion{} admits a uniform polynomial kernel.
\begin{conjecture}
  \label{conj:uniform_kernel}
  There is a function $g$ and a constant $c$ such that for each $\eta \geq 0$, \pwEtaDeletion{} has a kernel of size $g(\eta) \cdot k^c$.
\end{conjecture}

Our results already show that the kernelization complexity of \pwEtaDeletion{} is closer to \tdEtaDeletion{} than it is to \twEtaDeletion{}.
We believe that the approaches we develop in this work are useful in settling \cref{conj:uniform_kernel}, positively or negatively.
\begin{itemize}
  \item
        A positive resolution might follow the aforementioned path and extend our uniform polynomial bound for (ii) to all connected components of $G-M$ without needing any further restrictions.
  \item
        A negative resolution might derive a hardness result based on large connected components that are not covered by our approach for bounding (ii).
        In that case, it would be interesting to find the smallest parameter for which uniform kernels are possible.
        Our results already provide a good start.
\end{itemize}
\paragraph*{Our Results}

Let us present our main results formally.
For a graph class $\GG$, we consider \pwEtaDeletion{} parameterized by the size of a vertex set $M$ (called a \emph{modulator}) such that $G-M$ is in the graph class $\GG$.
For technical reasons we assume that $M$ is part of the input, but we remark that typically this assumption is without loss of generality in the sense that for many graph classes, an approximately optimal modulator can be computed in polynomial-time.

\defparproblem{\pwEtaDeletionDistToG\phantomsection{}\label{problem:pw_eta_delim_dist_to_g}}
{Graph $G$, integer $k \geq 0$, a vertex set $M \subseteq V(G)$ such that $G-M$ is in $\GG$}{$|M|$}{Is there a set $S \subseteq V(G)$ of size at most $k$ such that $G - S$ has pathwidth at most $\eta$?}

Ideally, one would obtain a uniform polynomial kernel for \pwEtaDeletion parameterized by the solution size $k$,
i.e., for \pwEtaDeletionDistToG where $\GG=\pwEtaBoundClass[\eta]$ is the set of graphs of pathwidth at most $\eta$.
In this work, we obtain a uniform polynomial kernel if we additionally restrict the graph class $\GG$ to graphs that have bounded elimination distance to graphs of pathwidth $1$.
Graphs with pathwidth at most one are forests in which each connected component is a caterpillar~\cite[Lemma 2.4]{arnborgMonadicSecondOrder1990}.

Here, the \emph{elimination distance} to $\HH$~\cite{bulianGraphIsomorphismParameterized2016,bulianFixedparameterTractableDistances2017} is a generalization of the parameter treedepth.
Recall that the treedepth is the minimum depth of an elimination forest whose leaves correspond to single isolated vertices.
Now, roughly speaking, the \emph{elimination distance to a graph class $\HH$} follows the same concept, but the leaves of the elimination forest now correspond to graphs in $\HH$.
We refer to \cref{def:elimination_trees_forests_distance} for details.
For a graph class $\HH$ and graph $G$, we write $\elimDist[\HH]{G}$ to denote the elimination distance of $G$ to $\HH$.
Next, we define the graph class used in our main result.

\begin{restatable}{definition}{defPWTwGraphClasses}
  For integers $i ,j\geq 0$, let $\GG_{\tw \leq i}$ ($\GG_{\pw \leq i}$, $\GG_{\td \leq i}$) be the class of graphs with treewidth (pathwidth, treedepth) at most $i$.
  Moreover, let $\GG^j$ be the class of graphs $G$ with $\pwElimDist{G} \leq j$
  and let $\GG_{\pw \leq i}^{j} = \GG_{\pw \leq i} \cap \GG^{j}$.
\end{restatable}

Now we are ready to state our main technical result.
\begin{restatable}{mtheorem}{mThmGeneral}
  \label{main_thm:uniform_kernel_general}
  For some function $g$ and all constants $\eta, \beta \geq 0$, \pwEtaDelElimGeneral{} admits a kernel on $g(\eta,\beta) \cdot |M|^{60}$ vertices.
  
  Moreover, given an instance $(G,k,M)$ of \pwEtaDelElimGeneral{}, the kernelization algorithm outputs an equivalent instance $(G',k',M)$ where $G'$ is a minor of $G$, and if $(G,k,M)$ is a no-instance then $k = k'$.
\end{restatable}
We stress that our kernel size is uniformly polynomial in $\eta$ and $\beta$.
While the constant in the exponent is quite large, our goal was not to optimize it, but just to obtain an exponent that is independent of $\eta$.

It is tempting to think that \cref{main_thm:uniform_kernel_general} can be obtained
by applying an iterative paradigm similarly to polynomial kernels for other \fDeletion{} problems,
as in e.g.,~\cite{jansen2020polynomial,hols2022elimination,bougeretKernelizationDichotomiesHitting2025a}.
That is, to build upon (1) a reduction rule to reduce the number of connected components of $G-M$ and (2) a uniform polynomial kernel for \pwEtaDeletion{} parameterized by the size of a modulator to $\pwEtaBoundClass[1]$ only.
Indeed, then one could bound the number of connected components of $G-M$ and lift (a bounded number of) roots of the elimination forest of $G - M$ into the modulator.
The result is an instance of $\pwEtaDeletionDistToG[\GG_{\pw \leq 1}^{\beta - 1}]$ where the modulator remains of bounded size.
After repeating this step for at most $\beta$ rounds, this would yield an instance of $\pwEtaDeletionDistToG[\GG_{\pw \leq 1}]$.
Then applying the kernel for $\pwEtaDeletionDistToG[\GG_{\pw \leq 1}]$ would result in a polynomial kernel.

However, this powerful lifting approach {fails} when aiming for a \emph{uniform} polynomial kernel
because each round increases the exponent of the kernel size.
After $\beta$ rounds, the size bound of the resulting kernel depends on $\beta$ in the exponent.
Therefore, it is highly non-trivial to obtain a uniform kernel for the parameterization	we consider, even when already having the two ingredients (1) and (2) at hand (which are themselves non-trivial tasks).

\paragraph*{Consequences}

While the graph class $\GG_{\pw \leq \eta}^\beta$ might seem artificial at first glance, our \cref{main_thm:uniform_kernel_general} implies multiple results for more standard graph classes.

\begin{restatable}[Of \cref{main_thm:uniform_kernel_general}]{corollary}{corollaryUniformKernelsNaturalClasses}
  \label{corollary:uniform_kernels_natural_classes}
  Each of the following problems admits a uniform polynomial kernel in $\eta, \gamma$:
  \begin{bracketenumerate}
    \item \pwEtaDeletion{} parameterized by $k + |M|$ where $M$ is a modulator to $\GG_{\td \leq \gamma}$ given in the input.

    \item \pwEtaDeletionDistToG[{\pwEtaBoundClass[1]}] if $\eta \geq 1$.

    \item \pwEtaDeletionDistToG[\GG_{\td \leq \eta + 1}]{}.

    \item \pwEtaDeletionDistToG[\GG^{\eta - 1}] if $\eta \geq 1$.

    \item \pwEtaDeletionDistToG[{\GG_{\pw \leq \eta}^{\pl \leq \gamma}}]{} where $\GG_{\pw \leq \eta}^{\pl \leq \gamma}$ is the class of graphs with pathwidth at most $\eta$ that do not contain a path on $\gamma + 1$ vertices as a subgraph.
  \end{bracketenumerate}
\end{restatable}

If we restrict ourselves to vertex cover parameterization, we obtain the following result with a much better kernel size as a corollary of our preprocessing rule to reduce the number of connected components.

\begin{restatable}{corollary}{thmUniformKernelVcParameterization}
  \label{thm:uniform_kernel_vc_parameterization}
  There is a function $g$ such that for every constant $\eta \geq 0$, \pwEtaDeletion{} parameterized by the size of a vertex cover $M$ given in the input admits a kernel on $g(\eta) \cdot |M|^3$ vertices.
  
  Moreover, given an instance $(G,k,M)$ of the problem, the kernelization algorithm outputs an instance $(G',k',M)$ where $G'$ is a subgraph of $G$, and if $(G,k,M)$ is a no-instance then $k = k'$.
\end{restatable}

Moreover, because our kernel outputs a minor of the input graph, our result has interesting consequences which are purely graph-theoretical.
These follow from a connection between kernels (that output a minor of the input graph) and the size of such obstructions, which has also been observed in other settings previously~\cite{fominPlanarFdeletionApproximation2012,giannopoulouUniformKernelizationComplexity2017,donkersPreprocessingOuterplanarVertex2022}.
To explain these, we first require the notion of $k$-apices and obstructions.

\begin{restatable}[Apices and Obstructions~\cite{DBLP:journals/jctb/SauST23}]{definition}{defApicesObstructions}
  \label{def:apices_obstructions}
  Let $\GG$ be a minor-closed graph class.
  For each integer $k \geq 0$, define $\GG_k$ to be the set of graphs that have a modulator to $\GG$ of size at most $k$, clearly $\GG_0 = \GG$.
  A graph $H \in \GG_k$ is a $k$-\emph{apex} of $\GG$.

  A graph $H \notin \GG$ is a \emph{minor-minimal obstruction} to $\GG$ if any proper minor of $H$ is in $\GG$.
\end{restatable}

Note that, for minor-closed class $\GG$ and any $k \geq 0$, the class $\GG_k$ is again minor-closed.
A consequence of the famous Graph Minor Theorem of Robertson and Seymour~\cite{DBLP:journals/jct/RobertsonS04} is that
for any minor-closed graph class $\GG$, the set of minor-minimal obstructions to
$\GG$ is finite~\cite[Corollary~2.1.1]{BIENSTOCK1995481}.
Using our kernel, we can obtain uniform bounds on the size of minor-minimal obstructions to the class of $k$-apices.

\begin{restatable}{theorem}{mainThmObstructionBound}
  \label{main_thm:obstruction_bound}
  There is a function $f: \nat \times \nat \rightarrow \nat$, such that
  for any $\eta, \beta \geq 0$ and any $k \geq 0$, any minor-minimal obstruction $H$ to the class of $k$-apices of $\pwEtaBoundClass[\eta]$ has at most $f(\eta,\beta) \cdot |M|^{60}$ vertices, where $M \subseteq V(H)$ is a smallest set such that $\pwElimDist{H - M} \leq \beta$ and $\pw(H -M) \leq \eta$.
\end{restatable}

Similarly, by using our kernel of \cref{thm:uniform_kernel_vc_parameterization}, we can bound the size of such graphs $H$ in terms of their \emph{vertex cover number} $\vc(H)$.

\begin{restatable}{theorem}{thmObstructionBoundVertexCover}
  \label{thm:obstruction_bound_vertex_cover}
  There is a function $f: \nat \rightarrow \nat$, such that
  for any $\eta \geq 0$ and any $k \geq 0$, any minor-minimal obstruction $H$ to the class of $k$-apices of $\pwEtaBoundClass[\eta]$ has at most $f(\eta) \cdot \vc(H)^{3}$ vertices.
\end{restatable}

\cref{thm:obstruction_bound_vertex_cover} is significant because the same statement is \emph{not true} for treewidth, as the following result shows.
Note that this result was informally stated by Giannopoulou et al.~\cite{giannopoulouUniformKernelizationComplexity2017}
as it indeed follows quite easily from a construction they use in a kernelization lower bound.
For the sake of completeness, we provide a proof of the theorem in \cref{sec:consequences}.

\begin{restatable}{theorem}{thmTreewidthObstructionsNonUniform}
  \label{thm:treewidth_obstructions_non_uniform}
  For any $\eta \geq 0$, and any $k \geq 0$, there is a minor-minimal obstruction $H$ to the graph class of $k$-apices of $\GG_{\tw \leq \eta}$ with $|V(H)| = k + \eta + 1 + \binom{\vc(H)}{\eta + 1}$ and $\vc{(H)} = k + \eta + 1$.
\end{restatable}

\cref{thm:treewidth_obstructions_non_uniform} shows that a result analogous to \cref{thm:obstruction_bound_vertex_cover} is \emph{impossible} for treewidth: the exponent in the size of obstructions in terms of their vertex cover size must always depend on $\eta$ in a non-uniform manner.

\subparagraph*{Related work.}
To the best of our knowledge, only few cases of uniform polynomial kernels for \fDeletion{} problems parameterized by the solution size are known.
Apart from the uniform kernel for \textsc{Treedepth-$\eta$ Deletion} \cite{giannopoulouUniformKernelizationComplexity2017}, uniform polynomial kernels exist when $\mathcal{F}$ contains a $\theta_p$~\cite{fominHittingForbiddenMinors2016} or, more generally, a $K_{2,p}$~\cite{DBLP:conf/isaac/LochetS24}.
Moreover, the problem family where for each fixed integer $d \geq 0$ the family $\FF$ contains all connected graphs on $d+1$ vertices also admits a uniform polynomial kernel \cite{kumar2lkKernelLComponent2016,xiaoLinearKernelsSeparating2017}.

Another problem for which uniform kernelization has been studied is $d$-\textsc{Set Packing}.
In this problem, given a universe $U$, a family $\mathcal{F}$ of $d$-sized subsets of $U$, and a positive integer $k$,
the goal is to decide if there exist $k$ sets in the family $\mathcal{F}$ which are pairwise disjoint. This problem admits a (non-uniform) polynomial kernel of size $f(d) \cdot k^d$, for every fixed $d$ and some function $f$~\cite{DBLP:journals/algorithmica/FellowsKNRRSTW08}.
A matching lower bound shows that there is no kernel for this problem which has size $g(d) \cdot k^{d-\varepsilon}$, for any $\varepsilon >0$, and any function $g$, unless $\np{}\subseteq \conppoly{}$~\cite{DBLP:conf/soda/DellM12}.
For the special case  $P_d$-\textsc{Set Matching},
where the universe is the vertex set of some underlying graph $G$ and $\mathcal{F}$ is the (only implicitly given) family of all paths of $G$ of length $d$, Dell and Marx~\cite{DBLP:journals/corr/abs-1812-03155} showed that a uniform kernel with respect to $d$ exists.

\paragraph*{Organization}
We present a high-level overview of the used techniques in \cref{sec:tech_overview}.
The preliminaries are given in \cref{sec:prelims}.
In \cref{sec:kernel:bounding_ccs} we obtain a uniform bound on the number of connected components.
The rest of the proof of \cref{main_thm:uniform_kernel_general} is given in \cref{sec:kernel:bounding_mod_degree,sec:kernel_reducing_cc_size}.
Further, \cref{sec:consequences} presents some consequences of our main result.
\cref{sec:appendix_constants} provides an overview of
	constants (depending on $\eta,\beta$) and functions used in our intermediate results. 
\section{Technical Overview}
\label{sec:tech_overview}
We now present a high-level overview of the proof of \cref{main_thm:uniform_kernel_general}.
The main tasks are to give uniform polynomial bounds for
\begin{romanenumerate}
    \item the number of connected components of $G-M$, and
    \item the size of an arbitrary connected component of $G-M$.
\end{romanenumerate}

\subsection{Bounding the Number of Connected Components}
As the first step, we bound the number of connected components by a uniform polynomial.
We stress that this preprocessing step also works for \pwEtaDeletionEta{},
a problem that is essentially the solution-size parameterization of \pwEtaDeletion{}.

Kernelization enthusiasts might know similar notions to the one we will exploit, called ``simplicial components'' and ``simplicial vertices''.
A simplicial component is a connected component of $G - M$ whose neighborhood is a clique, and a simplicial vertex is a vertex whose neighborhood is a clique.
These notions have been exploited in many papers, for example~\cite{bodlaenderKernelBoundsStructural2012,giannopoulouUniformKernelizationComplexity2017,bodlaenderPreprocessingTreewidthCombinatorial2013,bodlaenderPreprocessingRulesTriangulation2005,bodlaenderLineartimeAlgorithmFinding1996,clautiauxNewLowerUpper2003,cyganHardnessLosingWidth2014}, and the way we make sure we obtain connected components that are essentially simplicial follows the ideas of previous work on \fDeletion{} problems~\cite{fominPlanarFdeletionApproximation2012,giannopoulouUniformKernelizationComplexity2017,cyganHardnessLosingWidth2014}.
More concretely, our approach can be seen as a generalization of a reduction rule used to remove simplicial components by Bodlaender et al.~\cite[Rule 6]{bodlaenderKernelBoundsStructural2012} that is adjusted to work for \pwEtaDeletion{}.

The goal of our reduction rule is as follows:
we find a suitable connected component $C^\star$ of $G-M$ and remove $C^\star$ from the input graph to obtain the graph $G' = G - C^\star$ of our output instance.
For the correctness, we have to make sure that removing $C^\star$ does not affect
whether there is a solution $S$,
that is, a set of at most $k$ vertices whose removal results in a graph of pathwidth at most $\eta$.
More concretely, if there is a set $S$ of size at most $k$ such that $G' - S$ has pathwidth at most $\eta$, then we need to ensure that also $G - S$ has pathwidth at most $\eta$.

On a high level, we make sure that there exists a connected component $C$ of $G-M$ not intersected by the solution $S$ that is `similar' to $C^\star$.
The basic idea is that, if a path decomposition of $G'-S$ could accommodate $C$, then it should also accommodate
the similar component $C^\star$.
What we mean by similarity here is that (1) $C$ should have a pathwidth at least
$\pw(C^\star)$, and (2) given a path decomposition $\PP$ of $G'-S$, we would like
for $C$ to appear in bags that \emph{all} contain the entire neighborhood of $C^\star$.
In particular, the bags containing $C$ in $\PP$ serve as a `blueprint' that lets
us add similar bags that represent~$C^\star$ without increasing the width.
In that case, we obtain a path decomposition of $G-S$, as desired.

In order to enforce the neighborhood of $C^\star$ to appear in one bag,
the basic idea is that a large number of disjoint paths between two vertices $x, y$
forces them to appear together in a bag, \emph{virtually} simulating an edge $xy$.
If those disjoint paths intersect sufficiently many components,
then a simple property of path decompositions ensures that one of them (our `blueprint' $C$)
appears only in bags that contain both $x$ and $y$.

We find the components $C^\star$ and $C$ with a simple marking scheme (\cref{rule:reduce_ccs}) which we sketch now:
For each $X \in \binom{M}{2} \cup \binom{M}{1}$, we consider the connected components $C'$ of $G - M$ with $X \subseteq N_G(C')$ and mark the $2|M|(\eta + 2) + 1$ such components having the largest pathwidth.
If the number of connected components of $G - M$ is larger than $|M|^2 \cdot (2|M|(\eta + 2) + 1)$, at least one component is not marked; this component is the sought-after component $C^\star$.

Let us now explain why this marking scheme allows us to find a suitable component~$C$.
For this purpose, fix any $X = \{m_1,m_2\} \subseteq N_G(C^\star) \setminus S$.\footnote{Possibly all of $N_G(C^\star)$ is selected by $S$, but we ignore this trivial case here.}
As $X$ did not mark the component~$C^\star$, it marked $2|M|(\eta + 2) + 1$ components $C'$ of $G - M$ with $X \subseteq N_G(C')$.
Since $M$ is a solution, we may assume $|S| \leq |M|$, and therefore at least $2|M|(\eta + 1) + 1$ components marked by $X$ do not intersect $S$.
Thus, there are many disjoint $m_1$-$m_2$ paths in $G' - S$, forcing $m_1$ and $m_2$ to appear together in a bag.
From Helly's theorem for trees~\cite[Theorem 4.1]{hornThreeResultsTrees1972}, applied in the same way as a classic proof showing each clique of a graph appears in a bag of any path decomposition of it~\cite{bodlaenderPathwidthTreewidthCographs1993a}, we then obtain that some bag of $\PP$ contains all vertices of $N_G(C^\star)$.
So, the component $C^\star$ is essentially a simplicial component.

Now, let $m_\ell$ be the vertex of $N_G(C^\star) \setminus S$ that is introduced last, and $m_r$ be the vertex of $N_G(C^\star) \setminus S$ that is forgotten first in $\PP$.
Then, also $X = \{m_\ell,m_r\}$ marked $2|M|(\eta + 2) + 1$ components $C'$ of $G - M$ with $X \subseteq N_G(C')$; of which again $2|M|(\eta + 1) + 1$ do not intersect $S$.
For each vertex of $M$ there is some bag of $\PP$ in which it is introduced, or forgotten.
In particular, at most $2|M|$ bags forget or introduce a vertex of $M$, and these contain at most $2|M|(\eta + 1)$ vertices overall.
Therefore, one of the components $C$ marked by $X$ is such that each bag of $\PP$ that contains $C$ contains all vertices of $(N_G(C) \cup N_G(C^\star)) \setminus S$.
Since $\pw(C) \geq \pw(C^\star)$ due to our marking scheme, we can use (copies of) a bag $X_i$ that contains $\pw(C) + 1$ vertices of $C$ to embed $C^\star$ into the path decomposition.

Formally, we obtain the following result.

\begin{restatable}{theorem}{thmReduceCCs}
    \label{thm:reduce_ccs}
    Define $f_{\ref{thm:reduce_ccs}}(x) = x^2(2x(\eta +2)+1) \in O(x^3)$.
    \cref{rule:reduce_ccs}~\nameref{rule:reduce_ccs} runs in polynomial time,
    and given instance \((G,k,M)\) of \pwEtaDeletionEta{}, where
    $G-M$ has more than $f_{\ref{thm:reduce_ccs}}( |N_G(G-M)| )$ connected components,
    it outputs an instance \((G',k,M)\) of \pwEtaDeletionEta{} such that
    \begin{bracketenumerate}
        \item \((G,k,M)\) is a yes-instance if and only if \((G',k,M)\) is a yes-instance, and
        \item \(G'\) is a proper subgraph of \(G\).
    \end{bracketenumerate}
\end{restatable}

In \cref{thm:reduce_ccs}, we state the bound on the number of connected components as a function of the size of $N_G(G - M)$.
Note that clearly $|N_G(G - M)| \leq |M|$, but it can also be significantly smaller, which is actually exploited in a later stage of our kernelization routine (we omit the details in this overview).
\cref{thm:reduce_ccs} yields \cref{thm:uniform_kernel_vc_parameterization}.

\thmUniformKernelVcParameterization*

\subsection{Bounding the Size of Connected Components}
We proceed to the technically more challenging part of our work: bounding the size of connected components.
Here, we return to \pwEtaDelElimGeneral{}, where components of $G - M$ not only have pathwidth at most $\eta$, but additionally have elimination distance at most $\beta$ to the class of graphs with pathwidth one.

\subparagraph*{The high-level strategy.}
Our approach combines a known high-level strategy and highly non-trivial new insights.
Let us first discuss the known high-level strategy, which has been, at least partly, used in previous works such as \cite{fominPlanarFdeletionApproximation2012,giannopoulouUniformKernelizationComplexity2017,DBLP:conf/isaac/LochetS24,fominHittingForbiddenMinors2016,donkersPreprocessingOuterplanarVertex2022,kimLinearKernelsSingleexponential2016,DBLP:journals/jacm/BodlaenderFLPST16}.
In this strategy, the task of bounding the size of connected components reduces to the task of reducing the number of neighbors that vertices $m \in M$ have in ``nice'' connected components.

The eventual goal of this strategy is to obtain an instance $(G,k,M')$, where the modulator $M'$ has bounded size, such that each connected component of $G - M'$ is a \emph{protrusion}.
That is, each connected component of $G - M'$ is a graph with constant treewidth and a boundary of constant size to the rest of the graph.
Such graphs can be immediately replaced with graphs of constant size using the technique of protrusion replacement \cite{DBLP:journals/jacm/BodlaenderFLPST16,DBLP:journals/siamcomp/FominLST20,fominPlanarFdeletionApproximation2012}.

Towards this goal, it is useful to work with \emph{near-protrusions} (formally defined in \cref{def:near_protrusion}), which are not protrusions, but essentially become protrusions after removing a solution.
Concretely, we preprocess the input instance $(G,k,M)$ such that each component $C$ of $G - M$ has a neighborhood that can be decomposed into a set $X$ of size at most $2(\eta + 1)$ and a set $Y$, which may be large, but any optimal solution must delete all apart from at most $\eta + 1$ vertices of $Y$.
We skip the details of this standard step and refer to \cref{subsec:key_ingredients_modulator_degree_reduction} for the details.
One useful property of such a near-protrusion $C$ is that an optimal solution selects at most $3(\eta + 1)$ vertices of $C$.

Therefore, to bound the size of the components, it remains to turn these near-protrusions into actual protrusions.
Here, the known strategy again provides a recipe, so that it
actually suffices to bound the total degree of the vertices in $M$.
Let us explain why that is helpful.
Assume that the total degree of the vertices in $M$ is, say, $g(\eta, \beta) \cdot |M|^c$, for some constant $c$ independent of $\eta,\beta$ and some function $g$.
Then, let $X_1,\dots,X_t$ be a path decomposition of width $\eta$ of $G - M$.
We can create the modulator $M'$ by lifting vertices from $G-M$ as follows:
for every $v \in N_G(M)$, we lift a bag that contains $v$ into the modulator.
Then, each connected component of $G - M'$ has at most $2(\eta + 1)$ neighbors in $M'$, and treewidth (even pathwidth) at most $\eta$.
Moreover, the size of $M'$ is still bounded by a uniform polynomial in the parameter $|M|$.

To summarize, a uniform kernel immediately follows from bounding the total degree of the vertices in $M$, and we may even assume that each component of $G - M$ is a near-protrusion.

\subparagraph*{Bounding the degree of vertices in $M$.}
As we can already bound the number of connected components of $G - M$, bounding the total degree of the vertices in $M$ can be achieved by bounding, for each vertex $m \in M$ and each component $C$ of $G - M$,
the number of neighbors of $m$ in $C$.

In this part of our approach, we exploit that $C$ has elimination distance at most $\beta$ to the class of graphs of pathwidth at most $1$.
To state this more formally, we first define a few notions.
In a rooted forest $F$ containing node $n \in V(F)$, we write $\desc_F(n)$ for the set of \emph{descendants} of $n$ in $F$, including $n$ itself.
Now, we proceed to the second structural parameter we utilize, the $\pwEtaBoundClass[1]$-elimination distance.
This parameter is an instance of the general parameter elimination distance to some graph class $\HH$~\cite{bulianGraphIsomorphismParameterized2016,bulianFixedparameterTractableDistances2017}.
The definition we state is similar to the one of Bougeret et al.~\cite[Definition 8]{bougeretKernelizationDichotomiesHitting2025a} for a concrete instance of this parameter.

\begin{restatable}[Elimination Distance]{definition}{defElimDistance}
    \label{def:elimination_trees_forests_distance}
    Let $G$ be a graph, and $\HH$ a graph class.
    An $\HH$-\emph{elimination forest} of $G$ is a tuple $\FF = (F,\{Y_n\}_{n \in V(F)})$, where $F$ is a rooted forest, the sets $\{Y_n\}_{n \in V(F)}$ partition $V(G)$ and
    \begin{bracketenumerate}
        \item for each node $n \in V(F)$ that is not a leaf of $F$, we have $Y_n = \{n\}$,
        \item for each leaf $n \in V(F)$, we have that $G[Y_n] \in \HH$,
        and
        \item for each edge $v_1 v_2 \in E(G)$, where $v_1 \in Y_{n_1}$ and $v_2 \in Y_{n_2}$, we have that $n_1$ and $n_2$ are in an ancestor-descendant relationship in $F$.
    \end{bracketenumerate}
    The depth of $\FF$ is the depth of $F$.
    The $\HH$-\emph{elimination distance} $\elimDist[\HH]{G}$ of $G$ is the minimum depth of any $\HH$-elimination forest of $G$.
    For a node $n \in V(F)$, we write $V^{\FF}_n= \bigcup_{n' \in \desc_F(n)} Y_{n'}$.
    An $\HH$-elimination forest $\FF$ of $G$ is \emph{nice} if for each $n \in V(F)$ we have that $G[V^{\FF}_n]$ is a connected graph.
    The $\HH$-elimination forest $\FF$ is an $\HH$-\emph{elimination tree} if $F$ is a rooted tree.
\end{restatable}

When $\HH$ is the class of graphs on a single vertex, then a nice $\HH$-elimination forest corresponds to a nice treedepth decomposition as introduced by Reidl et al.~\cite[Definition 2]{reidlFasterParameterizedAlgorithm2014}.
We also define the related notion of $\FF$-neighbors.

\begin{restatable}[Neighbors in an Elimination Forest]{definition}{defNeighborsInElimForest}
    \label{def:neighbor_in_elim_forest}
    Let $G'$ be a subgraph of a graph $G$, $\HH$ a graph class, and
    $\mathcal{F} = (F,\{Y_n\}_{n \in V(F)})$ an $\HH$-elimination forest of $G'$.
    A vertex $v \in V(G) \setminus V(G')$ and a node $n \in V(F)$ are \emph{$\FF$-neighbors} if $N_G(v) \cap Y_n \neq \emptyset$.
We may simply call $v$ a neighbor of $n$ if $\FF$ is clear from the context.
\end{restatable}

In polynomial time, we can compute a nice $\pwEtaBoundClass[1]$-elimination forest $\FF = (F,\{Y_n\}_{n \in V(F)})$ of $C$ that has depth at most $\beta$ by observing that the class of graphs with elimination distance at most $\beta$ to the class $\pwEtaBoundClass[1]$ is again a minor-closed graph class (see \cite{bulianFixedparameterTractableDistances2017}) and by using a minor-testing algorithm~\cite{DBLP:journals/jct/RobertsonS95b,DBLP:conf/focs/KorhonenPS24}; or by directly employing an algorithm of Bulian et al.~\cite[Corollary 1]{bulianFixedparameterTractableDistances2017}, see also the proof of Theorem 1 in Bougeret et al.~\cite{bougeretKernelizationDichotomiesHitting2025a}.
Consider a vertex $m$ in the modulator $M$.
If $m$ has many neighbors in $C$, then
\begin{bracketenumerate}
    \item
    $m$ has many $\FF$-neighbors, or
    \item
    $m$ has many neighbors in a caterpillar of $C$ that corresponds to a leaf of $\FF$.
\end{bracketenumerate}
We obtain uniform polynomial bounds for these two cases in the following.

\subparagraph*{(1) Bounding the number of $\FF$-neighbors.}
Let $m \in M$.
We begin by elaborating on how we bound the number of $\FF$-neighbors of $m$.
Here, the idea is that we want to find some neighbor $v \in V(C)$ of $m$ such that we can safely delete the edge from $m$ to $v$.
In other words, we want to find an \emph{irrelevant edge}: an edge whose deletion does not change the answer to the instance.\footnote{This is in the same spirit as the notion of an \emph{irrelevant vertex}, which was introduced by Robertson and Seymour~\cite{DBLP:journals/jct/RobertsonS95b}. Irrelevant vertices and irrelevant edges have since been used in more works on parameterized algorithms and kernelization, for example~\cite{fominPlanarFdeletionApproximation2012,DBLP:journals/siamdm/HeggernesHLP13,DBLP:conf/isaac/LochetS24,DBLP:journals/jacm/KratschW20}.}
Next, we sketch how our algorithm finds this vertex $v$, then we explain why $v$ has useful properties.

Recall that, for any $n \in V(F)$, $V_n^\FF$ is the set of vertices that are introduced in the bags below and including node $n$.
Moreover, for a node $n \in  V(F)$, we call the set $N_G(V_n^\FF) \cap V(C)$ the \emph{ancestor-type} of $n$; observe that each node in the ancestor-type of $n$ is an ancestor of $n$ in $F$.

We begin with a marking scheme reminiscent of the one used to reduce the number of connected components.
For each $X = \{m_1,m_2\} \in \binom{M}{2} \cup \binom{M}{1}$ and each depth $d$ of the elimination forest, we mark a constant (depending on $\eta,\beta$) number of nodes $n \in V(F)$ at depth $d$, such that $X \subseteq N_G(V_{n}^\FF)$.
Since the depth of $F$ is bounded by the constant $\beta$, the overall number of marked nodes is uniformly bounded, and some $\FF$-neighbors of $m$ remain unmarked.

After applying the marking scheme, a basic graph-theoretic argument lets us find a parent node $p \in V(F)$ with $\alpha(\eta,\beta)$ (again a constant depending on $\eta,\beta$) number of children such that each child $c$ is not marked, but $m$ is a neighbor of $G[V_c^\FF]$.
By choosing $\alpha(\eta,\beta)$ large enough, we can even ensure that sufficiently many of these children have the same ancestor-type.
We then choose such a child $z$ such that the pathwidth of $G[V_z^\FF]$ is small (relative to the pathwidth of $G[V_c^\FF]$ of the other children $c$).
Let $v$ be a neighbor of $m$ in $V_z^\FF$.

Our reduction rule then proceeds to delete the edge $mv$, that is, the output instance is $(G' = G - mv, k, M)$.
The forward direction of safety of this rule is trivial.
For the backward direction, consider some minimum solution $S$ of the output instance and some width-$\eta$ path decomposition $\PP = X_1,\dots,X_t$ of $G' - S$.
If one of $m$, $v$ is selected, or if $m$ and $v$ appear in some bag together, then $\PP$ is also a path decomposition of $G - S$.
Therefore, we assume that this is not the case.
Our goal now is showing that the path decomposition $\PP$ can be modified without increasing the width so that $m$ and $v$ do appear in some bag together.

Consider the connected component $H$ of the graph $G[V_z^\FF \setminus S]$ that contains $v$.
Our algorithm has chosen $v$ so that there are \emph{two} graphs in $G' - S$ that are ``similar'' to $H$.
Let us elaborate in more detail on what we mean by similar in this case.

The neighborhood of $H$ can be split into the neighborhood in $M \setminus S$ and the neighborhood in $C - S$.
The first of the two graphs that is similar to $H$, call it $H_1$, is used to ensure that each vertex of $H$ meets its neighbors in $M \setminus S$, and the second of the two graphs, call it $H_2$, is used to ensure that each vertex of $H$ meets its neighbors in $C - S$.
More concretely, we guarantee that
\begin{enumerate}
    \item no vertex of $H_1$ is deleted by $S$,
    \item each bag that contains a vertex of $H_1$ contains $N_G(H_1) \setminus S$,
    \item each bag that contains a vertex of $H_1$ contains $N_{G-S}(H) \cap M$ (which includes $m$), and
    \item $\pw(H_1) \geq \pw(H)$.
\end{enumerate}
The existence of graph $H_1$ follows from our marking scheme; $H_1 = G[V_{\hat z}^\FF]$ for a marked node~$\hat z$.

The second graph $H_2$ is going to be $G[V_{z'}^\FF]$ for a node $z'$ that is a sibling of node $z$ in the elimination forest $\FF$.
We ensure that
\begin{enumerate}
    \item no vertex of $H_2$ is deleted by $S$,
    \item each bag that contains a vertex of $H_2$ contains $N_G(H_2) \setminus S$,
    \item $z'$ has the same ancestor-type as $z$,
    \item $\pw(H_2) \geq \pw(H)$, and
    \item $m$ is a neighbor of $H_2$.
\end{enumerate}
The existence of $H_2$ follows from our choice of $z$ as one of many siblings with similar properties.

Having these graphs $H_1$, $H_2$ at hand, we eventually find a bag $X_i$ that we can copy to embed (a part of) $H$.
Since the details are rather involved, we do not present them here.
Overall, we obtain \cref{thm:treedepth_case}.
We strengthen the result by actually allowing $\FF$ to be an $\HH$-elimination tree of $C$ for any graph class $\HH$.
The bound on the number of $\FF$-neighbors is expressed in terms of a function $\degreeBoundTreedepthCase(x) \in O(x^2)$, where the $O$-notation hides constant factors depending on the constants $\eta, \beta$.

\begin{restatable}{lemma}{thmTreedepthCase}
    \label{thm:treedepth_case}
    Let $(G,k,M)$ be an instance of \pwEtaDeletionEta{}, and $C \subseteq G- M$ be a near-protrusion of $(G,k)$.
    Let $\mathcal{F}$ be a nice $\HH$-elimination tree of $C$ of depth at most $\beta$ for some arbitrary graph class $\HH$, and assume that some vertex $m \in N_G(C)$ has more than $\degreeBoundTreedepthCase(|N_G(C)|)$ $\FF$-neighbors.

    \cref{rule:degree_reduction_elimination_forest} \nameref{rule:degree_reduction_elimination_forest}  runs in polynomial time, and given $(G,k,M),C,\FF$ as input, the algorithm outputs an instance $(G',k,M)$ of \pwEtaDeletionEta{} equivalent to $(G,k,M)$ where $G'$ is a proper subgraph of $G$.
\end{restatable}

\subparagraph*{(2) Bounding the number of neighbors in caterpillars.}
Let $m \in M$, and let $C$ be a caterpillar graph corresponding to a leaf in the elimination forest $\FF$.
Our task is to bound the number of neighbors that $m$ has in $C$.
We provide the following much stronger result that can immediately bound the size of caterpillars in terms of their neighborhood size.
Here, we again express the obtained bounds in terms of a function $f_{\ref{thm:caterpillar_case}}(x) \in O(x^{8})$, where the $O$-notation hides constant factors depending on the constant $\eta$.

\begin{restatable}{lemma}{thmCaterpillarCase}
    \label{thm:caterpillar_case}
    Define $f_{\ref{thm:caterpillar_case}}(x) = (8(\eta+1)(x+5))^{8}$.
    Let $(G,k,M)$ be an instance of \pwEtaDeletionEta, $C$ be an induced subgraph of $G-M$ that is a caterpillar
    satisfying $|V(C)| > \BoundCaterpillar$.

    \cref{rule:degree_reduction_caterpillar}~\nameref{rule:degree_reduction_caterpillar}  runs in polynomial time and,
    given $(G,k,M)$, $C$ as input, it outputs an equivalent instance $(G',k,M)$ where
    $G'$ is a proper minor of $G$.
\end{restatable}

Therefore, for each $m \in M$, we bound the number of neighbors that $m$ has in caterpillars corresponding to leaves of $\FF$, as well as the number of $\FF$-neighbors of $m$.
As a result, we bound the number of neighbors each $m \in M$ has in $G - M$, which bounds the total degree of $M$.

This concludes the construction of our kernel since now we have all the properties to apply the high-level strategy mentioned earlier.
In our actual kernel, we do not rely on the big hammer of protrusion replacers in the following sense: we create our own explicit approach to reduce the size of protrusions.
This allows us to obtain a uniform polynomial kernel with exact knowledge of the constants involved in the kernel size. 
\section{Preliminaries}
\label{sec:prelims}
\subsection{Basic Notation}

We use $\nat$ to denote the set of natural numbers, and $0 \in \nat$.
For integers $c_1,c_2 \in \mathbb{Z}$, we write $\range[c_1]{c_2} = \{x \in \mathbb{Z} \mid c_1 \leq x \leq c_2\}$ and use $[c_2]$ as a shorthand for $[1, c_2]$;
we call $[c_1, c_2]$ an \textit{interval}.
For a set $X$ and $c \in \nat$ we write $\binom{X}{c}$ for the set of all subsets of $X$ of size exactly $c$.
We write $\binom{X}{\leq c}$ for the set of all subsets of $X$ of size at most $c$.
For integers $c_1,c_2 \in \nat$ we write $\binom{c_1}{\leq c_2}$ for the number $\big| \binom{\range{c_1}}{\leq c_2} \big|$.

\subsection{Graphs}

For graphs $G,H$ we write $H \subseteq G$ if $H$ is a subgraph of $G$.
A subgraph $H \subseteq G$ is \emph{proper} if $H \neq G$.
Usually we assume that graphs have at least one vertex, but in some cases, for example when preprocessing trivial problem instances, it is convenient to also allow the empty graph to be output.
Given a graph $G$ and a set $S \subseteq V(G)$ we denote $N_G(S) = \bigcup_{v \in S} N_G(v) \setminus S$.
When $H \subseteq G$ we set $N_G(H) = N_G(V(H))$.
For graphs $G,H$ the graph $G \cup H$ is the graph with vertex set $V(G) \cup V(H)$ and edge set $E(G) \cup E(H)$.
For graph $G$ and set $X \subseteq V(G)$ we define $G - X = G[V(G) \setminus X]$.
We use $\Delta(G)$ to denote the maximum degree of $G$.

Similarly, for $G$ and $H \subseteq G$ we set $G - H = G - V(H)$.

Given a graph $G$ and edge $uv \in E(G)$, the operation of \emph{contracting the edge $uv$} results in the graph $G'$ obtained from $G$ by adding a fresh vertex $v'$ with neighborhood $N_G(v) \cup N_G(u)$,
and then deleting the vertices $u$ and $v$.
A graph $H$ is a \emph{minor} of a graph $G$ if (a graph isomorphic to) $H$ can be obtained from $G$ by (repeatedly) deleting vertices, edges, and contracting edges.
A minor $H$ of $G$ is a \emph{proper minor} if $H \neq G$.

\subsection{Path Decompositions}

We continue by defining the notion of a path decomposition (see~\cite[Chapter 7]{cyganParameterizedAlgorithms2015} for a textbook introduction to this topic).

\begin{definition}[Path Decomposition and Pathwidth]
    \label{def:path_decomposition}
    A \emph{path decomposition} of a graph $G$ is a sequence of sets $\mathcal{P} = X_1,\dots,X_t$,  such that
    \begin{bracketenumerate}
        \item for each edge $uv \in E(G)$, there exists an index $i$ such that $\{u,v\} \subseteq X_i$, and
        \item for each $v \in V(G)$, there is a nonempty interval $\range[i]{j} \subseteq \range[1]{t}$ such that $v \in X_r$ if and only if $r \in \range[i]{j}$.
    \end{bracketenumerate}
    The width of decomposition $\mathcal{P}$ is $\max_{i \in [t]} |X_i| - 1$.
    The \emph{pathwidth} of $G$, denoted by $\pw(G)$, is the minimum width of any path decomposition of $G$.

    The sets $X_i$ ($i \in \range[1]{t}$) are usually referred to as \emph{bags}.
    We say that vertex $v \in V(G)$ is \emph{introduced} in bag $X_i$ where $i$ is the smallest index of a bag that contains $v$.
    Similarly, we say that $v$ is \emph{forgotten} in bag $X_i$ if $v \in X_{i-1}$ but not $v \in X_i$.

    Path decomposition $\mathcal{P}$ is \emph{nice} if each bag $X_i$ with $i > 1$ either fulfills $X_i \cup \{v\} = X_{i-1}$ for some $v \notin X_{i}$, in which case $X_i$ is a \emph{forget bag}, or $X_i \setminus \{v\} = X_{i-1}$ for some $v \in X_i$, in which case $X_i$ is an \emph{introduce bag}.
\end{definition}

It is a standard fact that when $H$ is a minor of $G$, then we have $\pw(H) \leq \pw(G)$, and if a graph $G$ is the disjoint union of graphs $H_1$ and $H_2$ we have $\pw(G) = \max\{\pw(H_1),\pw(H_2)\}$.
We will implicitly make use of these facts in the paper.
We also require the following standard
notion of a graph separation, see, for example, Cygan et al.~\cite[Section 7.6.1]{cyganParameterizedAlgorithms2015}.

\begin{definition}[Separation and Separators]
    \label{def:separation_separators}
    Let $G$ be a connected graph and $V_1, V_2 \subseteq V(G)$.
    We say that $(V_1, V_2)$ is a
    separation of $G$ of order $|V_1 \cap V_2|$ if $V_1 \cup V_2 = V(G)$ and there are no edges in $G$ between
    $V_1 \setminus V_2$ and $V_2 \setminus V_1$.
    The set $V_1 \cap V_2$ is the associated \emph{separator}.
\end{definition}

One of the key properties of path decompositions is that the intersection of subsequent bags, called \emph{adhesions}, give rise to separations.
See \cite[Lemma 7.1]{cyganParameterizedAlgorithms2015} for the following lemma.

\begin{lemma}[{\cite[Lemma 7.1]{cyganParameterizedAlgorithms2015}}]
    \label{thm:adhesions_are_separators}
    Let $G$ be a graph with path decomposition $X_1,\dots,X_t$.
    For any $i \in [t-1]$ let $L_i = \bigcup_{j \in [i]} X_j$, and $R_i = \bigcup_{j \in [i+1,t]} X_j$.
    Then, $(L_i,R_i)$ is a separation of $G$ with separator $X_i \cap X_{i+1}$.
\end{lemma}

In our proofs, we use the following convenient notation that allows us to concisely obtain certain indices of bags of a path decomposition.

\begin{definition}[Set Indices]
    Let $G$ be a graph with path decomposition $\mathcal{P} = X_1,\dots,X_t$, and $Y$ be a set.
    Set $\II_\forall = \{i \in \range[1]{t} \mid Y \subseteq X_i\}$, and $\II_\exists = \{i \in \range[1]{t} \mid Y \cap X_i \neq \emptyset\}$.

    For $Q \in \{\exists,\forall\}$ we set $\ell_Q^{\mathcal{P}}(Y) = \min \II_Q$ and $r_Q^{\mathcal{P}}(Y) = \max \II_Q$.
    We note that these values might be undefined if $\II_Q = \emptyset$.
\end{definition}

Note that $\leftBagAll^{\mathcal{P}}(Y)$ is the lowest index of a bag that contains all vertices of $Y$, similarly $\leftBagAny^{\mathcal{P}}(Y)$ is the lowest index of a bag such that there exists a vertex of $Y$ in the bag.
It is well-known and easy to check that path decompositions fulfill the following property, we provide a proof purely for completeness.

\begin{observation}
    \label{thm:connected_subgraphs_have_interval}
    Let $G$ be a graph with path decomposition $\mathcal{P} = X_1,\dots,X_t$, and $H \subseteq G$ be a connected graph.
    Then, there exist $i,j \in \range[1]{t}$ with $i \leq j$ such that, for any bag $X_r$ of the path decomposition we have $X_r \cap V(H) \neq \emptyset$ if and only if $r \in \range[i]{j}$.
    Therefore, for each $i \in \range[\leftBagAny^\mathcal{P}(V(H))]{\rightBagAny^\mathcal{P}(V(H))}$
    we have $X_i \cap V(H) \neq \emptyset$.
\end{observation}
\begin{proof}
    Towards a contradiction, assume that the observation does not hold.
    Hence, for connected $H$, there are three distinct indices $x < y < z$, such that $X_x \cap V(H) \neq \emptyset$, $X_y \cap V(H) = \emptyset$, and $X_z \cap V(H) \neq \emptyset$.
    Let $v$ be a vertex in $X_x \cap V(H)$ and $w$ be a vertex in $X_z \cap V(H)$, by the properties of a path decomposition and the fact that $X_y$ contains no vertex of $H$ we have $v \neq w$.
    By \cref{thm:adhesions_are_separators}, we have that $X_y \cap X_{y+1}$ is a separator of $G$.
    This means that each path from $v$ to $w$ must contain a vertex of $X_y \cap X_{y+1}$, and therefore a vertex of $X_y$.
    However, since $V(H) \cap X_y = \emptyset$, this implies that each path from $v$ to $w$ must use a vertex outside $H$.
    This contradicts that $H$ is connected.
\end{proof}

Similarly, a basic property is that every clique is a subset of some bag~\cite[Lemma 3.1]{bodlaenderPathwidthTreewidthCographs1993a}.
This directly follows from the well-known Helly's theorem for trees, which we quickly recap here.

\begin{theorem}[Helly's Theorem for Trees {\cite[Theorem 4.1]{hornThreeResultsTrees1972}}]
    \label{thm:helly_property_of_trees}
    Let $T$ be a tree, and $\mathcal{S} \subseteq 2^{V(T)}$ such that for each $S \in \mathcal{S}$ the graph $T[S]$ is connected, and such that for any $S,S' \in \mathcal{S}$ we have $S \cap S' \neq \emptyset$.
    Then, $\bigcap_{S \in \mathcal{S}} S \neq \emptyset$.
\end{theorem}

Using this theorem in the same way as in the proof of \cite[Lemma 3.1]{bodlaenderPathwidthTreewidthCographs1993a}, we prove the following simple and well-known observation (see, e.g.,~\cite[Lemma 4.1]{bodlaenderLineartimeAlgorithmFinding1996}), we present a proof purely for completeness.
\begin{observation}
    \label{obs:disjoint_paths_yield_clique}
    Let $G$ be a graph with path decomposition $X_1,\dots,X_t$ of width at most $\eta$, and $Y \subseteq V(G)$ be a set of vertices, such that for any distinct $u,v \in Y$ there are $\eta + 2$ internally vertex-disjoint $u$-$v$ paths in $G$.
    Then, there exists a bag $X_i$ with $Y \subseteq X_i$.
\end{observation}
\begin{proof}
    We show that for any distinct $u,v \in Y$, there must be a bag $X_j$ that contains $u$ and $v$.
    As $u,v$ have $\eta + 2$ internally vertex-disjoint paths between them, there is no set of size at most $\eta + 1$ that separates $u$ and $v$.
    If $u$ and $v$ did not appear in a bag together, then due to \cref{thm:adhesions_are_separators}, some adhesion would need to separate $u$ and $v$, which is not possible since each bag, and therefore also each adhesion, contains at most $\eta + 1$ vertices.

    Now, for each vertex $u \in Y$ let $I_u$ be the indices of the bags that contain $u$.
    Applying \cref{thm:helly_property_of_trees} on the path given by the path decomposition and these sets $I_u$ yields the existence of bag containing all of $Y$.
\end{proof}

We remark that path decompositions can be computed quickly using a famous algorithm by Bodlaender~\cite{bodlaenderLineartimeAlgorithmFinding1996}.
\begin{theorem}[{{\cite[Theorem 7.2]{bodlaenderLineartimeAlgorithmFinding1996}}}]
    \label{thm:path_decomposition_algorithm}
    Let $\eta$ be a constant.
    There is a linear-time algorithm that, given a graph $G$ with pathwidth at most $\eta$ as input, computes a path decomposition of $G$ of minimum width.
\end{theorem}

We also make use of the following basic observation, letting us assume that
a solution to an instance of \pwEtaDeletionDistToG does not disturb the structure
of low-pathwidth components too much.

\begin{observation}\label{obs:mostly_intact}
    Let $G$ be a graph, $\eta \geq 0$ an integer, $H$ an induced subgraph of $G$ with $\pw(H) \leq \eta$ and $N_G(H) \neq \emptyset$,
    and $S \subseteq V(G)$ be such that $\pw(G-S) \leq \eta$.
    Then there exists $S' \subseteq V(G)$ such that $\pw(G-S') \leq \eta$, where $|S'| \leq |S|$ and $|S' \cap V(H)| < |N_G(H)|$.
\end{observation}
\begin{proof}
    If $|S \cap V(H)| < |N_G(H)|$ then $S$ is itself the desired set $S'$.

    Otherwise, $|S \cap V(H)| \geq |N_G(H)|$.
    Let $S' \coloneq (S \setminus V(H)) \cup N_G(H)$.
    Then, $|S'| \leq |S| - |S \cap V(H)| + |N_G(H)| \leq |S|$.
    Since $N_{G-S'}(H) = \emptyset$ and $S' \cap V(H) = \emptyset$, we have that $G-S' = H \uplus (G-S'-H)$.
    Since $\pw(G - S) \leq \eta$ and $S \setminus V(H) \subseteq S' \setminus V(H)$, we then have $\pw(G - S' - H) \leq \pw(G - S - H) \leq \eta$.
    Thus, $\pw(G-S') \leq \max\{\pw(H), \pw(G-S'-H)\} \leq \eta$, which completes the proof.
\end{proof}

\subsection{Treewidth}

We also define the standard notions of treewidth and tree decompositions~\cite[Section 7]{cyganParameterizedAlgorithms2015}.

\begin{definition}[Tree Decomposition and Treewidth]
    \label{def:treewidth}
    Let $G$ be a graph.
    A tree decomposition of $G$ is a pair $\mathcal{T} = (T,\{B_n\}_{n \in V(T)})$ where $T$ is a tree, and for each node $n \in V(T)$ we have that $B_n$, the bag of $n$, is a subset of $V(G)$.
    It fulfills the following two properties:
    \begin{bracketenumerate}
        \item for each edge $uv \in E(G)$ there is a node $n \in V(T)$ such that $\{u,v\} \subseteq B_n$, and
        \item for each $v \in V(G)$ the set $\{n \in V(T) \mid v \in B_n\}$ induces a nonempty subtree of $T$.
    \end{bracketenumerate}
    The width of $\mathcal{T}$ is $\max_{n \in V(T)} |B_n| - 1$.
    The treewidth $\tw(G)$ of a graph is the minimum  width of any tree decomposition of $G$.
\end{definition}

\subsection{Elimination Distance}

Now, we proceed to the formal definition of the $\pwEtaBoundClass[1]$-elimination distance.
Before giving the formal definition of the parameter, we set some terminology for rooted trees and forests.

\begin{definition}[Basic Notions for Rooted Trees and Forests]
    \label{def:depths_ancestors_descendants}
    Let $T$ be a rooted tree with root $r$.
    The \emph{depth} of a vertex $v \in V(T)$ is the number of edges on the $r$-$v$ path in $T$ and is denoted by $\depth_T(v)$.
    The \emph{depth} of $T$ is $\depth(T) = \max_{v \in V(T)} \depth_T(v)$.
    For a vertex $v \in V(T)$, we let $T_v$ be the subtree of $T$ rooted at $v$.
    A forest $F$ is a \emph{rooted forest} if each connected component of $F$ is a rooted tree.
    The \emph{depth} of $F$ is the maximum depth of any connected component of $F$.

    For a vertex $v \in V(T)$, we write $\anc_T(v)$ for the set of \emph{ancestors} of $v$ in rooted tree $T$, including vertex $v$ itself.
    Similarly, we write $\desc_T(v) = V(T_v)$ for the set of \emph{descendants} of $v$ in $T$, including $v$ itself.
    Two vertices $v,u \in V(T)$ are \emph{siblings} if they have the same parent in $T$.
    Vertices $v,u \in V(T)$ are in an \emph{ancestor-descendant relationship} if either $v \in \anc_T(u)$ or $u \in \anc_T(v)$.

    We naturally extend these notions to forests $F$.
    The respective notion for a vertex $v \in V(F)$ is defined as for the tree $T$ containing $v$.
\end{definition}
Now we recall the definition of the $\HH$-elimination distance \cite{bulianGraphIsomorphismParameterized2016,bulianFixedparameterTractableDistances2017}.
Let us remark again that the definition we state is similar to the one of Bougeret et al.~\cite[Definition 8]{bougeretKernelizationDichotomiesHitting2025a} for a concrete instance of this parameter.

\defElimDistance*{}

We remark again that when $\HH$ is the class of graphs on a single vertex, then a nice $\HH$-elimination forest corresponds to a nice treedepth decomposition as introduced by Reidl et al.~\cite[Definition 2]{reidlFasterParameterizedAlgorithm2014}.

Using this notion, we can define the treedepth of a graph, see e.g.,~\cite[Section 6.1]{nesetrilSparsityGraphsStructures2012}.
\begin{definition}[Treedepth]
    \label{def:treedepth}
    Let $\HH$ be the class of graphs having a single vertex.
    Then, the \emph{treedepth} of a graph $G$ is $\elimDist[\HH]{G} + 1$.
\end{definition}

Recall the graph classes we defined based on these parameters.
\defPWTwGraphClasses*{}

Next, we define what it means when we say that a vertex of $G$ has neighbors in a $\HH$-elimination forest of $G$.
\defNeighborsInElimForest*{}

Now, we establish the useful notion of ancestor-types.

\begin{definition}[Ancestor-Type]
    \label{def:ancestor_type}
    Let $G$ be a connected graph and $\FF = (F,\{Y_n\}_{n \in V(F)})$ a $\HH$-elimination tree of $G$ for some graph class $\HH$.
    The ancestor-type of node $n \in V(F)$ (in $\FF$, $G$) is the set $\anc_F(n) \cap N_G(V_n^{\FF})$.
\end{definition}

Note that $n \in \anc_F(n)$, but $n$ is not in the ancestor-type of $n$.
We observe that siblings can have only a few ancestor-types.
\begin{observation}
    \label{obs:bounded_ancestor_type_number}
    Let $G$ be a connected graph and $\FF = (F,\{Y_n\}_{n \in V(F)})$ be an $\HH$-elimination tree of $G$ for some graph class $\HH$.
    Let $n_1,n_2,\dots,n_r$ be sibling nodes in $F$.
    Then, $|\{\anc_F(n_i) \cap N_G(V_{n_i}^\FF) \mid  i \in \range[1]{r}\}| \leq 2^{\depth(F)}$.
    Therefore, there are at most $2^{\depth(F)}$ possible ancestor-types among the siblings $n_1,n_2,\dots,n_r$.
\end{observation}

This type of grouping of siblings in an elimination tree according to their neighborhood among their ancestors is one of the key structural properties that elimination forests enable, and has been, for example, also exploited in the uniform kernel for \tdEtaDeletion{}~\cite{giannopoulouUniformKernelizationComplexity2017}.

Note that for any constant $\beta$ the class of graphs $G$ having $\pwElimDist{G} \leq \beta$ is again minor-closed~\cite{bulianFixedparameterTractableDistances2017},
and therefore one can use a minor-testing algorithm~\cite{DBLP:journals/jct/RobertsonS95b,DBLP:conf/focs/KorhonenPS24} to check in polynomial-time whether a given graph $G$ has $\pwElimDist{G} \leq \beta$.
One can also directly employ an algorithm of Bulian et al.~\cite[Corollary 1]{bulianFixedparameterTractableDistances2017} to achieve this.
Using this fact we can compute $\pwEtaBoundClass[1]$-elimination trees of constant depth in polynomial-time by ``guessing'' the root of the elimination tree, and then recursing on the connected components of the graph after deleting the root.
Moreover, it is not difficult to see that this will eventually even result in a nice elimination tree.
See also the proof of Theorem 1 in Bougeret et al.~\cite{bougeretKernelizationDichotomiesHitting2025a}.
Therefore, we have the following theorem.

\begin{theorem}[{Follows from~\cite[Corollary 1]{bulianFixedparameterTractableDistances2017}}]
    \label{thm:compute_elimination_forests}
    Let $\beta$ be a constant.
    There exists an algorithm that, given a connected graph $G$ with $\pwElimDist{G} \leq \beta$, outputs a nice $\pwEtaBoundClass[1]$-elimination tree of $G$ in polynomial time.
\end{theorem}

Note that in all settings we may drop subscripts and/or superscripts if it is clear from the context which object must be present there. 
\section{Bounding the Number of Connected Components}
\label{sec:kernel:bounding_ccs}
We now proceed to the uniform kernel for \pwEtaDelElimGeneral{}.
In this section, we focus on a preprocessing algorithm that can reduce the number of connected components of \(G - M\).
The used reduction rule works for the problem \pwEtaDelElimGeneral{}, but, in fact, it does not use the assumption that the connected components of \(G - M\) have bounded $\pwEtaBoundClass[1]$-elimination distance.
Therefore, it works even for the more general problem \pwEtaDeletionEta{}, defined as follows.

\defparproblem{\pwEtaDeletionEta}{Graph \(G\), integer \(k \geq 0\), set \(M \subseteq V(G)\) such that \(\pw(G- M) \leq \eta\)}{\(|M|\)}{Is there a set \(S \subseteq V(G)\) such that \(\pw(G - S) \leq \eta\) and \(|S| \leq k\)?}

Note that \pwEtaDeletionEta{} is essentially equivalent to the solution-size parameterization of \pwEtaDeletion{}, since \pwEtaDeletion{} admits a polynomial-time deterministic constant-factor approximation algorithm~\cite[Corollary 1.2]{guptaLosingTreewidthSeparating2019}.
However, the stated definition is more convenient for us since it fits together more seamlessly with the definition of \pwEtaDelElimGeneral{}: indeed, any instance of \pwEtaDelElimGeneral{} is also an instance of \pwEtaDeletionEta{}.

As remarked earlier, the reduction rule can be seen as a generalization of \cite[Rule 6]{bodlaenderKernelBoundsStructural2012}, which is used to reduce the number of simplicial components of known pathwidth for structural parameterizations of the \textsc{Pathwidth} problem.
Of course, adjustments are necessary to obtain a rule that works for \pwEtaDeletionEta{}.
The reduction rule finds a connected component that is essentially a simplicial component, and the way this is done follows ideas of previous work on \fDeletion{} problems~\cite{fominPlanarFdeletionApproximation2012,giannopoulouUniformKernelizationComplexity2017,cyganHardnessLosingWidth2014}.

From now on, we treat \(\eta \geq 1, \beta \geq 0\) as \emph{arbitrary but fixed} constants.
Therefore, we may hide any constant factors depending on $\eta$, $\beta$ in the $O$-notation, e.g., we have $g(\eta,\beta) \cdot O(|M|^2) \subseteq O(|M|^2)$ for any function $g$.
Since some of our statements also allow $\eta = 0$, when this happens we explain how this case is handled at the respective occurrence.

\subsection{The Main Ingredients}
Before providing the actual reduction rule, we will introduce some further ingredients of our algorithm and proof.
We begin by introducing the notion of \emph{stability} in path decompositions.

\begin{restatable}[Stability]{definition}{defStability}
    \label{def:stability}
    Let \(G\) be a graph with path decomposition \(\mathcal{P} = X_1,\dots,X_t\), and \(Y,Z \subseteq V(G)\).
    Then \(Y\) is \emph{\(Z\)-stable} (in \(\mathcal{P}\))
    if there exists a set \(Z' \subseteq Z\) such that for each bag \(X_i\) with \(X_i \cap Y \neq \emptyset\), we have that \(X_i \cap Z = Z'\).
    For \(H \subseteq G\), we say that \(H\) is \emph{\(Z\)-stable} (in \(\mathcal{P}\)) if \(V(H)\) is \(Z\)-stable.

    Set \(Y\) is \emph{stable (in \(G\), \(\PP\))} if \(Y\) is \(N_G(Y)\)-stable in \(\PP\).
    Similarly, a graph $H \subseteq G$ is \emph{stable (in \(G\), \(\PP\))} if $H$ is $N_G(H)$-stable in $\PP$.
\end{restatable}

Stability tells us that every bag that contains a vertex of \(Y\) contains exactly the same subset of \(Z\).
This notion is useful when a graph $H \subseteq G$ is stable, as the next lemma shows.

\begin{lemma}
    \label{thm:stable_graph_bags_contain_neighbors}
    Let \(G\) be a graph with path decomposition \(\mathcal{P} = X_1,\dots,X_t\), and $Y \subseteq V(G)$ be a set that is stable in $G, \PP$.
    Then, each bag $X_i$ with $X_i \cap Y \neq \emptyset$ contains all vertices of $N_G(Y)$.
    Similarly, if $H \subseteq G$ is a stable graph, then each bag $X_i$ with $X_i \cap V(H) \neq \emptyset$ contains all vertices of $N_G(H)$.
\end{lemma}
\begin{proof}
    For each $w \in N_G(Y)$ there must be some bag $X_i$ such that $X_i$ contains $w$ and a vertex of $Y$ since $G$ contains an edge from a vertex of $Y$ to $w$.
    Then, since $Y$ is stable, actually all bags $X_i$ that contain a vertex of $Y$ must contain $w$.
    The proof for graphs $H$ is obtained by applying the statement for vertex sets on $V(H)$.
\end{proof}

Crucially, it is not too difficult to find stable subgraphs, as we show in \cref{thm:find_stable_graph}.

\begin{lemma}
    \label{thm:find_stable_graph}
    Define \(f_{\ref{thm:find_stable_graph}}(y) = 2y \cdot (\eta + 1) + 1\).
    Let \(G\) be a graph with path decomposition \(X_1,\dots,X_t\) of width at most \(\eta\).
    Moreover, let \(H_1,\dots,H_r\) be pairwise vertex-disjoint connected subgraphs of \(G\), and \(Y \subseteq V(G)\).
    If \(r \geq f_{\ref{thm:find_stable_graph}}(|Y|)\), then at least one graph \(H_i\) is \(Y\)-stable.
\end{lemma}
\begin{proof}
    At most \(2|Y|\) bags of \(X_1,\dots,X_t\) are such that a vertex of \(Y\) is introduced or forgotten in them, and each of these contains at most \((\eta + 1)\) vertices.
    Therefore, there exists some \(H_i\) such that all bags that contain a vertex of \(H_i\) are not bags that introduce, or forget a vertex of \(Y\).
    Since the indices of bags that contain a vertex of \(H_i\) form an interval (\cref{thm:connected_subgraphs_have_interval}), this graph \(H_i\) is \(Y\)-stable.
\end{proof}

In a similar vein, we also obtain the following lemma, used later to find a
suitable non-full bag in a given path decomposition to insert additional vertices.

\begin{lemma}\label{lem:pd_hierarchy}
    Let $G$ be a graph with path decomposition $\PP$ of width at most $\eta$,
    and let $h \geq 2 \eta + 2$ be an integer.
    Further, let $H_1, \dots, H_{h}$ be vertex-disjoint connected subgraphs of $G$
    and $x, y \in V(G)$ be vertices such that $N_G(x) \cap V(H_i) \neq \emptyset$ and $N_G(y) \cap V(H_i) \neq \emptyset$ for every $i \in [h]$.
    Then there exists $i \in [h]$ such that $\leftBagAll^\PP(\{x, y\}) < \leftBagAny^\PP(H_i) \leq \rightBagAny^\PP(H_i) < \rightBagAll^\PP(\{x, y\})$.
\end{lemma}
\begin{proof}
    If $x = y$, or if there is an edge $xy$ in $G$, then $\leftBagAll^\PP(\{x, y\})$ and $\rightBagAll^\PP(\{x, y\})$ are well-defined by the definition of path decompositions.
    Otherwise, we note that there are $h \geq \eta+2$ internally vertex-disjoint $x$-$y$ paths.
    This is because, for each $i\in[h]$, the graph $G[V(H_i)\cup \{x,y\}]$ is connected and therefore contains an $x$-$y$ path whose internal vertices lie in $V(H_i)\setminus\{x,y\}$.
    Moreover, since there is no edge $xy$ in $G$ the path goes from $x$ to a vertex of $V(H_i) \setminus \{x,y\}$ and then eventually reaches vertex $y$.
    Since the subgraphs $H_1,\dots,H_h$ are vertex-disjoint, these paths are internally vertex-disjoint and distinct.

    Then, from \Cref{obs:disjoint_paths_yield_clique},
    we get that $\leftBagAll^\PP(\{x, y\})$ and $\rightBagAll^\PP(\{x, y\})$ are well-defined.
    Moreover, $\leftBagAll^\PP(\{x, y\}) \leq \rightBagAll^\PP(\{x, y\})$.
    Since both of these bags contain $\{x,y\}$, and therefore have at least one vertex in common, we have that $|X_{\rightBagAll^\PP(\{x,y\})} \cup X_{\leftBagAll^\PP(\{x,y\})}| \leq 2\eta + 1$.

    Now, assume for contradiction that the statement does not hold. That is, for every $i \in [h]$,
    either $\leftBagAll^\PP(\{x, y\}) \geq \leftBagAny^\PP(H_i)$ or $\rightBagAny^\PP(H_i) \geq \rightBagAll^\PP(\{x, y\})$.
    In the former case, since some vertex of $H_i$ has $y$ as a neighbor, and some vertex of $H_i$ has $x$ as neighbor, we have $\rightBagAny^\PP(H_i) \geq \leftBagAll^\PP(\{x, y\})$.
    Hence, by \cref{thm:connected_subgraphs_have_interval}, a vertex of $H_i$ is part of the bag $X_{\leftBagAll^\PP(\{x, y\})}$.
    In the latter case, similarly we have $\leftBagAny^\PP(H_i) \leq \rightBagAll^\PP(\{x, y\})$.
    Hence, some vertex of $H_i$ is part of the bag $X_{\rightBagAll^\PP(\{x,y\})}$.
    Therefore, we have that $h \leq |X_{\leftBagAll^
        \PP(\{x, y\})} \cup X_{\rightBagAll^\PP(\{x, y\})}| \leq 2\eta + 1$, a contradiction.
\end{proof}

Now, we state one of the main tools we use in our kernel.
The goal of the following theorem is modifying a path decomposition of a graph into a path decomposition that has some additional useful property.

\begin{restatable}{theorem}{thmEmbeddingWithTwinningGraph}
    \label{thm:embedding_with_twinning_graph}
    Let \(G\) be a graph, and \(H_1\), \(H_2\) be vertex-disjoint induced subgraphs of \(G\) such that \(\pw(H_1) \geq \pw(H_2)\), and there are no edges from \(V(H_1)\) to \(V(H_2)\) in \(G\).
    Let \(X_1,\dots,X_t\) be a path decomposition of \(G\) of width at most \(\eta\)
    that contains a bag \(X_i\) with the following properties:
    \begin{bracketenumerate}
        \item \(X_i\) contains at least \(\pw(H_1) + 1\) vertices of \(H_1\), and
        \item \(N_{G}(H_1) \cup N_G(H_2) \subseteq X_i\).
    \end{bracketenumerate}
    Then, there exists a path decomposition \(\hat \PP\) of \(G\) of width at most \(\eta\)
    such that each bag \(\hat X\) of $\hat\PP$ with \(\hat X \cap (V(H_1) \cup V(H_2)) \neq \emptyset\) contains all vertices of \(X_i \setminus (V(H_1) \cup V(H_2))\).
\end{restatable}

Observe that \cref{thm:embedding_with_twinning_graph} yields a path decomposition
where, in particular, \(V(H_1) \cup V(H_2)\) is stable.
\begin{proof}
    Let \(Y_1,\dots,Y_r\) be a path decomposition of \(H_1\) of width \(\pw(H_1)\),
    and let \(Z_1,\dots,Z_q\) be a path decomposition of \(H_2\) of width \(\pw(H_2)\).
    We modify \(X_1,\dots,X_t\) by adding \(r+q\) copies \(X_i^j\) of \(X_i\) for \(j \in [r+q]\),
    to obtain
    \[\mathcal{P} = X_1,\dots,X_{i},X_i^1,\dots,X_i^{r+q},X_{i+1},\dots,X_t,\]
    which still is a path decomposition of \(G\) of width \(\eta\).
    We obtain the sequence \(\PP'\) from \(\PP\) by deleting the vertices \(V(H_1) \cup V(H_2)\) from every bag.
    Then \(\PP'\) is a path decomposition of the graph \(G - (H_1 \cup H_2)\).
    Further, \(|X_i^j| \leq \eta - \pw(H_1)\) for each copy \(j\).
    We note that, since \(G\) contains no edge between \(V(H_1)\) and \(V(H_2)\),
    each bag \(X_i^j\) of \(\mathcal{P}'\) still contains all vertices of \(N_G(H_1) \cup N_G(H_2)\).

    We obtain the bag sequence \(\hat\PP\) from \(\PP'\)
    by adding \(Y_j\) to \(X_i^j\) for each \(j \in \range[1]{r}\)
    and adding \(Z_j\) to bag \(X_i^{r+j}\) for each \(j \in \range{q}\).
    Finally, we show that \(\hat \PP\) is a path decomposition as claimed by the theorem statement to conclude the proof:
    \begin{itemize}
        \item The width of \(\hat\PP\) is at most \(\eta\)
              since every bag either occurs in \(\PP\) or has size at most \(|X_i|- \pw(H_1) + \max\{\pw(H_1),\pw(H_2)\} \leq |X_i|\).
              Here, we use that \(\pw(H_1) \geq \pw(H_2)\).
        \item
              Every bag \(\hat X_k\) with \(\hat X_k \cap (V(H_1) \cup V(H_2)) \neq \emptyset\)
              contains all vertices of \(X_i \setminus (V(H_1) \cup V(H_2))\).
              Indeed, the intersection is nonempty only for bags \(X_i^j\) for some \(j \in [r+q]\).
              Then, by construction, \(X_i^j\) contains all of \(X_i \setminus (V(H_1) \cup V(H_2))\).
        \item
              Finally, we claim that \(\hat \PP\) is a path decomposition of \(G\).
              Indeed, $\hat \PP$ is built by overlapping path decompositions of $G - (H_1 \cup H_2)$, $H_1$ and $H_2$, and all edges between these three graphs are taken care of by the fact that each bag \(\hat X_k\) with \(\hat X_k \cap (V(H_1) \cup V(H_2)) \neq \emptyset\)
              contains all vertices of \(X_i \setminus (V(H_1) \cup V(H_2)) \supseteq N_G(H_1) \cup N_G(H_2)\). \qedhere
    \end{itemize}
\end{proof}

\subsection{The Reduction Rule}
We are now ready to state the algorithm to reduce the number of connected components.

\begin{algorithm}[t]
    \caption{\algofont{ReduceComponents}}
    \label{rule:reduce_ccs}
    \Input{Instance \((G,k,M)\) of \pwEtaDeletionEta{}
        where $G-M$ has more than $f_{\ref{thm:reduce_ccs}}( |N_G(G-M)| )$ connected components.}
    \(\CC \leftarrow\) connected components of \(G - M\)\;
    \(N \leftarrow N_G(G-M)\)\;
\For{\(X \in \binom{N}{2} \cup \binom{N}{1}\) \label{rule:reduce_ccs:marking_start}}{
        \(\mathcal{C}_X \leftarrow \{C \in \mathcal{C} \mid X \subseteq N_G(C)\}\)\;
        \uIf{\(|\mathcal{C}_X| \leq 2|N|(\eta +2) + 1\)}{
            Mark each graph of \(\mathcal{C}_X\)\;
        }
        \Else{
            \(C_1,\dots,C_\ell \leftarrow\) graphs of \(\mathcal{C}_X\) ordered by nonincreasing pathwidth\;
            Mark the first \(2|N|(\eta +2) + 1\) graphs of \(C_1,\dots,C_\ell\)\;
        }
    }\label{rule:reduce_ccs:marking_end}
    \(C^\star \leftarrow\) an arbitrary unmarked graph of \(\mathcal{C}\)\;
    \Return{\((G - C^\star,k,M)\)}
\end{algorithm}

We state our reduction rule in form of the algorithm \nameref{rule:reduce_ccs}.
First, we apply a marking scheme on the connected components of \(G-M\)
for every set \(X \in \binom{N}{2} \cup \binom{N}{1}\), where \(N = N_G(G-M)\).
If such a component \(C\) is marked due to \(X\),
we say that \(X\) \emph{marked} \(C\).
Finally, we delete a connected component that is unmarked.
As we show next, this reduction rule works as intended.

\thmReduceCCs*{}
\begin{proof}
    Set \(N = N_G(G - M)\) like in the algorithm.

Initially, a marking scheme is applied in \reflines{rule:reduce_ccs:marking_start}{rule:reduce_ccs:marking_end}.
    For each \(X \in \binom{N}{2} \cup \binom{N}{1}\), we mark at most \(2|N|(\eta + 2) + 1\) connected components of \(G - M\) that have \(X\) in their neighborhood.
    Therefore, we mark at most \(|N|^2 \cdot (2|N|(\eta + 2) + 1)\) connected components in total, meaning that some component \(C^\star\) is not marked.
    Further, for each \(X\), every component that is marked by \(X\) has a pathwidth that is at least as large as that of each unmarked component with \(X\) in its neighborhood.

    We need to show that \((G,k,M)\) is a yes-instance if and only if \((G-C^\star,k,M)\) is a yes-instance.
    The forward direction follows immediately since \(G-C^\star\) is a subgraph of \(G\) and the solution size \(k\) remains unchanged.

    To see the backward direction,
    assume that instance \((G - C^\star,k,M)\) is a yes-instance.
    Then, there is a subset of vertices \(S \subseteq V(G-C^\star)\) of size at most \(k\)
    such that the resulting graph \(\tilde{G} = (G - C^\star) - S\)
    has pathwidth \(\pw(\tilde{G}) \leq \eta\), and, among all such sets,
    let $S$ minimize $|S \cap V(G - M)|$.
    Then, there is a path decomposition \(\mathcal{P} = X_1,\dots,X_t\) of \(\tilde{G}\) of width at most \(\eta\).
    If \(S\) contains all vertices of \(N_G(C^\star)\), then it is easy to extend \(\PP\) to a path decomposition of \(G - S\) since $\pw(C^\star) \leq \eta$.
    So, we assume that \(N_G(C^\star) \setminus S\) is nonempty, which also implies $|N| \geq 1$.
    By choice of $S$ and due to \cref{obs:mostly_intact}, we obtain the following.
    \begin{claim}\label{thm:solution_intersection_components_small}
        We have \(|S \cap V(G - M)| < |N|\).
    \end{claim}

    Now, we use the above claim to find a useful bag \(X_{i^*}\).
    \begin{claim}
        \label{thm:reduce_ccs:claim:clique_bag}
        There exists a bag \(X_{i^*}\) that contains \(N_{G}(C^\star) \setminus S\).
    \end{claim}
    \begin{claimproof}
        For any distinct \(x,y \in N_G(C^\star) \setminus S\) we show that there are \(\eta + 2\) internally vertex-disjoint \(x\)-\(y\) paths in \(\tilde{G}\).
        Then, the proof follows from \cref{obs:disjoint_paths_yield_clique}.
        Hence, consider some distinct \(x,y \in N_G(C^\star) \setminus S\).
        Recall that procedure \nameref{rule:reduce_ccs} did not mark \(C^\star\),
        and therefore \(\{x,y\}\) did not mark \(C^\star\).
        This implies that there are \(2|N|(\eta +2) + 1\) marked connected components \(C_i\) of $G - M$
        such that \(\{x,y\} \subseteq N_G(C_i)\), and \(\pw(C_i) \geq \pw(C^\star)\).
        Hence, there are \(2|N|(\eta +2) + 1\) internally vertex-disjoint \(x\)-\(y\) paths in \(G-C^\star\).
        Since \(|S \cap V(G - M)| < |N|\) by \cref{thm:solution_intersection_components_small}, graph \(\tilde{G}\) still contains at least \(\eta + 2\) of these paths.
    \end{claimproof}

    Next, we use this claim to find a connected component \(C\) that is similar to \(C^\star\).
    \begin{claim}
        \label{thm:reduce_ccs:claim:twinning_graph}
        There is a connected component \(C\) of \(G - M\) such that \begin{bracketenumerate}
            \item \(C\) is also a connected component of \(\tilde{G} - M\),
            \item each bag \(X_i\) with \(X_i \cap V(C) \neq \emptyset\) contains all vertices of \(N_{G}(C)\setminus S\),
            \item each bag \(X_i\) with \(X_i \cap V(C) \neq \emptyset\) contains all vertices of \(N_G(C^\star) \setminus S\),
            \item \(\pw(C) \geq \pw(C^\star)\).
        \end{bracketenumerate}
    \end{claim}
    \begin{claimproof}
        Let \(\ell = \leftBagAll^{\mathcal{P}}(N_G(C^\star) \setminus S)\) and \(r = \rightBagAll^{\mathcal{P}}(N_G(C^\star) \setminus S)\), note that these indices are well-defined by \cref{thm:reduce_ccs:claim:clique_bag}.
        Then, since \(N_G(C^\star) \setminus S\) is nonempty by assumption, there is some vertex \(m_\ell \in N_G(C^\star) \setminus S\) that is contained in \(X_\ell\) but not in \(X_{\ell - 1}\), where we set \(X_0 = X_{t+1} = \emptyset\) for convenience.
        Similarly, there is some vertex \(m_r \in N_G(C^\star) \setminus S\) that is contained in \(X_{r}\) but not in \(X_{r+1}\).
        Note that \(m_\ell\) may be identical to \(m_r\).

        Let \(\mathcal{C}'\) be the set of connected components of \(G - M\) marked by the pair \(\{m_\ell,m_r\}\).
        Observe that $\{m_\ell,m_r\} \subseteq N_G(C^\star)$, but $C^\star$ was not marked by \(\{m_\ell,m_r\}\).
        Thus, the cardinality of the marked components \(\CC'\) is \(|\mathcal{C}'| = 2|N|(\eta + 2) + 1\).
        Since we mark in order of nonincreasing pathwidth,
        every component \(C' \in \CC'\) satisfies property (4), that \(\pw(C') \geq \pw(C^\star)\).
        Let \(\CC''\) be the components \(C'\) of \(\CC'\) where \(S \cap V(C') = \emptyset\).
        Every component \(C'' \in \CC''\) satisfies property~(1) that \(C''\) is also a connected component of \(\tilde{G}-M\).
        Further, \(|\CC''|\geq 2|N|(\eta +1) + 1\), since \(|S \cap V(G - M)| < |N|\) by \cref{thm:solution_intersection_components_small}.

        By \cref{thm:find_stable_graph}, there is a graph \(C \in \CC''\) that is \(N \setminus S\)-stable.
        As $N_{\tilde G}(C) \subseteq N\setminus S$, graph $C$ is stable.
        From \cref{thm:stable_graph_bags_contain_neighbors} we obtain property (2), that each bag with \(X_i \cap V(C) \neq \emptyset\) contains all of \(N_{\tilde G}(C)\).

        It remains to show property (3).
        By our choice of $m_\ell$, $m_r$, each bag that contains both $m_\ell$ and $m_r$ contains all vertices of $N_G(C^\star) \setminus S$.
        Note that $\{m_\ell,m_r\} \subseteq N_G(C)$, and therefore each bag that contains a vertex of $C$ contains both $m_\ell$ and $m_r$.
        Hence, each bag that contains a vertex of $C$ contains all vertices of $N_G(C^\star) \setminus S$.
    \end{claimproof}

    We want to finish the backward direction of safety by applying \cref{thm:embedding_with_twinning_graph}.
    Let \(X_i\) be a bag that contains \(\pw(C) +1 \geq \pw(C^\star) + 1\) vertices of \(C\).
    Create \(\hat G\) by taking the disjoint union of \(C^\star\) and \(\tilde G\).
    Create a path decomposition \(\hat{\mathcal{P}}_1\) of \(\hat G\) by appending a path decomposition of \(C^\star\) of width at most \(\eta\) to the decomposition \(\mathcal{P}\).
    Note that, by \cref{thm:reduce_ccs:claim:twinning_graph}, the bag $X_i$ contains all vertices of $N_{\hat G}(C) \cup N_{\hat G}(C^\star) = N_{G}(C) \setminus S$.
    Therefore, \(\hat G,C,C^\star\), path decomposition \(\hat{\mathcal{P}}_1\) and bag \(X_i\) fulfill the conditions of \cref{thm:embedding_with_twinning_graph}.
    By the theorem, there exists a path decomposition \(\hat{\PP}_2 = \hat{X}_1,\dots, \hat{X}_{\hat{t}}\) of \(\hat G\) of width at most \(\eta\), such that each bag \(\hat{X}_k\) that contains a vertex of \(C^\star\) contains all vertices of \(X_i \setminus (V(C) \cup V(C^\star))\).
    Since \(C, C^\star\) are connected components of \(G - M\), and \(N_G(C^\star) \setminus S \subseteq X_i\) by \cref{thm:reduce_ccs:claim:twinning_graph}, $\hat{\PP}_2$ is also a path decomposition of \(G - S\).

    To finish the proof, we notice that \nameref{rule:reduce_ccs} outputs a graph that is a proper subgraph of \(G\).
    Finally, note that the reduction rule runs in polynomial-time, because in particular we can compute the pathwidth of components of \(G - M\) in polynomial-time using \cref{thm:path_decomposition_algorithm} since $\eta$ is a constant.
\end{proof}

Note that we have fixed $\eta \geq 1$ earlier at the start of this section as this global assumption is useful later in the paper.
However, \cref{thm:reduce_ccs} and its proof also work for $\eta = 0$, that is for \textsc{Vertex Cover}.

Since $|N| \leq |M|$ (where $N = N_G(G - M)$) we can use \nameref{rule:reduce_ccs} to obtain a cubic bound on the number of connected components.
However, the reduction rule is significantly more powerful than that, because there can be instances \((G,k,M)\) of the problem where \(M\) is large but \(N\) is very small.
In such a case the obtained cubic bound in \(|N|\) is much better than a cubic bound in \(|M|\).
In fact, in our kernel we will also apply this rule for the case when \(M\) is roughly as large as \(G\), but \(N\) is a constant, in which case a cubic bound in \(|N|\) is a constant, whereas a cubic bound in $|M|$ does not give any nontrivial guarantees.

As a corollary, we obtain a uniform kernel when parameterizing by vertex cover size.
Since there are some minor details we need to take care of, we add a proof.

\thmUniformKernelVcParameterization*{}
\begin{proof}
    If the input instance $(G,k,M)$ fulfills $k \geq |M|$, then it is a yes-instance since $\pw(G - M) = 0 \leq \eta$.
    In this case, we output the trivial yes-instance $((M,\emptyset),0,M)$, this is one of the cases where we output the empty graph if $M = \emptyset$.

    If the number of connected components of $G - M$ is at most $f_{\ref{thm:reduce_ccs}}(|M|)$, then we already have $|V(G)|,k \in O(|M|^3)$, and we output $(G,k,M)$.
    Otherwise, we exhaustively apply \cref{thm:reduce_ccs}, and eventually obtain an equivalent instance $(G',k,M)$ with $|V(G')|,k \in O(|M|^3)$, where $G'$ is a subgraph of $G$.

    As explained after the proof of \cref{thm:reduce_ccs}, the case $\eta = 0$ can also be handled by our algorithm.
    Alternatively, the case $\eta = 0$ can, for example, also be handled
    by a standard modification of the Buss kernel~\cite{bussNondeterminism1993} that ensures the desired properties about $G'$ and $k'$,
    or by directly using the \textsc{Vertex Cover} kernel of Fomin et al.~\cite{DBLP:journals/jcss/FominJP14}.
\end{proof} 
\section{Reducing the Degree of Modulator Vertices}
\label{sec:kernel:bounding_mod_degree}
Now, we proceed to the reduction rules that are responsible for bounding the degree of vertices in $M$.
We will again start with some key results and definitions that appear multiple times in the rules.
We aim to give the lemmas and definitions in the most general way, that is, even though we want to eventually obtain a kernel for the problem \pwEtaDelElimGeneral{}, we may define notions and give results for \pwEtaDeletion{} or \pwEtaDeletionEta{}.

\subsection{The Key Ingredients}
\label{subsec:key_ingredients_modulator_degree_reduction}
We begin by introducing the notion of \emph{virtual cliques}.

\begin{definition}[Virtual Edge, Virtual Clique]
    \label{def:virtual_clique_edge}
    Let $(G,k)$ be an instance of \pwEtaDeletion{}.
    Two vertices $x,y \in V(G)$ with $x \neq y$ have a \emph{virtual edge} between them if, for any set $S \subseteq V(G)$ of size at most $k+1$ such that $\pw(G - S) \leq \eta$, and any path decomposition $X_1,\dots,X_t$ of $G - S$ of width at most $\eta$, there exists a bag $X_i$ that contains $\{x,y\} \setminus S$.

    For a vertex set $Z \subseteq V(G)$ we say that $Z$ is a \emph{virtual clique} (of $(G,k)$) if for each $x,y \in Z$ with $x \neq y$ we have that $x,y$ have a virtual edge between them.
\end{definition}

The key point of virtual cliques is that they function as if $Z$ were an actual clique.
Formally, by applying Helly's theorem for trees (\cref{thm:helly_property_of_trees}) similarly as for \cref{obs:disjoint_paths_yield_clique} we obtain the following.

\begin{observation}
    \label{obs:virtual_clique_appears_in_bag}
    Let $(G,k)$ be an instance of \pwEtaDeletion{}, and $Z \subseteq V(G)$ a virtual clique of $(G,k)$.
    Then, for any set $S \subseteq V(G)$ of size at most $k+1$ such that $\pw(G - S) \leq \eta$, and any path decomposition $X_1,\dots,X_t$ of $G - S$ of width at most $\eta$, there is a bag $X_i$ that contains $Z \setminus S$.
\end{observation}

By \cref{obs:virtual_clique_appears_in_bag}, we could clearly have a rule that makes a virtual clique $Z$ into an actual clique by adding edges.
In fact, this is a well-known step used in the literature for other \fDeletion{} problems; see, e.g.,~\cite[Lemma 2]{cyganHardnessLosingWidth2014} and \cite[Lemma 4.7]{giannopoulouUniformKernelizationComplexity2017}.
Similar edge addition rules have also been used in algorithms and kernels for the problems of computing the treewidth and pathwidth of graphs~\cite{bodlaenderLineartimeAlgorithmFinding1996,bodlaenderPreprocessingTreewidthCombinatorial2013,bodlaenderKernelBoundsStructural2012,clautiauxNewLowerUpper2003}.
However, since we want to ensure that the graph output by our kernel is a minor of the input graph, we want to avoid having such a rule.
Instead, we use this more abstract notion of virtual cliques to achieve the same properties without adding edges.

Note that \cref{def:virtual_clique_edge} has some leeway with regard to the solution size: vertices $\{x,y\} \setminus S$ appear in some bag even if the solution has size $k+1$.
This property will be handy, because we want to exploit virtual cliques even after deleting an edge of the graph $G$.
Now, we introduce the notion of a \emph{polishing} set of an instance, which builds upon the definition of virtual cliques.

\begin{definition}[Polishing Set]
    \label{def:polishing_set}
    Let $(G,k,M)$ be an instance of \pwEtaDeletionEta{}.
    Then, a set $Y \subseteq V(G) \setminus M$ is a \emph{polishing set} of this instance if
    \begin{bracketenumerate}
        \item $|Y| \leq |M|^2 \cdot (k + \eta + 2) \cdot (\eta + 1)$, and
        \item for each connected component $C$ of $G - (M \cup Y)$, the neighborhood $N_G(C) \cap M$ is a virtual clique of $(G,k)$, and $|N_G(C)  \cap Y| \leq 2( \eta + 1)$.
    \end{bracketenumerate}
\end{definition}

A polishing set always exists and can be computed efficiently.
We note that the computation of a polishing set is, by now, a well-known step present in the literature.
For example, it is one of the key ingredients for the solution-size kernel for $\fDeletion{}$
in case $\mathcal{F}$ contains a planar graph given by Fomin et al.~\cite{fominPlanarFdeletionApproximation2012}.
We provide a proof purely for the sake of completeness.

\begin{restatable}{lemma}{thmComputePolishingSet}
    \label{thm:polish_instance}
    There is a polynomial-time algorithm \polishingSetAlgo{} that takes an instance $(G,k,M)$ of \pwEtaDeletionEta{} as input and outputs a polishing set $Y$ of the instance.
\end{restatable}
\begin{proof}
    The algorithm works as follows.
    Compute a path decomposition $X_1,\dots,X_t$ of $G - M$ of width at most $\eta$ in linear time using \cref{thm:path_decomposition_algorithm}.
    For each $\{m_1,m_2\} \in \binom{M}{2}$, with $m_1m_2 \notin E(G)$, we compute a minimum-size set $Z_{\{m_1,m_2\}}$ such that $G - Z_{\{m_1,m_2\}}$ has no $m_1$-$m_2$ path and $m_1,m_2 \notin Z_{\{m_1,m_2\}}$.
    This is possible in polynomial-time using max-flow/min-cut algorithms.
    For each pair $\{m_1,m_2\}$ with $m_1m_2 \notin E(G)$ and $|Z_{\{m_1,m_2\}}| \leq k + \eta + 2$, and each $w \in Z_{\{m_1,m_2\}} \setminus M$ we mark a bag of $X_1,\dots,X_t$ that contains vertex $w$.
    Finally, our algorithm outputs $Y$, the union of all marked bags.

    We claim that the output set $Y$ is a polishing set of the input instance $(G,k,M)$.
    The size of $Y$ is at most $|M|^2 \cdot (k + \eta + 2) \cdot (\eta + 1)$, because we mark at most $k + \eta + 2$ bags for each pair $\{m_1,m_2\}$, and each bag contains at most $\eta + 1$ vertices.

    We proceed to proving the second property of polishing sets and begin by showing that for each connected component $C$ of $G - (M \cup Y)$ we have that $N_G(C) \cap M$ is a virtual clique.
    Fix some connected component $C$ of $G - (M \cup Y)$, let $S$ be a set of size at most $k+1$ such that $\pw(G- S) \leq \eta$, and let $X_1',\dots,X_t'$ be a path decomposition of $G - S$ of width at most $\eta$.
    Let $m_1,m_2$ be two distinct vertices of $N_G(C) \cap M$.

    If $\{m_1,m_2\} \cap S \neq \emptyset$, it is clear that some bag $X_i'$ contains $\{m_1,m_2\} \setminus S$.
    So we assume $\{m_1,m_2\} \cap S = \emptyset$.
    If $m_1m_2 \in E(G)$, then it is also clear that $m_1,m_2$ must always appear together in some bag $X_i'$.
    Therefore, we may additionally assume there is no edge between $m_1$ and $m_2$.

    There is an $m_1$-$m_2$ path in $G[V(C) \cup \{m_1,m_2\}]$, and therefore the pair $\{m_1,m_2\}$ did not mark any bags for the definition of the set $Y$.
    This implies $|Z_{\{m_1,m_2\}}| \geq k + \eta + 3$, and
    hence, by Menger's theorem, there are at least $k + \eta + 3$ internally vertex-disjoint paths from $m_1$ to $m_2$ in $G$.
    Since $|S| \leq k +1$, at least $\eta + 2$ of these paths still exist in $G - S$.
    Therefore, by \cref{obs:disjoint_paths_yield_clique} there is a bag $X_i'$ that contains $\{m_1,m_2\}$.
    So, there is a virtual edge between $m_1$ and $m_2$, and as they were arbitrary neighbors of $N_G(C) \cap M$, also that $N_G(C) \cap M$ is a virtual clique.

    Finally, observe that the vertices of $C$ stem from some interval of bags $X_i,\dots,X_j$, such that $N_{G - M}(C) \subseteq X_{i-1} \cup X_{j+1}$ (where we assume that $X_0 = X_{t+1} = \emptyset$).
    Therefore, $|N_G(C) \cap Y| \leq 2(\eta + 1)$.
    This concludes the proof, since we have now shown that $Y$ is a polishing set, and that it can be computed in polynomial-time.
\end{proof}

Next, we introduce protrusions and near-protrusions,
which are fundamental to the design of kernelization algorithms; see~\cite{DBLP:journals/jacm/BodlaenderFLPST16,DBLP:journals/siamcomp/FominLST20,fominPlanarFdeletionApproximation2012} for a more detailed overview.
Note that the definitions we use are not completely identical to the standard definitions of these concepts, but they still essentially serve the same purpose.

\begin{definition}[Protrusions]
    \label{def:protrusion}
    Let $G$ be a graph.
    A \emph{protrusion} of $G$ is an induced subgraph $H \subseteq G$ such that $|N_G(H)| \leq 2(\eta + 1)$ and $\pw(H) \leq \eta$.
\end{definition}

Protrusions represent parts of a graph with a small boundary that can be safely replaced
by constant-size equivalent structures without affecting the existence of a solution.
A similarly useful notion is that of a near-protrusion, introduced in~\cite{fominPlanarFdeletionApproximation2012}
as an induced subgraph that becomes a protrusion upon deletion of an optimal solution.

\begin{definition}[Near-protrusions]
    \label{def:near_protrusion}
    Let $(G,k)$ be an instance of \pwEtaDeletion{}.
    An induced subgraph $H \subseteq G$ is a \emph{near-protrusion}  of $(G,k)$ if $\pw(H) \leq \eta$, and for all sets $S \subseteq V(G)$ of size at most $k+1$ such that $\pw(G - S) \leq \eta$, we have $|N_G(H) \setminus S| \leq 3(\eta + 1)$.
\end{definition}
Therefore, if we have a near-protrusion $H$ of $G$, even a solution that is slightly larger than $k$ will select almost all neighbors of $H$.
Now, we exhibit why polishing sets are useful: we can use them to obtain near-protrusions.

\begin{lemma}
    \label{thm:ccs_are_near_protrusions}
    Let $(G,k,M)$ be an instance of \pwEtaDeletionEta{}, and $Y$ be a polishing set of $(G,k,M)$.
    Then, each connected component $C$ of $G - (M \cup Y)$ is a near-protrusion of $(G,k)$.
\end{lemma}
\begin{proof}
    Since $C$ is a connected component of $G - (M \cup Y)$, we have $\pw(C) \leq \eta$.
    Now, note that $N_G(C)$ can be partitioned into $A = N_G(C) \cap M$, a virtual clique, and $B = N_G(C) \cap Y$, of size at most $2(\eta + 1)$.
    Now, let $S$ be a set of size at most $k+1$ such that $\pw(G - S) \leq \eta$.
    Let $X_1,\dots,X_t$ be a path decomposition of $G -S$ of width at most $\eta$.
    By \cref{obs:virtual_clique_appears_in_bag} there exists a bag $X_i$ that contains $A \setminus S$.
    As $|X_i| \leq \eta + 1$, we obtain $|A \setminus S| \leq \eta + 1$.
    Therefore, $|N_G(C) \setminus S| \leq |(A \cup B) \setminus S| \leq \eta + 1 + 2(\eta + 1) = 3(\eta + 1)$, showing that $C$ is a near-protrusion.
\end{proof}

Further, we show that any minimum solution selects only few vertices of a near-protrusion, and that is true even after deleting an edge.
The following is similar to \cref{obs:mostly_intact}, however it gives slightly stronger guarantees and is tuned to the definition of near-protrusions.

\begin{lemma}
    \label{thm:near_protrusions_yield_small_intersection_and_unsel_neighborhood}
    Let $(G,k,M)$ be an instance of \pwEtaDeletionEta{} and $C$ be a near-protrusion of $(G,k)$.
    Let $G'$ result from $G$ by deleting an edge between $v \in V(C)$ and $m \in N_G(C)$.
    Let $S'$ be a set of minimum size such that $\pw(G' - S') \leq \eta$, and assume $|S'| \leq k$.
    Then, we have
    \begin{bracketenumerate}
        \item $|N_G(C) \setminus S'| \leq 3(\eta + 1)$, and
        \item $|V(C) \cap S'| \leq 3(\eta + 1)$.
    \end{bracketenumerate}
\end{lemma}
\begin{proof}
    We have that $S = S' \cup \{v\}$ is a set of size at most $k+1$ such that $\pw(G -S) \leq \eta$.
    By \cref{def:near_protrusion} we have that $|N_G(C) \setminus S| \leq 3(\eta + 1)$.
    As $v \in V(C)$ and $G'$ is created from $G$ by deleting the edge $mv$, we have that $N_G(C) \setminus S = N_{G}(C) \setminus S'$.
    Therefore, we have property (1).

    Now, assume that $|V(C) \cap S'| > 3(\eta + 1)$.
    Set $\hat S = (S' \setminus V(C)) \cup N_{G'}(C)$.
    Each connected component of $C$ is a connected component of $G' - \hat S$, and as $C$ is a near-protrusion we have $\pw(C) \leq \eta$.
    Moreover, we have $S' \cap V(G' - C) \subseteq \hat S \cap V(G' - C)$, and therefore $\pw(G' - \hat S) \leq \eta$.
    From property (1) and $N_{G'}(C) \subseteq N_{G}(C)$ we have $|N_{G'}(C) \setminus S'| \leq 3(\eta + 1)$.
    Hence, $|\hat S| \leq |S'| - |V(C) \cap S'| + |N_{G'}(C) \setminus S'| < |S'|$, contradicting that $S'$ was of minimum size.
\end{proof}

\subsection{The Reduction Rule}
We proceed to the reduction rule we use.
The rule is complex, and uses multiple different values that represent bounds we use to achieve certain properties.
To improve the exposition, we introduce symbols and functions that represent these values.
The function $\degreeBoundTreedepthCase: \nat \to \nat$ will be defined later in \cref{def:important_bounds} and is rather complex, but we already remark that it is in $\Oh(x^2)$ where $x$ is its input, and therefore uniformly polynomial, and also that it is monotonically increasing.

\begin{definition}[Important Bounds for the Modulator Degree Reduction]
    \label{def:important_bounds_degree_reduction}
    We define the following functions
    \begin{align*}
        \polishingSetSize[x,y]           & = x^2 \cdot (y + \eta + 2) \cdot (\eta + 1),                                    \\
        \degreeBoundCaterpillarCase[x] & = \BoundCaterpillarOverall[x] + 1,                                                   \\
        \degreeBoundCC[x]              & = \degreeBoundTreedepthCase(x + 2(\eta + 1)) \cdot \degreeBoundCaterpillarCase[x], \\
        \ccBound[x,y]                    & = (\polishingSetSize[x,y] + x)^2 \cdot (2(\polishingSetSize[x,y] + x)(\eta+2) + 1),     \\
        \degreeBound[x,y]                & = \degreeBoundCC[x] \cdot \ccBound[x,y] + \polishingSetSize[x,y].                              \\
    \end{align*}
\end{definition}

Let us elaborate a bit on the meaning of these values and the reduction rule.
We will compute a polishing set $Y$, that has size at most $\polishingSetSize$, of the instance $(G,k,M)$, and then aim to bound the number of neighbors that each $m \in M$ has in $G - M$.
This is achieved by bounding the degree of $m$ in each connected component of $G - (M \cup Y)$ by $\degreeBoundCC$, and the number of such components by $\ccBound$.
Then, the number of neighbors of $m$ in $G - M$ is bounded by $\degreeBound$.
To bound the degree of $m$ in component $C$ of $G - (M \cup Y)$, we compute a $\pwEtaBoundClass[1]$-elimination forest $\FF$ of $C$.
We then bound the number of $\FF$-neighbors of $m$ by $\degreeBoundTreedepthCase(|M| + 2(\eta + 1))$, and the number of neighbors of $m$ in caterpillars corresponding to leaves of $\FF$ by $\degreeBoundCaterpillarCase$.

The goal of the rest of this section is to prove \cref{thm:reduce_modulator_degree}.
\begin{restatable}{lemma}{thmReduceDegree}
    \label{thm:reduce_modulator_degree}
    Let $\eta, \beta \geq 0$ be fixed constants.
    \namecref{rule:reduce_modulator_degree} \nameref{rule:reduce_modulator_degree} runs in polynomial-time, and given an instance $(G,k,M)$ of \pwEtaDelElimGeneral{} as input, it outputs an equivalent instance $(G',k,M)$ where $G'$ is a minor of $G$.
    If there is a vertex $m \in M$ with $|N_G(m) \setminus M|  > \degreeBound$, then $G'$ is a proper minor of $G$.
\end{restatable}

Before proving this lemma and stating algorithm \nameref{rule:reduce_modulator_degree},
we introduce the properties of two other algorithms we will use as sub-routines.
The first of them is \nameref{rule:degree_reduction_elimination_forest}.
It reduces the instance size if there is an induced subgraph $C$ of $G - (M \cup Y)$, such that some vertex $m \in M$ has many $\FF$-neighbors in a $\pwEtaBoundClass[1]$-elimination forest $\FF$ of $C$; here $Y$ is again a polishing set of the instance.
By \cref{thm:ccs_are_near_protrusions} we have that $C$ is a near-protrusion, and moreover $|N_G(C)| \leq |M| + 2(\eta+1)$, which are properties we can exploit.
Formally, we have the following result.

\thmTreedepthCase*{}

Of course, we will apply \cref{thm:treedepth_case} for $\HH = \pwEtaBoundClass[1]$, but it actually works for any graph class $\HH$.
The second algorithm we utilize, \nameref{rule:degree_reduction_caterpillar}, aims to reduce the size of the instance if, for some connected component $H$ of $G - (M \cup Y)$,
\cref{thm:treedepth_case} does not apply,
but still some vertex $m \in M$ has many neighbors in $H$.
In that case, there is some subgraph $C$ of $H$, such that $C$ is a caterpillar, and $m$ still has numerous neighbors in $C$, which in turn ensures that $C$ is large.
We prove the following lemma in~\Cref{subsec:bounding_degree_in_elim_forest}.

\thmCaterpillarCase*{}

\begin{algorithm}[t]
    \caption{\algofont{ReduceModulatorDegree}}
    \label{rule:reduce_modulator_degree}
    \Input{Instance $(G,k,M)$ of \pwEtaDelElimGeneral{}}
    \If{for all $m \in M$ we have $|N_G(m) \setminus M| \leq \degreeBound$ \label{rule:reduce_modulator_degree_start}}{
        \Return{$(G,k,M)$}
    } \label{rule:reduce_modulator_degree_return1}
    $Y \leftarrow$ \polishingSetAlgo{}$(G,k,M)$ \label{rule:reduce_modulator_degree_compute_polishing_set}\;
    \If{$G - (M \cup Y)$ has more than $\ccBound$ connected components}{
        $(G',k,M \cup Y) \leftarrow$ \nameref{rule:reduce_ccs}$(G,k,M \cup Y)$ \label{rule:reduce_modulator_degree_reduce_components}\;
        \Return{$(G',k,M)$}
    }
    $m \leftarrow$ vertex of $M$ maximizing $|N_G(m) \setminus M|$ \label{rule:reduce_modulator_degree_fix_m}\;
    $C \leftarrow$ connected component of $G - (M \cup Y)$ such that $|N_G(m) \cap V(C)| \geq \degreeBoundCC$ \label{rule:reduce_modulator_degree_fix_component}\;
    $\mathcal{F} = (F,\{Y_n\}_{n \in V(F)}) \leftarrow$ nice $\pwEtaBoundClass[1]$-elimination tree of $C$ \label{rule:reduce_modulator_degree_compute_elim_tree} of depth at most $\beta$\;
    \If{$m$ has more than $\degreeBoundTreedepthCase(|N_G(C)|)$ $\FF$-neighbors in $\mathcal{F}$}{
        \Return{\nameref{rule:degree_reduction_elimination_forest}$\big((G,k,M), C, \mathcal{F}\big)$} \label{rule:reduce_modulator_degree_return_tree_case}
    }
    $\ell \leftarrow$ leaf of $F$ such that $|N_G(m) \cap Y_\ell| \geq \degreeBoundCaterpillarCase$ \label{rule:reduce_modulator_degree_fix_leaf}\;
    $H \leftarrow G[Y_\ell]$\;
    \Return{\nameref{rule:degree_reduction_caterpillar}$\big((G,k,M),H \big)$ \label{rule:reduce_modulator_degree_return_caterpillar_case}}
\end{algorithm}

Now, we can show the main result of this section, that we restate for convenience.
\thmReduceDegree*{}
\begin{proof}
    We go through the steps of the algorithm.
    \reflines{rule:reduce_modulator_degree_start}{rule:reduce_modulator_degree_return1} are clear.
    By \cref{thm:polish_instance} the algorithm to compute polishing set $Y$ used in \refline{rule:reduce_modulator_degree_compute_polishing_set} functions correctly and in polynomial-time; here we use that an instance of \pwEtaDelElimGeneral{} is also an instance of \pwEtaDeletionEta{}.
    Recall that $|Y| \leq |M|^2 \cdot (k + \eta + 2) \cdot (\eta + 1) = \polishingSetSize$ by the definition of polishing sets (\cref{def:polishing_set}).
    Further, recall that $\ccBound = (\polishingSetSize + |M|)^2 \cdot (2(\polishingSetSize + |M|) \cdot (\eta + 2) + 1)$.
    If the number of connected components of $G - (M \cup Y)$ is larger than $\ccBound$,
    we proceed with \refline{rule:reduce_modulator_degree_reduce_components},
    which sets $(G',k,M \cup Y)$ to the output of \nameref{rule:reduce_ccs} applied on instance $(G,k,M \cup Y)$.
    By \cref{thm:reduce_ccs} we have that $G'$ is a proper subgraph of $G$, $(G',k,M \cup Y)$ is equivalent to instance $(G,k,M \cup Y)$.
Our algorithm then outputs instance $(G',k,M)$ which is equivalent to, and strictly smaller than the input instance.

    Otherwise, the number of connected component of $G - (M \cup Y)$ is at most $\ccBound$, and we fix a vertex $m \in M$ with the most neighbors in $V(G) \setminus M$ in \refline{rule:reduce_modulator_degree_fix_m}.
    Since $m$ has at least $\degreeBound = \degreeBoundCC \cdot \ccBound + \polishingSetSize$ neighbors in $V(G) \setminus M$, $m$ has at least $\degreeBoundCC \cdot \ccBound$ neighbors in $V(G) \setminus (M \cup Y)$.
    By the bound on the number of connected components of $G - (M \cup Y)$, there is some connected component $C$ of $G - (M \cup Y)$ such that $m$ has at least $\degreeBoundCC$ neighbors in $C$.
    Therefore, the component $C$ fixed in \refline{rule:reduce_modulator_degree_fix_component} indeed exists.
    The nice $\pwEtaBoundClass[1]$-elimination tree $\FF = (F,\{Y_n\}_{n \in V(F)})$ of depth at most $\beta$ of \refline{rule:reduce_modulator_degree_compute_elim_tree} is computed using \cref{thm:compute_elimination_forests}.

    Then, if $m$ has more than $\degreeBoundTreedepthCase(|N_G(C)|)$ $\FF$-neighbors in $\mathcal{F}$, in \refline{rule:reduce_modulator_degree_return_tree_case}, we output the result of applying \nameref{rule:degree_reduction_elimination_forest} on input $\big((G,k,M),C, \mathcal{F} \big)$.
    The output instance has the correct properties by \cref{thm:treedepth_case}, and the fact that $C$ is a near-protrusion by \cref{thm:ccs_are_near_protrusions}.

    Otherwise, $m$ has at most $\degreeBoundTreedepthCase(|N_G(C)|)$ $\FF$-neighbors.
    Recall that $\degreeBoundTreedepthCase$ is a monotonically increasing function, and that $|N_G(C)| \leq |M| + 2(\eta + 1)$ because $N_G(C) \subseteq M \cup Y$ and $Y$ is a polishing set.
    As $m$ has at least $\degreeBoundCC = \degreeBoundTreedepthCase(|M| + 2(\eta + 1)) \cdot \degreeBoundCaterpillarCase \geq \degreeBoundTreedepthCase(|N_G(C)|) \cdot \degreeBoundCaterpillarCase$ neighbors in $C$, the pigeonhole principle tells us that there exists some node $\ell \in V(F)$ such that  $|N_G(m) \cap Y_\ell| \geq \degreeBoundCaterpillarCase$.
    Clearly this node $\ell$ must be a leaf of tree $F$, and therefore \refline{rule:reduce_modulator_degree_fix_leaf} is well-defined.
    We set $H = G[Y_\ell]$, and since $\mathcal{F}$ is a $\pwEtaBoundClass[1]$-elimination tree of $C$ we have that $H$ is a caterpillar.

    Moreover, $N_G(H) \cap Y \subseteq N_G(C) \cap Y$, and therefore $|N_G(H) \cap Y| \leq 2(\eta + 1)$.
    Additionally, $N_G(H) \cap V(C) \subseteq \anc_F(\ell) \setminus \{\ell\}$, and therefore $|N_G(H) \cap V(C)| \leq \beta$.
    This yields $|N_G(H)| \leq |M| + 2(\eta + 1) + \beta$.
    Moreover, we have that $m$ has at least $\degreeBoundCaterpillarCase = \BoundCaterpillarOverall + 1 >  \BoundCaterpillar[|N_G(H)|]$ neighbors in $H$.
    Therefore, we have $|V(H)| > \BoundCaterpillar[|N_G(H)|]$.
    By \cref{thm:caterpillar_case} we have that running algorithm \nameref{rule:degree_reduction_caterpillar} on input $\big((G,k,M),H \big)$ returns an equivalent instance $(G',k,M)$ where $G'$ is strictly smaller than $G$.

    We can see that, if the algorithm execution reaches \refline{rule:reduce_modulator_degree_compute_polishing_set}, then the graph $G'$ of the output instance is a proper minor of $G$.
\end{proof}

It remains to prove the two key lemmas used for the safety proof, namely \cref{thm:treedepth_case,thm:caterpillar_case}.
We prove \cref{thm:treedepth_case} in \cref{subsec:bounding_degree_in_elim_forest}, and \cref{thm:caterpillar_case} in \cref{subsec:bounding_degree_caterpillars}.

\subsection{Bounding the number of \texorpdfstring{$\FF$}{F}-neighbors}
\label{subsec:bounding_degree_in_elim_forest}
In this section, the goal is as follows.
Let $C$ be a connected component of $G - (M \cup Y)$ (where $Y$ is a polishing set)
and $\FF$ be a nice $\pwEtaBoundClass[1]$-elimination tree of $C$ of depth at most $\beta$.
We want to ensure that, for each vertex $m \in M$,
there is only a bounded number of nodes $n$ in the elimination tree $\FF$
where $m$ has a neighbor in the bag $Y_n$.
That is, we want to bound the number of $\FF$-neighbors of $m$.

We note that this does not imply that $m$ has a bounded number of neighbors in $C$.
A bag $Y_n$, corresponding to a leaf $n$,
may induce a large graph of pathwidth $1$ in $G$.
Then $m$ possibly has only a single $\FF$-neighbor, namely $n$,
however, in the actual graph $G$,
$m$ may neighbor many vertices in this leaf bag $Y_n$.

Our reduction rule actually works for a problem that is more general than \pwEtaDelElimGeneral{}.
Concretely, it works assuming that $\FF$ is an elimination tree to \emph{some} graph class, but the concrete graph class (the graphs of pathwidth $1$ for $\pwEtaBoundClass[1]$-elimination trees) does not matter.

\subsubsection{The Key Ingredients}

Let us introduce some key ingredients of our reduction rule.
We begin with a relatively simple observation about
rooted trees $T$ of depth at most $\beta$.
Given sets $X$ and $Y$ of certain size,
we can find many siblings
such that each of them is the root of a subtree containing a vertex from $X$
but none of the siblings is from $Y$.
We need that $Y$ is not too large
and $X$ contains many vertices of some depth $d_X$.

\begin{lemma}
    \label{thm:tree_lemma_to_find_siblings}
    Define $f_{\ref{thm:tree_lemma_to_find_siblings}}(c,x) = (\beta+1) \cdot c^{\beta+1} \cdot x$.
    Let $T$ be a rooted tree of depth at most $\beta$,
    and $X \subseteq V(T)$ be a subset of the vertices in depth $d_X$ of size at least
    $f_{\ref{thm:tree_lemma_to_find_siblings}}(c,x)$, for some integers $c \geq 2, x \geq 1$.
    Moreover, let $Y \subseteq V(T)$ be a vertex set containing at most $c \cdot x$ vertices of each depth.
    Then, there is a vertex $p \in V(T)$ that has $c$ distinct children $w_1,\dots,w_c \notin Y$  and for each $w_i$ we have $V(T_{w_i}) \cap X \neq \emptyset$.
\end{lemma}
\begin{proof}
    We prove the following claim by induction on $d$:
    if there is a subset $Z \subseteq V(T)$ at depth $d$ of at least $(d + 1) \cdot c^{d + 1} \cdot x$ vertices, then there is a vertex $p \in V(T)$ with $c$ distinct children $w_1,\dots,w_c \notin Y$ such that for each $w_i$ we have $V(T_{w_i}) \cap Z \neq \emptyset$.

    \proofsubparagraph*{Base case.}
    For the base case $d = 0$, the claim is vacuously true.

    \proofsubparagraph*{Induction step.}
    Assuming the claim holds at every depth $d' < d$, we show that it also holds for the set $Z$ at depth $d$.
    Since $|Z| \geq (d + 1) \cdot c^{d+1} \cdot x$, the size of $Z_0 = Z \setminus Y$ is at least $d \cdot c^{d+1} \cdot x$.
    If some $c$ vertices of $Z_0$ have the same parent, we are done.
    Otherwise, let $Z'$ be the set of parents of $Z_0$, and $d'$ the depth of the vertices in $Z'$, that is $d' = d - 1$.
    Then, $|Z'| \geq \frac{|Z_0|}{c} \geq d \cdot c^{d} \cdot x = (d' +1) \cdot c^{d'+1} \cdot x$.
    Applying the claim to the set $Z'$ yields a vertex $p \in V(T)$ with $c$ distinct children $w_1,\dots,w_c$, such that each $w_i \notin Y$ and $V(T_{w_i}) \cap Z' \neq \emptyset$.
    This parent $p$ and children $w_1,\dots,w_c$ also show that the claim holds for $Z$, since $V(T_{w_i}) \cap Z' \neq \emptyset$ implies $V(T_{w_i}) \cap Z \neq \emptyset$ by our choice of $Z'$.

    \proofsubparagraph*{Finishing the proof.}
    Since $d_X \leq \beta$, applying the claim to $Z = X$ at depth $d_X$ yields the lemma.
\end{proof}

We intend to apply this lemma with $T$ as (the tree of) an $\HH$-elimination tree
of a connected subgraph of the input graph.
As part of our reduction rule (Algorithm \nameref{rule:degree_reduction_elimination_forest})
we introduce the following values depending on
$\beta$ (the depth of the $\HH$-elimination tree we will consider),
$\eta$ and $|M|$.

\begin{definition}[Important Bounds for the Degree Reduction in Elimination Forests]
    \label{def:important_bounds}
    We set
    \begin{align*}
        \boundOnS           & = 3(\eta+1),                                                         \\
        \numbChildren       & = f_{\ref{thm:find_stable_graph}}(\boundOnS+\beta) + \boundOnS + 1   \\
        \markedNumbChildren & = 2^\beta \cdot \numbChildren,                                       \\
        \markedNumb         & = f_{\ref{thm:tree_lemma_to_find_siblings}}(\markedNumbChildren,1) =
        (\beta + 1) \cdot \left( \markedNumbChildren \right)^{\beta +1}.
    \end{align*}
    Moreover, we define the following functions $\nat \to \nat$
    \begin{align*}
        \numNeighborsInDepth(x)      & = f_{\ref{thm:tree_lemma_to_find_siblings}}
        \big(\markedNumb, \tbinom{x}{\leq 2}\big) = (\beta + 1) \cdot \left( \markedNumb \right)^{\beta + 1} \cdot \tbinom{x}{\leq 2}, \\
        \degreeBoundTreedepthCase(x) & = (\beta + 1) \cdot \numNeighborsInDepth(x).
    \end{align*}
\end{definition}

We note that $\boundOnS, \numbChildren, \markedNumbChildren$ and $\markedNumb$ are constants (depending only on the constants $\eta, \beta$).
Further, $\numNeighborsInDepth,\degreeBoundTreedepthCase$ are in $O(x^2)$.
The value $\boundOnS$ is an upper bound of the size of the intersection of a solution $S$ with $V(C)$.
Further,
$\markedNumb$ is the number of vertices we will mark for each $X \in \binom{N_G(C)}{2} \cup \binom{N_G(C)}{1}$.
Beyond these observations, the concrete values are not too important for now.
We will recall them when they become relevant, and then it will also become evident why we have chosen them this way.

Now, recall that for an $\HH$-elimination forest $\FF = (F,\{Y_n\}_{n \in V(F)})$ of a graph $G$ and node $n \in V(F)$,
we write $V^\FF_n$ for the union of all bags in the subtree rooted at the node $n$,
which is $V^\FF_n = \bigcup_{n' \in F_n} Y_{n'}$.
Next, we define the notion of \emph{robust} nodes, which connects $\HH$-elimination trees with path decompositions and the \pwEtaDelElimGeneral{} problem.

\begin{definition}[Robust Nodes]
    \label{def:marvelous_nodes}
    Let $C$ be an induced subgraph of some graph $W$, and $\mathcal{F} = (F,\{Y_n\}_{n \in V(F)})$ be an $\HH$-elimination tree of $C$ of depth at most $\beta$, for some graph class~$\HH$.
    Let $\PP$ be a path decomposition of $W-S$ for some set $S \subseteq V(W)$.
    A node $n \in V(F)$ is \emph{robust} (with regard to $W$, $\FF$, $C$, $S$, $\PP$) if,
    \begin{bracketenumerate}
        \item $V_n^\FF \cap S = \emptyset$, and
        \item $V_n^\FF$ is stable in $W-S$, $\PP$.
    \end{bracketenumerate}
\end{definition}

Now, we show that, in certain cases, robust nodes are guaranteed to exist.
\begin{lemma}
    \label{thm:marvelous_siblings}
    Let $C$ be an induced subgraph of some graph $W$, and $\mathcal{F} = (F,\{Y_n\}_{n \in V(F)})$ be a nice $\HH$-elimination tree of $C$ of depth at most $\beta$, for some graph class $\HH$.
    Let $\PP$ be a path decomposition of width at most $\eta$ of $W-S$ for some set $S \subseteq V(W)$.
    Assume that there are sibling nodes $u_1,\dots,u_{\numbChildren}$ in $V(F)$
    that have the same ancestor-type in $\FF$, $C$.
    Suppose $|S \cap V(C)| \leq \boundOnS$ and $|N_{W}(C) \setminus S|\leq\boundOnS$.
    Then, there exists a robust node $u_i$ for some $i \in \range[2]{\numbChildren}$.
\end{lemma}
\begin{proof}
    Let $U$
    be the set of nodes $u \in \{u_2,\dots,u_{\numbChildren}\}$
    such that $V_{u}^\FF \cap S = \emptyset$.
    Since $|S \cap V(C)|\leq \boundOnS$, we have that $|U|\geq \numbChildren - 1 - \boundOnS = f_{\ref{thm:find_stable_graph}}(\boundOnS +\beta)$.
    Hence, every node $u \in U$ satisfies the first condition of a robust node.
    Let $T$ be the ancestor-type of $u_1,\dots,u_{\numbChildren}$.
    To prove the second condition,
    we aim to show that, for some node $u \in U$, the set $V_u^\FF$ is $((T \cup N_{W}(C)) \setminus S)$-stable in $\PP$.
    Then, since $N_{W - S}(V_u^\FF)$ is a subset of $(T \cup N_{W}(C)) \setminus S$,
    we indeed have that $V_u^\FF$ is stable in $\PP$ and graph $W-S$.

    Since the vertices $U$ are siblings in $\FF$, we know that $V_u^\FF$ for $u \in U$ induce pairwise vertex-disjoint connected subgraphs of $W$.
    Further, $|(T \cup N_{W}(C)) \setminus S| \leq \beta + \boundOnS$
    because $F$ has depth at most $\beta$ and since we require $|N_{W}(C) \setminus S|\leq\boundOnS$.
    As such, \cref{thm:find_stable_graph} applies on graph $W-S$ with subgraphs $W[V_u^\FF]$ for $u \in U$,
    and set $Y = (T \cup N_{W}(C)) \setminus S$.
    Here, we use that $S$ does not intersect $V_u^\FF$ for any $u \in U$.
    As a result, we obtain a node $u \in U$ such that $V_u^\FF$ is $(T \cup N_{W}(C)) \setminus S$-stable in~$\PP$ and graph $W-S$,
    which concludes the proof.
\end{proof}

In the context of elimination trees,
we can state \cref{thm:embedding_with_twinning_graph}
in the following, more suitable form.

\begin{corollary}[of \cref{thm:embedding_with_twinning_graph}]
    \label{thm:embedding_with_twinning_graph:in:F}
    Let $\FF=(F,\{Y_n\}_{n \in V(F)})$ be an $\HH$-elimination tree of an induced connected subgraph $C$ of $W$,
    for some graph $W$ and some graph class $\HH$.
    Let $\PP$ be a path decomposition of $W-S$ of width at most $\eta$, for some set $S \subseteq V(W)$.
    Let nodes $z,w \in V(F)$ and induced subgraph $R \subseteq W[V_z^\FF \setminus S]$ be such that
    \begin{bracketenumerate}
        \item
        $z,w \in V(F)$ are not in an ancestor-descendant relationship in $F$, and
        \item
        $\pw(W[V_z^\FF]) \leq \pw(W[V_{w}^\FF])$, and
        \item
        $w$ is robust with regard to $W$, $\FF$, $C$, $S$, $\PP$, and
        \item
        there is a bag $X$ of $\PP$ where $|X \cap V_w^\FF| \geq \pw(W[V_{w}^\FF])+1$
        and  $N_{W - S}(R) \subseteq X$.
    \end{bracketenumerate}
    Then, there exists a path decomposition $\hat \PP$ of $W-S$ of width at most $\eta$
    such that each bag $\hat X$ of $\hat \PP$ with $\hat X \cap V(R) \neq \emptyset$ contains all vertices of $X \setminus (V(R) \cup V_w^\FF)$.
\end{corollary}
\begin{proof}
    Note that, because $w$ is robust, we have $V_w^\FF \cap S = \emptyset$, and that $V_w^\FF$ is stable in $\PP$ and $W - S$.
    Our goal is to apply \cref{thm:embedding_with_twinning_graph} on the subgraphs $H_1 = W[V_w^\FF \setminus S] = W[V_w^\FF]$, $H_2 = R - S = R$, and graph $W - S$.
    Since $z$ and $w$ are not in an ancestor-descendant relationship in $F$,
    the graphs $H_1$ and $W[V_z^\FF]$ are vertex-disjoint and
    without any edge between them.
    Because $V(H_2) \subseteq V_z^\FF \setminus S$, the same holds for $H_1$ and $H_2$.
    Moreover, we have that $\pw(H_1) = \pw(W[V_w^\FF]) \geq \pw(W[V_z^\FF]) \geq \pw(R) = \pw(H_2)$ by assumption and the fact that $R$ is an induced subgraph of $W[V_z^\FF]$.

    By assumption $X$ satisfies condition 1 of \cref{thm:embedding_with_twinning_graph}.
    It remains to show condition 2., that $N_{W - S}(H_1) \subseteq X$ and $N_{W - S}(R) \subseteq X$.
    The latter, $N_{W - S}(R) \subseteq X$ holds by assumption.
    The former, $N_{W - S}(H_1) \subseteq X$, follows from the fact that $V(H_1) = V_w^\FF$ is stable in $\PP$, $W-S$, which is itself a consequence of $w$ being robust.
    Then, \cref{thm:embedding_with_twinning_graph} applies, which concludes the proof.
\end{proof}

\subsubsection{The Reduction Rule}

\begin{algorithm}[t]
    \caption{\algofont{ReduceDegreeTree}}
    \label{rule:degree_reduction_elimination_forest}
    \Input{Instance $(G,k,M)$ of \pwEtaDeletionEta{},
    a near-protrusion $C \subseteq G - M$ of $(G,k)$, a nice $\HH$-elimination tree $\FF = (F,\{Y_n\}_{n \in V(F)})$ of $C$ of depth at most $\beta$ for an arbitrary graph class $\HH$}
    \If{each $m \in N_G(C)$ has at most $\degreeBoundTreedepthCase(|N_G(C)|)$ $\FF$-neighbors\label{rule:degree_reduction_elimination_forest:start}}{
        \Return{$(G,k,M)$} \label{rule:degree_reduction_elimination_forest:check_input}
    }
    $m \leftarrow$ vertex of $N_G(C)$ with at least $\degreeBoundTreedepthCase(|N_G(C)|)$ neighbors in $\mathcal{F}$ \label{rule:degree_reduction_elimination_forest:fix_m}\;
    \For{$X \in \binom{N_G(C)}{2} \cup \binom{N_G(C)}{1}$ and $d \in \range[0]{\depth(F)}$\label{rule:degree_reduction_elimination_forest:marking_start}}{
    $\mathcal{T}_d^X \leftarrow \{n \in V(F) \mid \depth_F(n) = d \text{ and } X \subseteq N_G(V_n^{\FF})\}$\;
    \uIf{$|\mathcal{T}_d^X| \leq \markedNumb$}{
        mark all vertices of $\mathcal{T}_d^X$\;
    }
    \Else{
        $z_1,\dots,z_\ell \leftarrow$ vertices of $\mathcal{T}_d^X$ ordered by nonincreasing pathwidth of $G[V_{z_i}^{\FF}]$ \label{rule:degree_reduction_elimination_forest:marking_order_by_pathwidth}\;
        mark the first $\markedNumb$ vertices of $z_1,\dots,z_\ell$ \label{rule:degree_reduction_elimination_forest:marking_scheme}\;
    }
    }\label{rule:degree_reduction_elimination_forest:marking_end}
    $d^* \leftarrow$ number in $\range[0]{\beta}$ such that $m$ has at least $\numNeighborsInDepth(|N_G(C)|)$ $\FF$-neighbors at depth~$d^*$ \label{rule:degree_reduction_elimination_forest:find_depth}\;
    $N_{d^*} \leftarrow$ $\FF$-neighbors of $m$ at depth $d^*$\;
    $p \leftarrow$ node of $F$ with $\numbChildren$ unmarked children $c_1',\dots,c_{\numbChildren}'$ in $F$ that have the same ancestor-type in $\mathcal{F}$, $C$ and for each $c_i'$ we have $V_{c_i'}^{\FF} \cap N_{d^*} \neq \emptyset$ \label{rule:degree_reduction_elimination_forest:find_parent}\;
    $c_1,\dots,c_{\numbChildren} \leftarrow$ vertices $c_1',\dots,c_{\numbChildren}'$ ordered by nondecreasing pathwidth of $G[V_{c_i'}^{\FF}]$ \label{rule:degree_reduction_elimination_forest:order_children_by_pathwidth}\;
    $v \leftarrow$ neighbor of $m$ in $V_{c_1}^{\FF}$ \label{rule:degree_reduction_elimination_forest:fix_vertex}\;
    \Return{$(G - mv,k,M)$} \label{rule:degree_reduction_elimination_forest:output_instance}
\end{algorithm}

Our reduction rule is \nameref{rule:degree_reduction_elimination_forest} stated as \cref{rule:degree_reduction_elimination_forest}.
We now show that it can be executed in polynomial time.

\begin{lemma}
    \label{thm:treedepth_case_rule_valid}
    \namecref{rule:degree_reduction_elimination_forest} \nameref{rule:degree_reduction_elimination_forest} is well-defined and runs in polynomial time.
\end{lemma}
\begin{proof}
    \reflines{rule:degree_reduction_elimination_forest:start}{rule:degree_reduction_elimination_forest:fix_m} are clear.
    The marking scheme applied in \reflines{rule:degree_reduction_elimination_forest:marking_start}{rule:degree_reduction_elimination_forest:marking_end} can be executed in polynomial-time since, in particular, the pathwidth of $G - M$ is at most $\eta$, and thus we can use \cref{thm:path_decomposition_algorithm} to compute the pathwidth in \refline{rule:degree_reduction_elimination_forest:marking_order_by_pathwidth}.

    For \refline{rule:degree_reduction_elimination_forest:find_depth}, observe
    that such a depth $d^*$ exists by the pigeonhole principle since
    $\degreeBoundTreedepthCase(|N_G(C)|) = (\beta + 1) \cdot \numNeighborsInDepth(|N_G(C)|)$
    and the depth of $F$ is at most $\beta$.

    For \refline{rule:degree_reduction_elimination_forest:find_parent}, let us use \cref{thm:tree_lemma_to_find_siblings} applied on the tree $F$, $X$ being the set $N_{d^*}$, and $Y$ being the set of vertices marked in \refline{rule:degree_reduction_elimination_forest:marking_scheme}.
    Then, we have that the corresponding value $x$ of the lemma is $\binom{|N_G(C)|}{\leq 2}$, $c = \markedNumb$, and $d_X \leq \beta$.
    We have set $\numNeighborsInDepth(|N_G(C)|) = f_{\ref{thm:tree_lemma_to_find_siblings}}( \markedNumb, \binom{|N_G(C)|}{\leq 2})$ so that the lemma gives us a vertex $p$
    with $\markedNumb \geq \markedNumbChildren = 2^\beta \cdot \numbChildren$ children $w$ in $F$, such that $V^{\mathcal{F}}_{w}$ contains a vertex of $N_{d^*}$ and is unmarked.
    We then find the desired set of $\numbChildren$ children of $p$ by using the pigeonhole principle and the fact that there are at most $2^\beta$ possible ancestor-types for children of $p$.
    For the output instance, notice that since $C \subseteq G - M$ by a pre-condition of the algorithm, we have that the output instance is a valid instance of \pwEtaDeletionEta{}.
    The remaining lines are clear.
\end{proof}

It remains to show that this reduction rule is safe.
\begin{lemma}
    \label{thm:treedepth_case_equivalent}
    Let $(G,k,M),C,\mathcal{F} = (F,\{Y_n\}_{n \in V(F)})$ be the input to algorithm \nameref{rule:degree_reduction_elimination_forest}.
    Then, the output instance $(\hat G,k,M)$ is equivalent to $(G,k,M)$.
\end{lemma}
\begin{proof}
    The forward direction of equivalence is clear since $\hat G$ is a subgraph of $G$.

    For the backward direction, assume that $(\hat G,k,M)$ is a yes-instance, and
    let $S \subseteq V(\hat G)$ be a minimum-size set
    such that $\pw(\hat G - S) \leq \eta$.
    For convenience, let $G' = \hat G - S$.
    Let $\mathcal{P} = X_1,\dots,X_t$ be a path decomposition of $G'$ of width at most $\eta$.
    If $\hat G = G$ we are done.
    Otherwise, $\hat G = G - mv$ for some edge $mv \in E(G)$ where $m \in N_G(C)$ and $v \in V(C)$.
    If $m \in S$ or $v \in S$ we are again done
    since $G' = G - S$ in that case,
    and hence $\PP$ is also a path decomposition of $G - S$.
    Similarly, we are done if there is a bag $X_i$ that contains $m$ and $v$, as then again $\mathcal{P}$ is also a path decomposition of $G - S$.

    It remains to consider that no bag $X_i$ contains both $m$ and $v$.
    Without loss of generality, we may assume that all bags that contain vertex $m$ have an index that is smaller than the index of any bag that contains $v$.
    For convenience, in this proof we drop the superscripts $\mathcal{F}$ and $\mathcal{P}$ when possible.
    Let $c_1,\dots,c_{\numbChildren}$ be the children computed by \nameref{rule:degree_reduction_elimination_forest} in \refline{rule:degree_reduction_elimination_forest:order_children_by_pathwidth}.
    Hence, $c_1,\dots,c_{\numbChildren}$ have the same parent $p$ in $F$, and
    the same ancestor-type.
    Recall that $v$ is in the set $V_{c_1}$.
    For convenience, let us denote $z = c_1$, and $G'_z = G'[V_z \setminus S]$.

    We now define $M_z = (N_{G'}( V_z \setminus S ) \cap N_G(C)) \cup \{m\}$,
    that is, $M_z$ is the neighborhood of $V_z \setminus S$ in $G'$ restricted to $N_G(C)$, as well as vertex $m$.
    Observe that $M_z = N_{G - S}(V_z \setminus S) \cap N_G(C)$.
    Next, we show that $M_z$ essentially forms a clique of $G'$.
    For this purpose, we fix $d_z = \depth_F(z)$.
    \begin{claim}
        \label{claim:M_z_is_clique}
        There is at least one bag containing all vertices of $M_z$.
    \end{claim}
    \begin{claimproof}
        In the marking scheme (Lines \ref{rule:degree_reduction_elimination_forest:marking_start} to \ref{rule:degree_reduction_elimination_forest:marking_end}) each pair $\{m_1,m_2\} \in \binom{M_z}{2}$ marked $\markedNumb \geq 9\eta + 11$ nodes $n$ of $F$ at depth $d_z$, such that $\{m_1,m_2\} \subseteq N_G(V_n)$.
        We note that $|S \cap V(C)| \leq \boundOnS = 3\eta + 3$
        by \cref{thm:near_protrusions_yield_small_intersection_and_unsel_neighborhood}
        and as $C$ is a near-protrusion.
        Hence, at least $\eta + 2$ of them are such that $S$ selects no vertex of $V_n$.
        Therefore, there are at least $\eta + 2$ internally vertex-disjoint $m_1$-$m_2$ paths in $G'$.
        It follows from \cref{obs:disjoint_paths_yield_clique} that there is a bag containing all of~$M_z$.
    \end{claimproof}

    Observe that by \cref{claim:M_z_is_clique} the indices $\leftBagAll(M_z)$ and $\rightBagAll(M_z)$ are defined.
    Now, we find a vertex with similar properties as vertex $z$.

    \begin{claim}
        \label{claim:find_robust_node_1}
        There is a robust node $z' \in \{c_2,\dots,c_{\numbChildren}\}$ with regard to $\hat G, \FF, C, S, \PP$.
        Since $z' \in \{c_2,\dots,c_{\numbChildren}\}$, we further have
        \begin{itemize}
            \item
                  $z,z'$ have the same ancestor-type in $\FF$, $C$, and
            \item
                  $\pw(G[V_z]) \leq \pw(G[V_{z'}])$, and
            \item $m \in N_{G'}(V_{z'})$.
        \end{itemize}
    \end{claim}
    \begin{claimproof}
        We want to apply \cref{thm:marvelous_siblings} on $C$, $W = \hat G$, $\FF$, $\PP$, $S$, and the siblings $\{c_1,\dots,c_{\numbChildren}\}$.
        Observe that indeed $C$ is an induced subgraph of $W$ with $\HH$-elimination tree $\FF$ of depth at most $\beta$.
        Moreover, $\PP$ is a path decomposition of $W - S$ of width at most $\eta$.
        Also, $c_1,\dots,c_{\numbChildren}$ are siblings with the same ancestor-type in $\FF$, $C$.
        Finally, as $C$ is a near-protrusion, by \cref{thm:near_protrusions_yield_small_intersection_and_unsel_neighborhood},
        $|S \cap V(C)| \leq 3(\eta+1) = \boundOnS$ and $|N_W(C) \setminus S| \leq \boundOnS$.
        Therefore, all conditions of \cref{thm:marvelous_siblings} are fulfilled, and we find the robust node $z' \in \{c_2,\dots,c_{\numbChildren}\}$.
    \end{claimproof}

    Next, we exploit the marking scheme to find yet another node that is similar to $z$.
    \begin{claim}
        \label{claim:find_robust_node_2}
        There is a node $\hat z \in V(F)$ such that
        \begin{bracketenumerate}
            \item $\hat z$ is robust with regard to $\hat G$, $\FF$, $C$, $S$, $\PP$,
            \item $\pw(G[V_{z}]) \leq \pw(G[V_{\hat z}])$,
            \item each bag $X_i$ with $X_i \cap V_{\hat z} \neq \emptyset$ contains all vertices of $M_z$, and
            \item $z$ and $\hat z$ are \emph{not} in an ancestor-descendant relationship in $F$.
        \end{bracketenumerate}
    \end{claim}
    \begin{claimproof}
        Let $m_1$ be a vertex of $M_z$ that is introduced in bag $X_{\leftBagAll(M_z)}$.
        Similarly, let $m_2$ be a vertex of $M_z$ that is forgotten in bag $X_{\rightBagAll(M_z) + 1}$ (note that this bag exists, since $m \in M_z$ and $\leftBagAny(\{v\}) > \rightBagAny(\{m\})$).
        Possibly, we have that $m_1=m_2$.
        In any case, we have that the body of the for-loop of \refline{rule:degree_reduction_elimination_forest:marking_start}
        is invoked for $X=\{m_1,m_2\}$.
        Since $z=c_1$ is an unmarked node, in particular $X$ does not mark $z$.
        Thus, $X$ \emph{does} mark at least $\markedNumb$ nodes $w_1,\dots,w_{r}$ in depth $d_z$.
        For each node $w \in \{w_1,\dots,w_r\}$,
        we have $\{m_1,m_2\} \subseteq N_G(V_{w})$, $\depth_F(w) = d_z$,
        and $\pw(G[V_{w}]) \geq \pw(G[V_z])$.

        Recall that $\markedNumb = f_{\ref{thm:tree_lemma_to_find_siblings}}(\markedNumbChildren,1)$.
        Then, \cref{thm:tree_lemma_to_find_siblings} (using $X$ as the set of these marked vertices and $Y = \emptyset$) yields a node $\hat p \in V(F)$
        with distinct children $\hat w_1,\dots,\hat w_{\markedNumbChildren}$ in~$F$,
        and, for each node $\hat w_i \in \{\hat w_1,\dots, \hat w_{\markedNumbChildren}\}$, we have that $V_{\hat w_i}$ contains a vertex marked by $\{m_1,m_2\}$ at depth~$d_z$.

        Recall that $\boundOnS = 3(\eta+1)$.
        Since $C$ is a near-protrusion, by \cref{thm:near_protrusions_yield_small_intersection_and_unsel_neighborhood},
        $|N_{G'}(C - S)| \leq \boundOnS$
        and $|S \cap V(C)| \leq \boundOnS$.
        Recall that $\markedNumbChildren \geq f_{\ref{thm:find_stable_graph}}(\boundOnS+\beta) + \boundOnS + 1$.
        Then, since the sets $V_{\hat w_1},\dots,V_{\hat w_\markedNumbChildren}$ are pairwise vertex-disjoint, at least $f_{\ref{thm:find_stable_graph}}(\boundOnS+\beta) + 1$ of these sets do not intersect $S$.
        Consider such a $\hat w_i$ such that $V_{\hat w_i} \cap S = \emptyset$.
        Let $A'$ be the set of ancestors of $\hat p$ in $F$, clearly $|A'| \leq \beta$ since $\hat p$ has children in $F$, and the depth of $F$ is at most $\beta$.
        Notably, the neighborhood $N_{G'}(V_{\hat w_i})$
        is a subset of $(A' \setminus S) \cup N_{G'}(C - S)$,
        which has size at most $\beta + \boundOnS$.
        Moreover, $G'[V_{\hat w_i}]$ is a connected subgraph of $G'$
        since $\FF$ is a nice elimination tree.

        Then, by \cref{thm:find_stable_graph},
        there is a vertex $\hat z \in \{{\hat w_1},\dots,{\hat w_\markedNumbChildren} \}$
        such that $V_{\hat z}$ is $(A' \setminus S) \cup N_{G'}(C - S)$-stable and $S$ does not intersect~$V_{\hat z}$.
        Since $N_{G'}(V_{\hat z}) \subseteq (A' \setminus S) \cup N_{G'}(C - S)$, we have that
        $V_{\hat z}$ is stable in $G'$, $\PP$.
        We conclude that $\hat z$ is robust with regard to $\hat G,\FF,C,S, \PP$, establishing property (1).
        It remains to show the properties (2), (3) and (4) for node~$\hat z$.
        \begin{bracketenumerate}\setcounter{enumi}{1}
            \item
            To see property (2), that $\pw(G[V_{z}]) \leq \pw(G[V_{\hat z}])$,
            we recall that $\hat z$ is an ancestor of at least one node $w \in \{w_1,\dots,w_{r}\}$ in $\FF$.
            Then, $\pw(G[V_{\hat z}]) \geq \pw(G[V_w]) \geq \pw(G[V_z])$.
            \item
            Regarding the third property,
            we note that $\{m_1,m_2\} \subseteq N_{G'}(V_{\hat z})$.
            Since $V_{\hat z}$ is stable, every bag $X_i$ with $X_i \cap V_{\hat z}\neq\emptyset$ contains $\{m_1,m_2\}$.
            By our choice of $m_1,m_2$, every bag containing $\{m_1,m_2\}$ contains all of $M_z$.
            Thus, we have property (3), that each bag $X_i$ with $X_i \cap V_{\hat z} \neq \emptyset$,
            contains all of $M_z$.

            \item
            Finally, to show property (4),
            assume for the sake of contradiction, that $z$ and $\hat z$ \emph{are} in ancestor-descendant relationship in $F$.
            Since $\hat z$ is an ancestor of at least one of $w_1,\dots,w_r$, of depth $d_z$,
            node $\hat{z}$ must be an ancestor of $z$, and as such $V_z \subseteq V_{\hat z}$.
            By property (3),
            each bag $X_i$ with $X_i \cap V_{\hat z} \neq \emptyset$ contains all vertices of $M_z$,
            and, therefore also $m$.
            Since $V_z \subseteq V_{\hat z}$,
            we have that each bag $X_i$ with $X_i \cap V_z \neq \emptyset$ contains $m$.
            Since $v \in V_z$, this contradicts our assumption that no bag contains $v$ and $m$.
        \end{bracketenumerate}
        This concludes the proof.
    \end{claimproof}

    Let nodes $z'$ and $\hat z$ be as \cref{claim:find_robust_node_1} and \cref{claim:find_robust_node_2}.
    For convenience, let us write $G'_{z'} = G'[V_{z'}] = G[V_{z'}]$ and $G'_{\hat z} = G'[V_{\hat z}] = G[V_{\hat z}]$.

    We define the \emph{rest} $R_i$ relative to some bag~$X_i$.
    For an integer $i \in \range[1]{\leftBagAny(\{v\})-1}$,
    the \emph{rest}~$R_i$ is the connected component of the graph $G'[V_z \setminus (S \cup \bigcup_{j \in [i]} X_j)]$ that contains vertex~$v$.
    By requiring $i < \leftBagAny(\{v\})$, there indeed is such a connected component.

    Our goal is to find some bag $X_i$ that has sufficient space to embed a path decomposition of $R_i$ in copies of $X_i$, such that $X_i$ also contains vertex $m$.
    We distinguish whether $\leftBagAll(M_z) \leq \leftBagAny(V_{z'})$ or not.
    We note that $\leftBagAny(V_{z'})$ is well-defined since $V_{z'} \subseteq V(G')$.

    \begin{claim}
        \label{claim:finish_proof_1}
        If $\leftBagAll(M_z) \leq \leftBagAny(V_{z'})$, the instance $(G,k,M)$ is a yes-instance.
    \end{claim}
    \begin{claimproof}
        We fix an index $i$ between $\leftBagAny(V_{z'})$ and $\rightBagAny(V_{z'})$ such that $X_i$ contains $\pw(G'_{z'}) + 1$ vertices of $V_{z'}$.
        Such an index $i$ exists because $S \cap V_{z'} = \emptyset$ by \cref{claim:find_robust_node_1} and the definition of robustness.
        Note that, by \cref{claim:find_robust_node_1} we have that $m \in N_{G'}(V_{z'})$, and as $z'$ is robust, $V_{z'}$ is stable which implies that any bag containing a vertex of $V_{z'}$ contains $m$.
        Hence, $m \in X_i$.
        Then, since any bag that contains $v$ has a higher index than any bag that contains $m$, we have $i \in [\leftBagAny(\{v\})-1]$.

        Our goal is to apply \cref{thm:embedding_with_twinning_graph:in:F}
        on the $\HH$-elimination tree $\FF$ of $C$, graph $W = \hat G$, set $S$, nodes $z$ and $z'$, induced subgraph $R = R_i$ and bag $X = X_i$ of the path decomposition $\PP$.
        Assuming all preconditions are met,
        we obtain a path decomposition $\hat \PP$ of $\hat G - S = G'$ of width at most $\eta$
        such that each bag $\hat X$ of $\hat \PP$ with $\hat X \cap V(R_i) \neq \emptyset$
        contains all vertices $X_i \setminus (V(R_i) \cup V_{z'})$.
        We have already argued that $m \in X_i$, and clearly $m \notin V(R_i) \cup V_{z'}$.
        Then, since $v \in V(R_i)$, each bag of $\hat{\mathcal{P}}$ that contains vertex $v$ also contains~$m$.
        Hence, $\hat{\mathcal{P}}$ is also a path decomposition of $G - S$, showing that $(G,k,M)$ is a yes-instance.

        Therefore, it remains to show the pre-condition 1., 2., 3., and 4. of \cref{thm:embedding_with_twinning_graph:in:F}.
        The first three conditions hold easily.
        \begin{bracketenumerate}
            \item
            Nodes $z,z'$ are not in an ancestor-descendant relationship in $F$, as they are siblings in $F$.
            \item
            $\pw(W[V_z]) \leq \pw(W[V_{z'}])$ by our choice of $z = c_1$.
            \item
            Node $z'$ is robust by \cref{claim:find_robust_node_1}.
        \end{bracketenumerate}
        Thus, it remains to show precondition (4),
        which is that the bag $X_i$ contains every vertex $w \in N_{G'}(R_i)$.
        The neighborhood $N_{G'}(R_i)$ has the three parts, the neighbors in $M_z$, the neighbors in $V(C - V_z)$, and the neighbors in $V_z$.
        We consider each of these neighborhoods separately.
        \begin{itemize}
            \item
                  Consider that $w \in N_{G'}(R_i) \cap M_z$.
                  We recall that $\leftBagAll(M_z) \leq i$.
                  Thus, $w$ is introduced in some bag $X_j$ with $j \leq i$.
                  Because $w \in N_{G'}(R_i)$, there is some vertex $r \in V(R_i)$ that is a neighbor of $w$.
                  By definition of $R_i$, vertex $r$ is introduced in a bag $X_{j'}$ with $i < j'$.
                  Therefore, $w$ must be part of $X_i$.
            \item
                  Consider that $w \in N_{G'}(R_i) \cap V_z$.
                  By definition of $R_i$, such a vertex $w$ is again introduced in a bag $X_j$ with $j \leq i$.
                  Analogously, to the previous case, we obtain that $w \in X_i$.
            \item
                  Finally, consider that $w \in N_{G'}(R_i) \cap V(C - V_z)$.
                  That is, $w$ is an ancestor of $z$.
                  Since $z$ and $z'$ have the same ancestor-type, also $w \in N_{G'}(G'_{z'})$.
                  By the robustness of $z'$, we have that $G'_{z'}$ is stable.
                  Then, \cref{thm:stable_graph_bags_contain_neighbors} yields that any bag that contains a vertex of $V(G'_{z'}) = V_{z'}$ contains all of $N_{G'}(G'_{z'})$, and therefore in particular $w \in X_i$.
        \end{itemize}
        Thus, bag $X_i$ contains every vertex of $N_{G'}(R_i)$.
        This shows the final precondition for \cref{thm:embedding_with_twinning_graph:in:F},
        and concludes the proof.
    \end{claimproof}

    It remains to consider that $\leftBagAny(V_{{z'}}) < \leftBagAll(M_z)$.
    The previous approach, for the case $\leftBagAll(M_z) \leq \leftBagAny(V_{{z'}})$,
    fails because a bag with index $i$ between $\leftBagAny(V_{z'})$ and $\rightBagAny(V_{z'})$ might not contain all of $N_{G'}(R_i) \cap M_z$.
    This time, we use the node $\hat z$.

    \begin{claim}
        \label{claim:finish_proof_2}
        If $\leftBagAll(M_z) > \leftBagAny(V_{{z'}})$, the instance $(G,k,M)$ is a yes-instance.
    \end{claim}
    \begin{claimproof}
        Let $i$ be the index of a bag of $\PP$ that contains $\pw(G'_{\hat z}) + 1$ vertices of $V_{\hat z}$.
        Since $S \cap V(G'_{\hat z}) = \emptyset$, such an index exists.
        Note that bag $X_i$ contains a vertex of $V_{\hat z}$, and therefore, by \cref{claim:find_robust_node_2} bag $X_i$ contains vertex $m \in M_z$.
        Since all bags containing $v$ have a higher index than any bag containing $m$ this yields $i \in [\leftBagAny(\{v\})-1]$.

        Our goal is to apply \cref{thm:embedding_with_twinning_graph:in:F}
        on the $\HH$-elimination tree $\FF$ of $C$,
        graph $W = \hat G$, set~$S$, nodes $z$ and $\hat z$, induced subgraph $R = R_i$
        and bag $X = X_i$ of the path decomposition $\PP$.
        Assuming all preconditions are met,
        we obtain a path decomposition $\hat{\mathcal{P}}$ of $\hat G - S = G'$ of width at most $\eta$, such that each bag $\hat X$ of $\hat \PP$ that contains a vertex of $R_i$ contains all of $X_i \setminus (V(R_i) \cup V_{\hat z})$.
        Since $m \in X_i$ by \cref{claim:find_robust_node_2}, $m \notin V(R_i) \cup V_{\hat z}$, and $v \in V(R_i)$, some bag of $\hat{\PP}$ contains both $m$ and $v$.
        Hence, $\hat \PP$ is also a path decomposition of $G-S$,
        showing that $(G,k,M)$ is a yes-instance.

        Therefore, it remains to show the preconditions (1), (2), (3), and (4) of \cref{thm:embedding_with_twinning_graph:in:F}.
        The first three conditions hold easily.
        \begin{bracketenumerate}
            \item Nodes $z,\hat z$ are not in ancestor-descendant relationship, by \cref{claim:find_robust_node_2}.
            \item $\pw(\hat{G}[V_{z}]) \leq \pw(\hat{G}[V_{\hat z}])$,
            since $G[V_z] = \hat G[V_z]$ and $\hat G[V_{\hat z}] = G[V_{\hat z}]$ and because of \cref{claim:find_robust_node_2}.
            \item Node $\hat{z}$ is robust with regard to $W,\FF,C,S, \PP$, by \cref{claim:find_robust_node_2}.
        \end{bracketenumerate}
        Thus, it remains to show precondition (4),
        that the bag $X_i$ contains $\pw(G'_{\hat z}) + 1$ vertices of $V_{\hat z}$ and that $N_{G'}(R_i) \subseteq X_i$.
        By our choice of $X_i$, it contains $\pw(G'_{\hat z}) + 1$ vertices of $V_{\hat z}$.
        It remains to show that, for every $w \in N_{G'}(R_i)$, also $w \in X_i$.
        The neighborhood $N_{G'}(R_i)$ has three parts: the neighbors in $M_z$, the neighbors in $V(C - V_z)$, and the neighbors in $V_z$.
        We consider each of these parts separately.
        \begin{itemize}
            \item
                  We have $M_z \subseteq X_i$ by \cref{claim:find_robust_node_2}.
            \item
                  Consider a vertex $w \in N_{G'}(R_i) \cap V(C - V_z)$.
                  Then, $w$ is an ancestor of $z$.
                  Since $z$ and $z'$ have the same ancestor-type, also $w \in N_{G'}(G'_{z'})$.
                  By the robustness of $z'$, we have that $G'_{z'}$ is stable (in $G'$, $\PP$).
                  Since further $\leftBagAny(V_{z'}) < \leftBagAll(M_z) \leq i$, vertex $w$ is introduced in a bag $X_j$ with $j < i$.
                  Let $r \in R_i$ be a neighbor of $w$ in $R_i$.
                  Then, by the definition of $R_i$, vertex $r$ is introduced in a bag $X_{j'}$ with $i< j'$.
                  Since $w,r$ are neighbors, we have $w \in X_i$.
            \item
                  Finally, consider a vertex $w \in N_{G'}(R_i) \cap V_z$.
                  By our choice of $R_i$, vertex $w$ is introduced in a bag $X_j$ with $j \leq i$.
                  Then, analogously to the previous case, $w$ is part of $X_i$.
        \end{itemize}
        Thus, bag $X_i$ contains every vertex of $N_{G'}(R_i)$.
        This shows the final precondition for \cref{thm:embedding_with_twinning_graph:in:F},
        and concludes the proof.
    \end{claimproof}
    Since both cases, the case of \cref{claim:finish_proof_1} and \cref{claim:finish_proof_2}
    yield that $(G,k,M)$ is a yes-instance,
    this concludes the backward direction and the whole proof.
\end{proof}

We can now prove \cref{thm:treedepth_case}, that we restate for convenience.
\thmTreedepthCase*{}
\begin{proof}
    Follows directly from \cref{thm:treedepth_case_rule_valid,thm:treedepth_case_equivalent} and the fact that algorithm \nameref{rule:degree_reduction_elimination_forest} outputs a graph in which one edge is deleted if a vertex of $N_G(C)$ has sufficiently large degree.
\end{proof} 
\subsection{Bounding the Degree in Caterpillars}
\label{subsec:bounding_degree_caterpillars}
The goal of this section is to introduce a number of reduction rules to
bound the size of a caterpillar in a given instance $(G,k,M)$ of
\pwEtaDeletionEta.
We show that the size of a given caterpillar $C$ can be bounded in terms of the
size of its neighborhood by a uniform polynomial $|N_G(C)|^{\Oh(1)}$.

More concretely, we give an algorithm
(\nameref{rule:degree_reduction_caterpillar}) with the following guarantees.

\thmCaterpillarCase*{}

As in the previous sections, we first lay out a few key definitions
and tools that we make use of later in our reduction rules.

\subsubsection{The Key Ingredients}

We start by recalling a few basic notions.
For a vertex subset or a subgraph $X$ in a graph $G$,
the \textit{pendants} of $X$, denoted by $\pendants_G(X)$,
are the vertices in $N_G[X]$ with
degree $1$ in $G$.
Recall that a graph $C$ is a caterpillar if it is a tree for
which the deletion of pendants results in a path. We refer to this path as
the \textit{spine} of $C$ and its vertices as \textit{spine vertices}.
The spine of a caterpillar is unique and is allowed to be the empty graph.
For $x, y \in V(C)$, let $P$ be the unique $x$-$y$ path in $C$.
We conveniently use $C[x,y]$ to denote the caterpillar $C[V(P) \cup \pendants_C(V(P))]$.

The following straightforward observation is useful for inserting a path decomposition
of a caterpillar between two designated bags in an existing path decomposition.
\begin{observation}\label{obs:natural_pw_1}
    Let $C$ be a caterpillar and $x, y \in V(C)$ such that $C = C[x, y]$.
    Then there exists a path decomposition $X_1, \dots, X_N$ of $C$
    of width $1$ with $X_1 = \{x\}$ and $X_N = \{y\}$.
\end{observation}

We also require the following well-known characterization of caterpillars, a claw is the graph $K_{1,3}$ and a subdivided claw is the graph on 7 vertices obtained from $K_{1,3}$ by subdividing each edge exactly once.
\begin{proposition}[{\cite[Theorem 1]{bryant1987finding}, \cite{DBLP:journals/dam/KinnersleyL94}}]\label{fact:pw-1-char}
    Let $G$ be a graph. The following statements are equivalent:
    \begin{itemize}
        \item Every connected component of $G$ is a caterpillar,
        \item $G$ has pathwidth at most $1$,
        \item $G$ contains neither a $K_3$ nor a subdivided claw as minors.
    \end{itemize}
\end{proposition}

We use \Cref{fact:pw-1-char} above to show that there are
low-width path decompositions of a given caterpillar whose first bag satisfies a certain
useful property. In contrast with \Cref{obs:natural_pw_1}, this
lets us insert a path decomposition of a caterpillar starting from some designated bag
in an existing path decomposition.
As mentioned, the following will be a key property in manipulating path
decompositions of caterpillars to prove the safety of our reduction rules.

\begin{lemma}\label{lem:one_sided_embedding}
    Let $G$ be a caterpillar and $(V_1, V_2)$ a separation of $G$
    such that $V_1 \setminus V_2$ contains a spine vertex, or such that $V_1 \setminus V_2$ contains a pendant vertex adjacent to an endpoint of the spine of $G$.
    Then there exists a path decomposition $X_1, \dots, X_N$ of $G[V_2]$
    of width at most $|V_1 \cap V_2|$ with $X_1 = V_1 \cap V_2$.
\end{lemma}
\begin{proof}
    We prove the lemma statement by induction on $|V_2|$.
    For the base case $V_2 = \emptyset$, the statement holds trivially.

    Now assume $|V_2| \geq 1$. It follows that $V_1 \cap V_2 \neq \emptyset$
    since $(V_1, V_2)$ is a separation of the connected graph $G$, where $V_1, V_2 \neq \emptyset$.

    \begin{claim}\label{claim:frontier}
        There exists $p \in V_1 \cap V_2$ such that
        $|N_G(p) \setminus V_1| \leq 1$, or there exists a vertex $q \in V_2 \setminus V_1$ such that
        $N_G(q) \subseteq V_1 \cap V_2$.
\end{claim}
    \begin{claimproof}
        Assume for contradiction the claim does not hold.
        Let $u$ be a vertex in $V_1 \setminus V_2$ that is either a spine vertex, or the pendant of an endpoint of the spine of $G$.
        By the lemma statement, such a vertex $u$ exists.

        Now, let $p$ be a vertex of $V_1 \cap V_2$ such that there is a path from $p$ to $u$ in $G[V_1]$.
        Note that such a vertex $p$ exists because $G$ is connected, and $(V_1,V_2)$ is a separation with $V_1 \cap V_2$ nonempty.

        Since $p$ is in $V_1 \cap V_2$, by assumption we have $|N_G(p) \setminus V_1| \geq 2$,
        so let $q_1, q_2$ be two distinct vertices in $N_G(p) \setminus V_1$, and therefore $q_1,q_2 \in V_2 \setminus V_1$.
        Again, by assumption, there are vertices $q_1' \in N_G(q_1) \setminus V_1$
        and $q_2' \in N_G(q_2) \setminus V_1$.
        We observe that $q_1, q_2, q_1', q_2'$ are
        all distinct; otherwise, $G$ contains a $K_{3}$ minor, a contradiction to \Cref{fact:pw-1-char}.

        Set $P'_{u,p}$ to be a shortest path from $p$ to $u$ in $G[V_1]$.
        If this path contains at least three vertices, then set $P = P_{u,p}'$.

        Otherwise, assume that $P'_{u,p}$ contains only two vertices.
        First consider that $u$ is a pendant adjacent to an endpoint of the spine of $G$.
        Since $P'_{u,p}$ is a path on exactly two vertices, $p$ is the parent of $u$ and an endpoint of the spine of $G$:
        However, $p$ has two neighbors $q_1,q_2$ with degree at least two each, which is not possible for an endpoint of a spine.
        Hence, we have reached a contradiction.

        Now, consider that $u$ is a spine vertex.
        Then, $u$ has two neighbors, and as $u \in V_1 \setminus V_2$, both of these neighbors must be part of $V_1$.
        Hence, at least one of these two neighbors, say $v$, is not part of $P_{u,p}'$.
        In this case, let $P$ be the path on three vertices obtained by appending vertex $v$ to the path $P_{u,p}'$.

        In any case, the paths $P$, $p q_1 q_1'$, and $p q_2 q_2'$ form a subdivided claw minor, again a contradiction to
        \Cref{fact:pw-1-char}, which proves the claim.
    \end{claimproof}

    We make a distinction between the two cases given by the claim.
    In the first case, there is some $p \in V_1 \cap V_2$ such that $|N_G(p) \setminus V_1| \leq 1$.
    Here, we let $V_1' \coloneq V_1 \cup (N_G(p) \setminus V_1)$ and  $V_2' \coloneq V_2 \setminus \{p\}$.

    In the second case, there exists $q \in V_2 \setminus V_1$ with $N_G(q) \subseteq V_1 \cap V_2$;
    here, we let $V_1' \coloneq V_1 \cup \{q\}$ and $V_2' \coloneq V_2 \setminus \{q\}$.
    In both cases, we observe that $(V_1', V_2')$ is a separation of $G$ with $|V_2'| < |V_2|$, and that $|V_1' \cap V_2'| \leq |V_1 \cap V_2|$.
    Moreover, since $V_2' \subseteq V_2$, the special vertex $u$ is part of $V_1' \setminus V_2'$.

    Applying induction with $(V_1', V_2')$, let $X_1', \dots, X_{N-2}'$ be a path decomposition of $G[V_2']$ of width at most $|V_1' \cap V_2'| \leq |V_1 \cap V_2|$
    with $X_1' = V_1' \cap V_2'$.
    We claim that the required path decomposition of $G[V_2]$ is given by $X_1, \dots, X_N$,
    where $X_1 = V_1 \cap V_2$, $X_2 = X_1 \cup (V_1' \setminus V_1)$, and $X_j = X'_{j-2}$ for every $j \in [3, N]$.

    Let us now confirm this is a valid path decomposition.
    Clearly, $X_3,\dots,X_N$ forms a path decomposition of $G[V_2']$.
    Then, the only vertices which are in $V_2$ but not in $V_2'$ are vertices $p$ in the first case, or vertex $q$ in the second case.

    In the first case, note that $p \in X_1$ since $p \in V_1 \cap V_2 = X_1$.
    Moreover, by assumption $p$ has at most one neighbor in $V_2 \setminus V_1$.
    Clearly, all neighbors of $p$ in $V_1 \cap V_2$ are part of $X_1$ too.
    Then, if $p$ has no neighbor in $V_2 \setminus V_1$, it is clear that $X_1,\dots,X_N$ is a path decomposition of $G[V_2]$ because $V_1 \cap V_2 = (V_1' \cap V_2') \cup \{p\}$, and $p$ only appears in the bags $X_1$ and $X_2$.

    Now, assume that $p$ has a neighbor $p' \in V_2 \setminus V_1$.
    Then, by the definition of $V_1'$, this vertex $p'$ is part of $V_1'$.
    Hence, $p' \in V_1' \setminus V_1$, and thus $p' \in X_2$.
    Next, observe that each vertex $w \in X_1$ is also part of $X_2$.
    Moreover, observe that in this case we have $V_1 \cap V_2 = ((V_1' \cap V_2') \cup \{p\}) \setminus \{p'\}$.
    Then, it is easy to see that we have a path decomposition of $G[V_2]$ because $p$ meets all of its neighbors and is only part of $X_1,X_2$, and $p'$ is part of $X_2$ and $X_3$.

    Finally, we need to consider the second case, that there is some $q \in V_2 \setminus V_1$ with $N_G(q) \subseteq V_1 \cap V_2$.
    In this case, we have $V_1 \cap V_2 = V_1' \cap V_2'$, and $V_1' \setminus V_1 = \{q\}$.
    Then, it is easy to confirm that $X_1,\dots,X_N$ is a path decomposition of $G[V_2]$, because $q$ meets all of its neighbors in bag $X_2$, and is not part of any bag apart from bag $X_2$; and $X_2 = X_1 \cup \{q\}$, and $X_1 = X_3$.

    It remains to argue that the width of this path decomposition is at most $|V_1 \cap V_2|$.
    Clearly, the width of $X_3,\dots,X_N$ is at most $|V_1' \cap V_2'| \leq |V_1 \cap V_2|$ by the induction hypothesis.
    Then, notice that $X_1 = V_1 \cap V_2$, and that in any case, $|X_2| \leq |X_1| + 1 \leq |V_1 \cap V_2| + 1$, concluding the proof.
\end{proof}

Finally, the following simple technical lemma, applied to vertex neighborhoods,
lets us find vertices whose neighborhood is ``popular'' among other vertices.

\begin{lemma}\label{lem:popular_neighborhood}
    Let $\{S_i\}_{i \in [N]}$ be a sequence of sets with union $U$,
    and let $r \geq 2$ be an integer such that $N \geq 2r|U|^2$ and $N \geq 1$.
    Then, there exists an integer $i \in [N]$ such that
    for every $x,y \in S_i$, it holds that
    $|\{j \in [1,i] \mid \{x,y\} \subseteq S_j \}| \geq r$
    and $|\{j \in [i,N] \mid \{x,y\} \subseteq S_j \}| \geq r$.
    Moreover, such an integer can be computed in $\Oh(r |U|^4 + N|U|^2)$ time.
\end{lemma}
\begin{proof}
    If $U = \emptyset$, return an arbitrary index in $[N]$.

    Otherwise,
    for each index $i \in [N]$ a set $X = \{x,y\} \subseteq S_i$ (where $x = y$ is allowed) is a witness for $i$ if
    either $|\{j \in [1,i] \mid X \subseteq S_j \}| \leq r-1$
    or $|\{j \in [i, N] \mid X \subseteq S_j \}| \leq r-1$.

    We claim that a set $X = \{x,y\} \subseteq U$ can be a witness for at most $2r-2$ indices.
    To see this, let $J = \{j \in [N] \mid X \subseteq S_j\}$.
    Clearly, $X$ can only be a witness for the $r-1$ smallest indices of $J$ and for the $r-1$ largest indices of $J$, and not a witness for any other index besides these.

    Each index $i \in [N]$ that does not fulfill the conditions of the lemma must have some set $X_i$ that is a witness for $i$.
    Now, since for each $x,y \in U$ we have that set $\{x,y\}$ can be a witness for at most $2r-2$ indices, and there are at most $|U|^2$ such sets, we know that at most $(2r-2)|U|^2$ indices of $[N]$ can have witness, and therefore at most $(2r-2)|U|^2$ indices do not satisfy the conditions of the lemma.
    Then, since $N \geq 2r|U|^2$, a suitable index $i$ exists.

    Reading the input and computing the set $U$ takes time $O(N|U|^2)$.
    To compute such an integer $i$ efficiently, it suffices to only consider the
    first $2r |U|^2$ sets $\{S_i\}_{i \in [2r|U|^2]}$.
    We pass through the sequence $\{S_i\}_{i \in [2r|U|^2]}$ twice.
    In the first pass, for every $\ell \in [2r|U|^2]$, we record the number $F(X, \ell)$ of occurrences of every nonempty subset $X \subseteq U$ of size at most $2$
    in the prefix of $\{S_i\}_{i \in [2r|U|^2]}$ of length $\ell$; this takes $\Oh(r |U|^4)$ time in total.
    Then in the second pass, we find an index $i \in [2r |U|^2]$
    that satisfies $F(X, i) \geq r$ and $F(X, 2r|U|^2) - F(X, i-1) \geq r$ for every
    $X \in \binom{S_i}{\leq 2} \setminus \{\emptyset\}$,
    which can be checked in $\Oh(|U|^2)$ time per index.
    Overall, this takes $\Oh(r |U|^4 + N|U|^2)$ time, which finishes the proof.
\end{proof}

\subsubsection{The Reduction Rule}

At a high level, our reduction rule works in three steps.
The first step reduces the maximum degree of the caterpillar; this is a simple
application of \Cref{thm:reduce_ccs}.
The second step then reduces the total number of ``boundary'' vertices,
that is, the number of vertices having neighbors outside the caterpillar;
see \Cref{lem:delete_edge_to_caterpillar}.
Finally, if the caterpillar remains too large even after the first two steps,
then the third step finds parts of the caterpillar with no neighbors outside
and reduces them by contracting spine edges (see \Cref{lem:contract_spine})
or removal of pendant vertices, which again uses \Cref{thm:reduce_ccs}.

In what follows, we proceed by giving an algorithm for each step and finally show how
they all combine in \Cref{thm:caterpillar_case}.

\paragraph{Reducing Pendants}

In order to both bound the degree of a given vertex in the caterpillar
and to remove irrelevant non-boundary pendants,
a direct application of \Cref{thm:reduce_ccs}
is used to give a rule, \nameref{rule:reduce_star}, satisfying the following lemma.

\begin{lemma}
    \label{lem:reduce_star}
    Let $(G,k,M)$ be an instance of \pwEtaDeletion, $T$ be
    an induced subgraph of $G-M$ that is a star $K_{1, |V(T)|-1}$
    with pendants $L$ satisfying $|L| > f_{\ref{thm:reduce_ccs}}(|N_G(L)|)$.

    \namecref{rule:reduce_star} \nameref{rule:reduce_star} runs in
    polynomial-time, and when given $(G,k,M), T$ as input, it outputs an
    equivalent instance $(G',k,M)$ where $G'$ is a proper subgraph of $G$.
\end{lemma}
\begin{proof}
    The number of connected components of $G-V(G-L)$ is $|L| > f_{\ref{thm:reduce_ccs}}(|N_G(L)|)$.
    We apply the algorithm of \Cref{thm:reduce_ccs} on input $(G,k,V(G- L))$, and by the theorem the algorithm outputs an equivalent instance $(G',k,V(G - L))$ where $G'$ is a proper subgraph of $G$.
    As the instance is also equivalent to $(G', k, M)$, this concludes the proof.
\end{proof}

\begin{algorithm}[t]
    \caption{\algofont{ReducePendants}}
    \label{rule:reduce_star}
    \Input{Instance $(G,k,M)$ of \pwEtaDeletion and an induced subgraph $T$ of
        $G-M$ that is a star $K_{1, |V(T)|-1}$ with pendants $L$
        satisfying $|L| > f_{\ref{thm:reduce_ccs}}(|N_G(L)|)$.}

    $(G', k, V(G - L)) \gets$ \nameref{rule:reduce_ccs}$\left(G, k, V(G - L)\right)$\;

    \Return{$(G', k, M)$}
\end{algorithm}

\paragraph{Reducing Boundary Vertices}

Having bounded the maximum degree within the caterpillar using \Cref{lem:reduce_star},
we turn to the task of reducing the number of vertices having neighbors outside the caterpillar.
Here, we make use of \Cref{lem:popular_neighborhood} on two levels.
The key idea is to first partition the caterpillar into
groups, each of which is a smaller induced caterpillar.
Because of the maximum degree bound, there are enough groups to let us
apply \Cref{lem:popular_neighborhood} and find one group
whose neighborhood outside the caterpillar is ``popular'' among other groups.
By having a large enough such group, we once again make use of
\Cref{lem:popular_neighborhood} to find a vertex whose neighborhood
is ``popular'' among other vertices in the group and remove its incident edges;
see \namecref{rule:delete_edge_to_caterpillar} \nameref{rule:delete_edge_to_caterpillar}.
The challenge is to then prove that this does not turn a negative instance into a
positive one by showing how to manipulate a path decomposition of width at most $\eta$
of the modified instance into a path decomposition of width at most $\eta$
for the original instance.
\Cref{lem:one_sided_embedding} crucially lets us achieve this.

We now prove the following.

\begin{algorithm}[t]
    \caption{\algofont{ReduceCaterpillarBoundary}}
    \label{rule:delete_edge_to_caterpillar}
    \Input{Instance $(G,k,M)$ of \pwEtaDeletionEta and an induced subgraph $C$ of $G$ that is a caterpillar
        satisfying $\Delta(C) < f_{\ref{thm:reduce_ccs}}(|N_G(C)| + 5)$
        and $|N_G(G-C)| > f_{\ref{lem:delete_edge_to_caterpillar}}(|N_G(C)|).$}

    $s_1 s_2 \dots s_h \gets$ the spine of $C$\; \label{line:delete_edge:trivial_1}

    $V_1, \dots, V_h \gets$ the partition of $V(C)$ where each $V_i = V(C[s_i, s_i])$\; \label{line:delete_edge:trivial_2}

    Let $B_1, \dots, B_{h'}$ be a partition of $V(C)$ and $a, b: [h'] \to [h]$ be
    such that for $i \in [h']$, we have
    $B_i = \bigcup_{j \in [a(i), b(i)]} V_j$
    and $|\{j \in [a(i), b(i)] \mid N_G(G-C) \cap V_j \neq \emptyset\}| \geq 8(\eta+1) (|N_G(C)| + 5)^2 $\;
    \label{line:delete_edge:block_partitioning}

    Find an index $i_1 \in [h']$ such that for every $x,y \in N_G(B_{i_1}) \setminus V(C)$, we have $|\{j \in [h'] \mid x, y \in N_G(B_j)\}| \geq 4(\eta+1)(|N_G(C)|+1)$\;
    \label{line:delete_edge:popular_block}

    Find an index $i_0 \in [a(i_1), b(i_1)]$ such that $N_G(V_{i_0}) \setminus V(C) \neq \emptyset$
    and for every $x,y \in N_G(V_{i_0}) \setminus V(C)$, we have $|\{j \in [a(i_1), i_0] \mid x, y \in N_G(V_j)\}| \geq 4(\eta+1)$
    and $|\{j \in [i_0, b(i_1)] \mid x, y \in N_G(V_j)\}| \geq 4(\eta+1)$\;
    \label{line:delete_edge:popular_star}

    $v^\star \gets$ an arbitrary vertex in $N_G(G-C) \cap V_{i_0}$\;
    \label{line:delete_edges}

    Let $G'$ be obtained from $G$ by removing all edges between $v^\star$ and $N_G(C)$\;

    \Return{$(G', k, M)$}
    \label{line:delete_edge:end}
\end{algorithm}

\begin{lemma}
    \label{lem:delete_edge_to_caterpillar}
    Define $f_{\ref{lem:delete_edge_to_caterpillar}}(x) = (8(\eta+1))^3 (x + 5)^8$.
    Let $(G,k,M)$ be an instance of \pwEtaDeletionEta, $C$ be an induced subgraph of $G$ that is a
    caterpillar satisfying $\Delta(C) < f_{\ref{thm:reduce_ccs}}(|N_G(C)| + 5)$
    and $|N_G(G-C)| > f_{\ref{lem:delete_edge_to_caterpillar}}(|N_G(C)|)$.

    \namecref{rule:delete_edge_to_caterpillar} \nameref{rule:delete_edge_to_caterpillar} runs in polynomial time,
    and when given $(G,k,M), C$ as input,
    it outputs an equivalent instance $(G',k,M)$ where $G'$ is a proper subgraph of $G$.
\end{lemma}
\begin{proof}
    First, we show the algorithm is well-defined and efficient.
    \begin{claim}
        The algorithm is well-defined and runs in polynomial time.
    \end{claim}
    \begin{claimproof}
        \refline{line:delete_edge:trivial_1} and \refline{line:delete_edge:trivial_2} are trivial.
        For \refline{line:delete_edge:block_partitioning},
        we first recall that $|N_G(G-C)| > f_{\ref{lem:delete_edge_to_caterpillar}}(|N_G(C)|) = (8(\eta+1))^3 (|N_G(C)| + 5)^8$
        and $\Delta(C) < f_{\ref{thm:reduce_ccs}}(|N_G(C)| + 5) < 8(\eta+1) (|N_G(C)| + 5)^3$.
        By the pigeonhole principle,
        $|\{j \in [h] \mid N_G(G-C) \cap V_j \neq \emptyset\}|
            \geq |N_G(G-C)|/(\Delta(C)+1)
            \geq (8(\eta+1))^2 (|N_G(C)| + 5)^5$.
        Thus, taking $h' \coloneq 8(\eta+1) (|N_G(C)| + 5)^3$, \refline{line:delete_edge:block_partitioning} can be implemented
        by a simple scan of the sequence $V_1, \dots, V_h$ from left to right.

        For \refline{line:delete_edge:popular_block}, \Cref{lem:popular_neighborhood} applied to the sequence
        $\{N_G(B_j) \setminus V(C)\}_{j \in [h']}$ with $r_1 \coloneq 4(\eta+1)(|N_G(C)| + 1)$ shows that such
        an index $i_1$ exists and can be found in polynomial time since $h' \geq 2 r_1 |N_G(C)|^2$.

        Similarly, for \refline{line:delete_edge:popular_star}, given $i_1$,
        an index $i_0$ exists and can be found in polynomial time
        by applying \Cref{lem:popular_neighborhood} to the sequence
        $\{N_G(V_i) \setminus V(C) \mid N_G(V_i) \setminus V(C) \neq \emptyset\}_{i \in [a(i_1), b(i_1)]}$
        with $r_0 \coloneq 4(\eta+1)$; we note that the length of the sequence is at least $2 r_0 |N_G(C)|^2$ by definition of $B_{i_1}$.
        Finally, since $i_0 \in [a(i_1), b(i_1)]$ with $N_G(V_{i_0}) \setminus V(C) \neq \emptyset$,
        there exists a vertex $v^\star \in V_{i_0}$ having neighbors outside $C$.
        Therefore, \reflines{line:delete_edges}{line:delete_edge:end} output an instance $(G', k, M)$,
        where $G'$ is a proper subgraph of $G$. This finishes the proof.
    \end{claimproof}

    Now, we show the algorithm is correct.
    Assume the algorithm outputs an instance $(G', k, M)$.
    Let $s_1, \dots, s_h$, $V_1, \dots, V_h$, $B_1, \dots, B_{h'}$, $a, b$, $i_1 \in [h']$, $i_0 \in [a(i_1), b(i_1)]$,
    and $v^\star$ be defined as in the algorithm.
    Hence, $G'$ is a proper subgraph of $G$ obtained by removing the edges between $v^\star$ and $N_G(C)$. We first show that
    $(G', k, M)$ and $(G, k, M)$ are equivalent.

    The forward direction is trivial as $\pw(G'-S) \leq \pw(G-S)$ for every $S \subseteq V(G)$.

    For the backward direction, assume there exists $S' \subseteq V(G')$ of size at most $k$ such that $\pw(G'-S') \leq \eta$.
    Moreover, assume that $S'$ minimizes $|S \cap V(C)|$ among all sets $S \subseteq V(G')$ with $|S| \leq k$ and $\pw(G'-S) \leq \eta$.
    We aim to show that $\pw(G-S') \leq \eta$.

    Due to \Cref{obs:mostly_intact}, it immediately follows from our choice of $S'$ that $|S' \cap V(C)| < |N_G(C)|$.
    Let $\PP = X_1, \dots, X_p$ be a nice path decomposition of $G'-S'$ of width at most $\eta$ with $X_1 = X_p = \emptyset$.

    Now let $M_1 = N_{G}(B_{i_1}) \setminus (V(C) \cup S')$ and $M_0 = N_{G}(V_{i_0}) \setminus (V(C) \cup S')$.
    Since $V_{i_0} \subseteq B_{i_1}$ we clearly have $M_0 \subseteq M_1$.
    There is nothing to prove if $v^\star \in S'$ or there is a bag containing
    $\{v^\star\} \cup M_0$ as then $\PP$ is itself a path decomposition of $G-S'$ of width at most $\eta$.
    We shall thus assume otherwise, that is,
    assume that $v^\star \notin S'$, and that no bag in $\PP$ contains $\{v^\star\} \cup M_0$, in particular this implies that $M_0,M_1$ are nonempty.

    \begin{claim}\label{claim:M1_is_clique}
        There exists a bag in $\PP$ containing $M_1$.
    \end{claim}
    \begin{claimproof}
        The proof is trivial for $|M_1| \leq 1$, so we assume otherwise.
        By choice of $i_1$ (\refline{line:delete_edge:popular_block}), every distinct $x, y \in M_1$ are
        connected in $G$ by at least $4(\eta+1)(|N_G(C)|+1)$ internally vertex-disjoint paths in $C$.
        As $|S' \cap V(C)| < |N_G(C)|$, it follows that every distinct $x, y \in N_{G}(B_{i_1}) \setminus (V(C) \cup S')$
        remain connected in $G'-S'$ by at least $4(\eta+1)(|N_G(C)|+1) - |N_G(C)| - 2 > \eta + 2$ internally vertex-disjoint
        paths in $C-S'$. \Cref{obs:disjoint_paths_yield_clique} yields the claim.
    \end{claimproof}

    Note that since $M_0 \subseteq M_1$ this implies $\leftBagAll^{\PP}(M_0)$ and $\rightBagAll^{\PP}(M_0)$ are well-defined.
    As $v^\star$ is not in a bag with $M_0$, without loss of generality, we
    shall assume that $\leftBagAny^\PP(\{v^\star\}) > \rightBagAll^\PP(M_0)$; if this is not the case we can reverse path decomposition $\PP$ and it becomes true.

    \begin{claim}\label{claim:mostly_intact}
        $|S' \cap B_{i_1}| < \eta + 3$.
    \end{claim}
    \begin{claimproof}
        Suppose that $|S' \cap B_{i_1}| \geq \eta + 3$.

        By \Cref{claim:M1_is_clique}, there is a bag in $\PP$ of size $|M_1| = |N_{G}(B_{i_1}) \setminus (V(C) \cup S')| \leq \eta + 1$;
        thus, $|N_{G}(B_{i_1}) \setminus S'| \leq \eta + 3$ as $B_{i_1}$ has at most two spine vertices as neighbors in $C$.

        Let $S'' = (S' \setminus B_{i_1}) \cup N_{G'}(B_{i_1})$.
        Since $N_{G'-S''}(B_{i_1}) = \emptyset$
        and $\pw(G'[B_{i_1}]) \leq 1$, we get that $\pw(G'-S'') \leq \max\{1, \pw(G'-S''-B_{i_1})\} \leq \eta$.
        Moreover, we observe that $|S''| \leq |S'|$.
        Finally, note that $|S'' \cap V(C)| < |S' \cap V(C)|$ because $S'$ contained $\eta + 3$ vertices of $B_{i_1}$ which are not part of $S''$, but we added at most $2$ new vertices of $C$ to $S''$ which were not part of $S'$, namely vertices of $N_{C}(B_{i_1})$.
        This contradicts our choice of $S'$.
    \end{claimproof}

    \begin{claim}\label{claim:hierarchy}
        There are integers $\ell \in [a(i_1), i_0-1]$, $r \in [i_0+1, b(i_1)]$
        such that for $x \in \{\ell, r\}$,
        it holds $\leftBagAll^\PP(M_0) < \leftBagAny^\PP(s_x) \leq \rightBagAny^\PP(s_x) < \rightBagAll^\PP(M_0)$.
    \end{claim}
    \begin{claimproof}
        Due to \Cref{claim:M1_is_clique}, the indices $\leftBagAll^{\PP}(M_0)$ and $\rightBagAll^{\PP}(M_0)$ are well-defined.
        Set $L = \leftBagAll^\PP(M_0)$ and $R = \rightBagAll^\PP(M_0)$.
        Let $x \in M_0$ be the unique vertex introduced in $X_{L}$ and $y \in M_0$ be the unique
        vertex forgotten in $X_{R+1}$.
        Note that $M_0 \subseteq X_R$ is non-empty as, otherwise, we contradict the assumption that no bag contains $\{v^\star\} \cup M_0$.
        Hence, the bag $X_{R+1}$ exists as $\PP$ is a nice path decomposition with $X_p = \emptyset$.

        Let $J_{x,y} = \{j \in [a(i_1), b(i_1)] \mid V_j \cap S' = \emptyset \text{ and } x, y \in N_{G'-S'}(V_j)\}$.
        By choice of $i_0$ (\refline{line:delete_edge:popular_star}), it follows that
        $|J_{x,y} \cap [a(i_1), i_0-1]| \geq 4\eta + 4 - |S' \cap B_{i_1}| - 1 \geq 2\eta + 2$,
        where the last inequality is due to \Cref{claim:mostly_intact}.
        Similarly, we have $|J_{x,y} \cap [i_0+1, b(i_1)]| \geq 2\eta + 2$.
        Thus, by \Cref{lem:pd_hierarchy},
        there are integers $\ell \in [a(i_1), i_0-1]$ and $r \in [i_0+1, b(i_1)]$,
        where all vertices of $V_\ell$ and $V_r$ are introduced and forgotten in bags with indices lying in $[L+1, R-1]$.
        This proves the claim.
    \end{claimproof}

    \begin{claim}\label{claim:candidate1}
        There is a bag $X_{j_0}$ with $|X_{j_0}| \leq \eta$ such that
        $M_0 \subseteq X_{j_0}$
        and $s_\ell, s_r \notin X_j$ for every $j \geq {j_0}$.
    \end{claim}
    \begin{claimproof}
        Follows from \Cref{claim:hierarchy}
        by setting $j_0 = \max\{\rightBagAny^\PP(\{s_\ell\}), \rightBagAny^\PP(\{s_r\})\} + 1$; note that indeed $|X_{j_0}| \leq \eta$ because $X_{j_0}$ is a forget bag of a nice path decomposition.
    \end{claimproof}

    \begin{claim}\label{claim:candidate2}
        There is a bag $X_{j_1}$ with $|X_{j_1}| \leq \eta$ and $M_1 \subseteq X_{j_1}$.
    \end{claim}
    \begin{claimproof}
        Due to \Cref{claim:M1_is_clique}, there is an interval $[L, R]$
        maximizing $R-L$ such that every $j \in [L, R]$ satisfies $M_1 \subseteq X_{j}$.
        Let $x \in M_1$ be the unique vertex introduced in $X_L$ and $y \in M_1$ be the unique
        vertex forgotten in $X_{R+1}$. As in \Cref{claim:hierarchy}, the bag $X_{R+1}$ exists because
        $M_1 \subseteq X_R$ is non-empty (recall that $M_0 \subseteq M_1$) and $X_p = \emptyset$.

        Let $J_{x,y} = \{{i} \in [h'] \mid B_{i} \cap S' = \emptyset \text{ and } x, y \in N_{G'-S'}(B_{i})\}$.
        By choice of $i_1$, and by our previous observation that $|S' \cap V(C)| < |N_G(C)|$, we have $|J_{x,y}| \geq 4(\eta+1)(|N_G(C)| + 1) - |S' \cap V(C)| - 1
            \geq 2\eta + 2$.
        Using \Cref{lem:pd_hierarchy}, there exists some $i$ such that a vertex of $B_i$ is forgotten in a bag $X_{j_1}$
        with ${j_1} \in [L, R]$.
    \end{claimproof}

    \begin{claim}
        If $s_{i_0} \in S'$, then $(G,k,M)$ is a yes-instance.
    \end{claim}
    \begin{claimproof}
        In this case, since $v^\star \notin S'$ and $v^\star \in V_{i_0}$ we have that $v^\star$ is a pendant of $s_{i_0}$.
        In this case, $N_{G'}(v^\star) \setminus S' = \emptyset$ since $s_{i_0}$ is the only neighbor of $v^\star$ in $C$, and all edges from $v^\star$ to vertices in $N_G(C)$ were deleted for the creation of $G'$.
        The bag $X_{j_0}$ contains all vertices of $M_0$ and at most $\eta$ vertices.
        Hence, $N_{G}(v^\star) \setminus S' \subseteq X_{j_0}$.
        Create $\PP'$ from $\PP$ by adding $v^\star$ to bag $X_{j_0}$ and removing it from all other bags.
        Then, $\PP'$ is a path decomposition of $G' - S'$ and a bag contains $\{v^\star\} \cup M_0$, and hence $\PP'$ is a path decomposition of $G - S'$.
        Therefore, $(G,k,M)$ is a yes-instance.
    \end{claimproof}

    Hence, for the rest of the proof we shall assume $s_{i_0} \notin S'$.
    Consider the component $C'$ of $C[s_\ell, s_r]-S'$ containing $v^\star$. We note that $V(C') \subseteq B_{i_1}$ since $s_\ell,s_r \in B_{i_1}$, which itself follows from $\ell,r \in [a(i_1),b(i_1)]$.

    \begin{claim}\label{claim:sr_or_sl_in_C}
        If $V(C') \cap \{s_\ell,s_r\} = \emptyset$, then $(G,k,M)$ is a yes-instance.
    \end{claim}
    \begin{claimproof}
        Suppose that $s_r, s_\ell \notin V(C')$.

        Let $m$ be an arbitrary vertex in $M_0$, and let $S'' = (S' \setminus V(C[s_\ell, s_{i_0}])) \cup \{m\}$.
        Because $v^\star$ is disconnected from $s_\ell$ in $C'$, but $s_{i_0} \notin S'$, we have $|S' \cap
            V(C[s_\ell, s_{i_0}])| \geq 1$; hence, we have $|S''| \leq |S'|$.

        We now show that $\pw(G'-S'') \leq \eta$ by showing that we can use the bag
        $X_{j_1}$ to add a path decomposition of $C[s_\ell, s_{i_0}]$.

        By \Cref{claim:hierarchy} we know that there is a bag $X_J$ that contains $s_\ell$ and $M_0$, moreover $|X_J| \leq \eta + 1$.
        Now, we consider three cases.

        \proofsubparagraph*{Case 1: $J < j_1$.}
        Let $C''$ be the connected component of $C[s_{\ell+1}, s_r] \setminus (S'' \cap V(C[s_{i_0+1},s_r]))$ that contains $v^\star$.
        Intuitively, this component is similar to $C'$, only that it does additionally contain all vertices of $C[s_\ell,s_{i_0}]$.

        We note that $V(C'') \subseteq B_{i_1}$, and therefore $N_{G'}(C'') \setminus V(C) \subseteq M_1$.
        Let $Y_1, \dots, Y_N$ be a path decomposition of $C''$ of width $1$ with
        $Y_1 = \{s_{\ell+1}\}$ as guaranteed by \Cref{obs:natural_pw_1}.
        For $i \in [p]$, define $X_i' \coloneq X_i \setminus (V(C'') \cup S'')$.
        We claim that the following is a path decomposition of $G'-S''$ of width at most $\eta$:
        \begin{align*}
            \mathcal{Q} = & X_1', \dots, X_{J-1}', X_{J}' \cup Y_1, X_{J+1}' \cup Y_1, \dots, X_{j_1 -1}' \cup Y_1, \\
                          & \qquad X_{j_1}' \cup Y_1, X_{j_1}' \cup Y_2, \dots, X_{j_1}' \cup Y_N,
            X_{j_1+1}', \dots, X'_p.
        \end{align*}

        By construction, $\mathcal{Q}$ contains a path decomposition of both $C''$ and $G'-S''-C''$.
        We only need to verify that the vertices in $C''$ appear in a bag with their neighbors
        from $N_{G'-S''}(C'')$.
        We observe that
        $N_{G'-S''}(C'') \setminus V(C) \subseteq M_1 \setminus \{m\} \subseteq X_{j_1}'$.
        This covers the neighbors of $C''$ outside $V(C)$, since $M_1 \subseteq X_{j_1}$.

        Now, we first argue that $N_{G'-S''}(C'') \cap V(C) = \{s_\ell\}$.
        Since $C''$ is a caterpillar that contains $s_{\ell + 1}$, we indeed have $s_\ell \in N_{G' - S''}(C'') \cap V(C)$.
        Observe that, if there is some $s_r' \neq s_\ell$ that is part of $N_{G'  -S''}(C'') \cap V(C)$, then by our choice of $C''$ this vertex $s_r'$ would need to be identical to $s_r$.
        However, $s_r$ cannot be part of $C''$: there is a unique path from $v^\star$ to $s_r$ in $C$, and this path only uses vertices of $C[s_{i_0},s_r]$.
        As $s_r$ is not part of $C'$, the set $S'$ contains at least one vertex of this path, and as $s_{i_0} \notin S'$ the set $S'$ contains a vertex of this path that is part of $C[s_{i_0+1},s_r]$.
        Therefore, $s_r$ is also not part of $C''$.

        As $J$ was chosen such that $s_\ell \in X_J$
        and since $Y_1 = \{s_{\ell+1}\}$, the $J$-th bag of $\mathcal{Q}$ contains $\{s_\ell, s_{\ell+1}\}$.

        It remains to verify the width of $\mathcal{Q}$.
        Observe that, for any $i \in [J,j_1]$ it holds that $M_0 \subseteq X_i$, and hence $m \in X_i$.
        As $m \in S'' \setminus S'$, we find that the bag $X_i'$ contains at most $\eta$ vertices, and hence $X_i' \cup Y_1$ contains at most $\eta + 1$ vertices.
        Then, note that the bag $X_{j_1}$ has even stronger guarantees: it contains $m$ and at most $\eta$ vertices overall.
        Hence, $|X_{j_1}'| \leq \eta - 1$, and thus each bag of $X_{j_1}' \cup Y_1$,\dots,$X_{j_1}' \cup Y_{N}$ contains at most $\eta + 1$ vertices.
        Hence, $\mathcal{Q}$ has width at most $\eta$.

        This shows that $\pw(G'-S'') \leq \eta$. Since $|S'' \cap V(C)| < |S' \cap V(C)|$, this contradicts the choice of $S'$ and
        proves the claim in this case.

        \proofsubparagraph*{Case 2: $J = j_1$.}
        Using an argument analogous to Case 1 the following path decomposition suffices:
        \[
            \mathcal{Q} = X_1', \dots,
            X_{j_1-1}', X_{j_1}' \cup Y_1,
            X_{j_1}' \cup Y_2, \dots, X_{j_1}' \cup Y_N,
            X_{j_1+1}', \dots, X'_p.
        \]

        \proofsubparagraph*{Case 3: $J > j_1$.}
        In this case, we use the following decomposition:
        \begin{align*}
            \mathcal{Q} = & X_1', \dots, X_{j_1-1}',
            X_{j_1}' \cup Y_N,X_{j_1}' \cup Y_{N - 1}, \dots, X_{j_1}' \cup Y_1,                    \\
                          & \qquad X_{j_1+1}' \cup Y_1, \dots, X_J' \cup Y_1, X_{J+1}', \dots X'_p.
        \end{align*}
        Also in this case the proof is analogous to the one in Case 1.
    \end{claimproof}

    Hence, for the rest of the proof we additionally assume that at least one of $s_r,s_\ell$ is part of $C'$.
    We now show how to modify decomposition $\PP$ into a path decomposition of $G-S'$ of width at most $\eta$.
    Set $J \coloneq \max\{j_0, j_1\}$.
    Further, let $L \coloneq \bigcup_{i \leq J}(X_i \cap V(C'))$
    and $R \coloneq \bigcup_{i \geq J}(X_i \cap V(C'))$.
    As $\PP$ is a path decomposition of $G'-S'$, it holds that $(L, R)$ is a separation of $C'$
    of order $|L \cap R| = |X_J \cap V(C')|$, and because bag $X_J$ contains at least one vertex of $M_0 \subseteq M_1$, we have that $|X_J \cap V(C')| \leq \eta$.
    Our goal is using bag $X_J$ to embed a path decomposition of graph $G[R]$.

    By choice of index $J$, we have $M_0 \subseteq X_J$, and hence $J < \leftBagAny(\{v^\star\}).$
    Moreover, one of $s_r$ and $s_\ell$ is in $C'$ with $v^\star$ by assumption.
    Then, by \cref{claim:candidate1} we have that bag $X_J$ does not contain $s_\ell$ or $s_r$ since $J \geq j_0$.
    Therefore,
    we have that $v^\star \in R \setminus L$ and $\{s_\ell, s_r\} \cap (L \setminus R) \neq \emptyset$.
    Hence, since $s_\ell, s_r$ are spine vertices of $C$, we have that $(L, R)$ is a separation of $C'$, where $L \setminus R$ contains a spine vertex of $C'$, or $L \setminus R$ contains a pendant adjacent to an endpoint of the spine of $C'$.
    Thus, using \Cref{lem:one_sided_embedding}, let $Y_1, \dots, Y_N$ be a path
    decomposition of $C'[R]$ of width $|L \cap R|$ with $Y_1 = L \cap R$.

    For $i \in [p]$, define $X_i' \coloneq X_i \setminus R$.
    We claim that the following is a path decomposition of $G'-S'$ of width at most $\eta$:

    \[
        \PP' = X_1, \dots, X_{J-1}, X_J' \cup Y_1, \dots, X_J' \cup Y_N, X_{J+1}', \dots, X'_p.
    \]
    First, we show that $\PP'$ is a path decomposition of $G' - S'$.
    Clearly, each vertex $w \in V(G' - S')$ is either part of $R$ or of $G' - S' - R$, and thus part of some bag of $\PP'$.
    We confirm the other properties in separate claims.

    \begin{claim}
        For each $w \in V(G' - S')$ the indices of the bags of $\PP'$ containing $w$ form an interval.
    \end{claim}
    \begin{claimproof}
        Let $w$ be an arbitrary vertex of $G' -S'$.
        If $w$ is a vertex of $G' - S' - R$ this property is easy to see since the indices of the bags containing such vertices are (modulo the bags we copied) identical in $\PP$ and $\PP'$.
        If $w  \in R \setminus L$, then $w$ can only be part of a bag of $X_J' \cup Y_1,\dots,X_J' \cup Y_N$.
        Moreover, $w \notin X_J'$, and as $Y_1,\dots,Y_N$ is a path decomposition of $G'[R]$, we again obtain the desired interval.
        Finally, if $w \in R \cap L$, then note that $w$ is part of $X_J$ and part of $Y_1$.
        Then, because $Y_1,\dots,Y_N$ is a path decomposition of $R$, and $w$ is not part of a bag $X_j'$, the bags of $\PP'$ that contain $w$ and come after bag $X'_J \cup Y_1$ form an interval.
        Moreover, because $\PP$ is a path decomposition and $w \in X_J$, also the bags of $\PP'$ containing $w$ that come before bag $X_J' \cup Y_1$ form an interval.
        Finally, these two subintervals combine to a single interval since $w$ is in the bag $X_J' \cup Y_1$.
    \end{claimproof}

    \begin{claim}
        For each edge $ww'$ of $G' - S'$ some bag of $\PP'$ contains $w$ and $w'$.
    \end{claim}
    \begin{claimproof}
        This property is obvious for edges with both endpoints being from $G' - S' -R$, and for edges with both endpoints being from $R$.
        Hence, consider some $w \in R$ with neighbor $w' \in V(G' - S' - R)$ in graph $G' - S'$.

        First, consider that $w \in R \setminus L$.
        We cannot have that $w' \in V(C) \setminus V(C')$, because the only neighbors of $V(C')$ in $C$ could be vertices $s_\ell$ or $s_r$, but neither are part of $R$.
        Then, since $(L,R)$ is a separation of $C'$, vertex $w$ cannot have a neighbor in $L \setminus R$, and thus $w'$, a vertex that is not part of $R$, cannot be part of $C'$.
        So, we have $w' \in V(G - S' - C)$, which yields $w' \in M_1$.
        Now, note that in path decomposition $\PP$ the vertex $w$ is introduced after bag $X_J$ because $w \in R \setminus L$, but vertex $w'$ is introduced in bag $X_{j_1}$ or before.
        Hence, we must have $w' \in X_J'$, and as any bag containing a vertex of $R \setminus L$ contains all of $X_J'$, there is a bag containing $w$ and $w'$.

        Now, consider that $w \in L \cap R$.
        First, consider that $w' \in V(C)$.
        Note that then, $w' \in \{s_\ell,s_r\} \cup (L \setminus R)$, and therefore $w'$ is introduced and forgotten before bag $X_J$ in $\PP$.
        Hence, there is some bag $X_j$ with $j < J$ that contains both $w'$ and $w$, and by construction this bag is still a bag of $\PP'$.
        Similarly, if $w' \notin V(C)$, then $w' \in M_1$, and hence $w'$ is introduced in bag $X_J$ or earlier in $\PP$.
        If $w'$ is forgotten before bag $X_J$, then again some bag $X_j$ with $j < J$ contains $w$ and $w'$.
        Otherwise, $w'$ is part of $X_J'$, and $w$ is part of $Y_1$, and hence the two vertices appear in bag $X_J' \cup Y_1$.
    \end{claimproof}

    \begin{claim}
        The width of $\PP'$ is at most $\eta$.
    \end{claim}
    \begin{claimproof}
        Clearly, it suffices to show that $|X_J' \cup Y_j| \leq \eta + 1$ for each $j \in  [N]$ since the other bags of $\PP'$ are subsets of bags of $\PP$.
        Note that $L \cap R \subseteq X_J$, but $R \cap X_J' = \emptyset$.
        Moreover, we have $|X_J| \leq \eta$.
        Therefore, $|X_J'| \leq \eta - |L \cap R|$.
        Moreover, the width of $Y_1,\dots,Y_N$ is at most $|L \cap R|$.
        Hence, we have $\max_{j \in [N]} |X_J' \cup Y_j| \leq \eta - |L \cap R| + |L \cap R| + 1= \eta + 1$.
    \end{claimproof}
    Finally, observe that since $v^\star$ is part of $R \setminus L$, and $M_0 \subseteq X_J'$, we have that there is a bag containing $v^\star$ and $M_0$.
    Hence, $\PP'$ is even a path decomposition of $G - S'$ of width at most $\eta$, concluding the proof.
\end{proof}

\begin{algorithm}[t]
    \caption{\algofont{ContractSpine}}
    \label{rule:contract_spine}
    \Input{Instance $(G,k,M)$ of \pwEtaDeletionEta,
        an induced caterpillar $C$ of $G - M$
        with spine $s_1 s_2 \dots s_h$ of length $h > 8(\eta+1)$
        satisfying $N_G(G-C) \subseteq \{s_1, s_h\}$.}

    Let $G'$ be obtained from $G$ by contracting the edge $s_{\floor{h/2}} s_{\floor{h/2}+1}$\;

    \Return{$(G', k, M)$}
\end{algorithm}

\paragraph{Reducing Non-boundary Vertices}

As mentioned previously, the result of exhaustively applying the previous two rules
is a caterpillar with bounded maximum degree and a bounded number of boundary vertices.
If the caterpillar remains large enough, then there are large parts of the caterpillar
with no external neighbors, where essentially nothing ``interesting'' happens.
In this case, we can either use \Cref{lem:reduce_star} again to get rid of
non-boundary pendants, or, as we show below, we start contracting spine edges;
see \namecref{rule:contract_spine} \nameref{rule:contract_spine}.
To show the safety of the reduction rule, the high-level idea is that
every long enough path in a path decomposition must have two vertices that do
not appear together in a bag, which lets us insert a caterpillar of arbitrary length
in between using \Cref{obs:natural_pw_1}.

\begin{lemma}
    \label{lem:contract_spine}
    Let $(G,k,M)$ be an instance of \pwEtaDeletionEta, $C$ be an induced subgraph of $G - M$ that is a
    caterpillar with spine $s_1 s_2 \dots s_h$ of length $h > 8(\eta+1)$
    satisfying $N_G(G - C) \subseteq \{s_1, s_h\}$.

    \namecref{rule:contract_spine} \nameref{rule:contract_spine} runs in polynomial time,
    and when given $(G,k,M), C$ as input,
    it outputs an equivalent instance $(G',k,M)$ where $G'$ is a proper minor of $G$.
\end{lemma}
\begin{proof}
    It is easy to verify that the algorithm runs in polynomial time and outputs some instance $(G', k, M)$,
    where $G'$ is a proper minor of $G$.
    Let $s_1, \dots, s_h$ be defined as in the algorithm, and set $\mu \coloneq \floor{h/2}$.
    Thus, $G'$ is obtained by contracting the edge $s_{\mu} s_{\mu+1}$;
    let $z \in V(G')$ be the resulting fresh vertex, so that
    $V(G') = (V(G) \setminus \{s_{\mu}, s_{\mu+1}\}) \cup \{z\}$ and
    $N_{G'}(z) = N_G(\{s_{\mu}, s_{\mu+1}\}) \setminus \{s_{\mu}, s_{\mu+1}\}$.
    We now show that $(G', k, M)$ and $(G, k, M)$ are equivalent.

    The forward direction is trivial as $G'$ is a minor of $G$.

    For the backward direction, assume there is $S' \subseteq V(G')$ with $|S'| \leq k$ and $\pw(G'-S') \leq \eta$.
    Let $H \coloneq C[s_2, s_{h-1}]$.
    We note that we have $N_G(H) = \{s_1, s_h\}$ due to the assumption $N_G(G - C) \subseteq \{s_1, s_h\}$.
    Let $H' \coloneq G'[(V(H) \setminus \{s_{\mu}, s_{\mu+1}\}) \cup \{z\}]$ be the image of $H$ in $G'$;
    as $H'$ is obtained from the pathwidth-$1$ caterpillar $H$ by contracting the single spine edge
    $s_{\mu} s_{\mu+1}$, we have $\pw(H') \leq 1 \leq \eta$, and $N_{G'}(H') \neq \emptyset$.
    Hence, by \Cref{obs:mostly_intact} applied to $G'$, $H'$, and $S'$, we may assume $|S' \cap V(H')| \leq 1$.

    Now let $\PP = X_1, \dots, X_p$ be a nice path decomposition of $G'-S'$ of width $\leq \eta$. We distinguish two cases.

    \setcounter{case}{0}
    \begin{case}
        $|S' \cap V(H')| = 0$.
    \end{case}
    Then there are $\ell \in [2,\floor{h/2}-2]$ and $r \in [\floor{h/2}+3, h-1]$ such that $s_\ell$ and $s_r$ do not appear together in a bag;
    otherwise, some bag is forced to contain $\floor{h/2} - 3 \geq h/2 - 4 > \eta+1$ vertices.
    Let $L = \rightBagAny^{\PP}(\{s_\ell\}) + 1$ and $R = \leftBagAny^{\PP}(\{s_r\})$, without loss of generality $L < R$.
    Since $S' \cap V(H') = \emptyset$, there is an $s_\ell$-$s_r$ path in $H'-S'$.

    By \Cref{obs:natural_pw_1}, and since $s_{\ell+1},s_{r-1} \in V(H)$ by our choice of $\ell,r$,
    there is a path decomposition $Y_1, \dots, Y_N$ of $H[s_{\ell+1}, s_{r-1}]$ with $Y_1 = \{s_{\ell+1}\}$
    and $Y_N = \{s_{r-1}\}$.
    For every $i \in [p]$, define $X_i' \coloneq X_i \setminus V(H'[s_{\ell+1}, s_{r-1}])$, note that $z \in H'[s_{\ell+1}, s_{r-1}]$ and therefore $z$ is not part of any bag $X_i'$.
    We claim that the following gives a path decomposition of $G-S'$ of width at most $\eta$:
    \begin{align*}
        \PP' = & X_1', \dots,X_{L-2}', X_{L-1}' \cup Y_1, X_L' \cup Y_2, \dots, X_L' \cup Y_{N-1}, X_{L+1}' \cup Y_N, \\
               & \qquad X_{L+2}' \cup Y_N, \dots, X_R' \cup Y_N, X_{R+1}', \dots, X_p'
    \end{align*}

    Indeed, this follows since $N_{G-S'}(H[s_{\ell+1}, s_{r-1}]) = \{s_\ell, s_r\}$, $s_\ell \in X'_{L-1} $, and
    $s_r \in X'_R$.
    As $|S' \cap V(H')| = 0$, there is an $s_{\ell}$-$s_r$ path in $G'-S'$,
    and, hence, a vertex of $H'[s_{\ell+1}, s_{r-1}]$ appears in every bag between
    $X_{L-1}$ and $X_R$.
    Therefore, $X_{L-1}',\dots,X_R'$ contain at most $\eta$ vertices each.
    The bag $X_L'$ even contains at most $\eta - 1$ vertices since $X_L$ is a forget bag.
    As $Y_1, \dots, Y_N$ is a path decomposition of width~$1$,
    the largest size of a bag of $\PP'$ is at most $\max\{|X_{L}'| + 2, \max_{i \in [L-1,R] \setminus \{L\}}|X_i'|+1\} \leq \eta + 1$.

    \begin{case}
        $|S' \cap V(H')| = 1$.
    \end{case}
    As such, either $S' \cap V(H'[s_2, s_{\floor{h/2}-1}]) = \emptyset$ or $S' \cap V(H'[s_{\floor{h/2}+2}, s_{h-1}]) = \emptyset$.
    Without loss of generality, we shall assume $S' \cap V(H'[s_2, s_{\floor{h/2}-1}]) = \emptyset$; the other case is symmetric.
    As $\floor{h/2} -2 > h/2 - 3 > \eta + 1$, there are distinct $\ell, r \in [2, \floor{h/2}-1]$ with $\ell < r$, such that $s_{\ell}$ and $s_{r}$ do not
    appear together in a bag.
    Let $L = \rightBagAny^{\PP}(\{s_\ell\}) + 1$ and $R = \leftBagAny^{\PP}(\{s_r\})$, without loss of generality $L < R$.
    By assumption, there is an $s_\ell$-$s_r$ path in $H-S'$.

    Using \Cref{obs:natural_pw_1},
    let $Y_1, \dots, Y_N$ be a path decomposition of $H[s_{\ell+1}, s_{h-1}]$ of width $1$ with $Y_1 = \{s_{\ell+1}\}$.
    Note that $s_{\ell+1} \in V(H)$ because $\ell < r \leq \floor{h/2}-1$, and hence $H[s_{\ell+1}, s_{h-1}]$ is indeed well-defined.
    Let $S'' \coloneqq (S' \setminus V(H')) \cup \{s_{h}\}$; we note that $|S''| \leq |S'|$. Now we show that $\pw(G-S'') \leq \eta$.
    For every $i \in [p]$, define $X_i' \coloneq X_i \setminus (V(H'[s_{\ell+1}, s_{h-1}]) \cup \{s_h\})$, observe that $z \in V(H'[s_{\ell+1}, s_{h-1}])$ and therefore $z$ is not in any bag $X_i'$.
    Since $N_{G-S''}(H[s_{\ell+1}, s_{h-1}]) = \{s_\ell\} = N_{G-S''}(s_{\ell+1}) \setminus V(H[s_{\ell+1}, s_{h-1}])$,
    it is simple to verify that the following is a path decomposition of $G-S''$:
    \[
        \PP' = X_1', \dots, X_{L-1}' \cup Y_1, X_L' \cup Y_2, X_L' \cup Y_3, \dots, X_L' \cup Y_N, X_{L+1}', \dots, X_p'.
    \]

    To bound the width of $\PP'$, we note that the $s_\ell$-$s_r$ path in $H'-S'$
    implies that every bag between $X_{L-1}$ and $X_{R}$ contains a vertex of $H'[s_{\ell+1}, s_{h-1}]$.
    Then, $X'_{L-1},\dots,X'_R$ contain at most $\eta$ vertices each, and $X_L'$ even at most $\eta - 1$ vertices since $X_L$ is a forget bag.
    Because $Y_1, \dots Y_N$ is a path decomposition of width $1$,
    the largest bag of $\PP'$ has size at most $\max\{|X_{L-1}'|+1, |X_L'| + 2\} \leq \eta + 1$.
\end{proof}

\paragraph{Reducing Caterpillars}

We finally show how to combine the previous algorithms into a single reduction rule (see \namecref{rule:degree_reduction_caterpillar} \nameref{rule:degree_reduction_caterpillar}).
The algorithm simply ensures that the conditions of each of \Cref{thm:reduce_ccs,lem:delete_edge_to_caterpillar,lem:contract_spine}
are met before applying their respective algorithms to make progress.
This lets us prove the main lemma of this section which we restate below once more
for convenience.

\begin{algorithm}[t]
    \caption{\algofont{ReduceCaterpillar}}
    \label{rule:degree_reduction_caterpillar}
    \Input{Instance $(G,k,M)$ of \pwEtaDeletionEta and an induced subgraph $C$ of $G-M$ that is a caterpillar
        satisfying $|V(C)| > \BoundCaterpillar$}

    \If{$\Delta(C) \geq f_{\ref{thm:reduce_ccs}}(|N_G(C)| + 5)$}{
        \label{line:reducestar1_cond}
        Let $s$ be a vertex of maximum degree in $C$\;
        \Return{\nameref{rule:reduce_star}$\left((G, k, M), C[s,s]\right)$}\;
        \label{line:reducestar1}
    }
    \If{$|N_G(G-C)| > f_{\ref{lem:delete_edge_to_caterpillar}}(|N_G(C)|)$}{
        \label{line:deleteedge_cond}
        \Return{\nameref{rule:delete_edge_to_caterpillar}$\left((G, k, M), C\right)$}
        \label{line:deleteedge}
    }

    \If{there are $x,y$ such that $C[x,y]$ has a spine with more than
        $8(\eta+1)$ vertices and no neighbors outside $C$}{
        \label{line:contractspine_cond}
        \Return{\nameref{rule:contract_spine}$\left((G, k, M), C[x, y]\right)$}
        \label{line:contractspine}
    }

    Let $s \in V(C)$ be a vertex with a set $L$ of $8(\eta+1)$ pendants not in $N_G(G-C)$\;
    \label{line:reducestar2_cond}
    \Return{\nameref{rule:reduce_star}$\left((G, k, M), C[\{s\} \cup L]\right)$}\;
    \label{line:reducestar2}
\end{algorithm}

\thmCaterpillarCase*{}
\begin{proof}
    We consider four cases corresponding to the execution of Lines \ref{line:reducestar1}, \ref{line:deleteedge}, \ref{line:contractspine}, or \ref{line:reducestar2}.
    In each case we show that one of \Cref{thm:reduce_ccs,lem:delete_edge_to_caterpillar,lem:contract_spine}
    can be applied to prove the statement.

    \setcounter{case}{0}
    \begin{case}\label{case:caterpillar:one}
        $\Delta(C) \geq f_{\ref{thm:reduce_ccs}}(|N_G(C)| + 5)$
    \end{case}
    Let $s \in V(C)$ be a spine vertex with $|N_C(s)| = \Delta(C)$ as in the algorithm,
    and let $L$ be its pendants.
    The number of pendants $|L|$ is at least
    $\Delta(C) - 2 > f_{\ref{thm:reduce_ccs}}(|N_G(C)| + 1)$.
    Since $N_G(L) \subseteq N_G(C) \cup \{s\}$, we have that
    $|N_G(L)| \leq |N_G(C)| + 1$.
    Therefore, $|L| > f_{\ref{thm:reduce_ccs}}(|N_G(L)|)$.
    Thus, $C[s,s]$ satisfies the conditions of \Cref{lem:reduce_star}
    and \refline{line:reducestar1} outputs an equivalent instance $(G', k, M)$ in polynomial time
    where $G'$ is proper subgraph of $G$.

    \begin{case}\label{case:caterpillar:two}
        \Cref{case:caterpillar:one} does not hold
        and $|N_G(G-C)| > f_{\ref{lem:delete_edge_to_caterpillar}}(|N_G(C)|)$.
    \end{case}
    As $\Delta(C)$ is bounded due to \Cref{case:caterpillar:one},
    $C$ satisfies the conditions of \Cref{lem:delete_edge_to_caterpillar}, so
    executing \refline{line:deleteedge} outputs an instance $(G', k, M)$ in
    polynomial time where $G'$ is a proper subgraph of $G$.

    \begin{case}\label{case:caterpillar:three}
        \Cref{case:caterpillar:one,case:caterpillar:two} do not hold
        and there are vertices $x,y \in V(C)$ where $C[x,y]$ has a spine
        of more than $8(\eta+1)$ vertices and no neighbors outside $C$.
\end{case}
    By definition of this case, $C[x,y]$ satisfies the conditions of \Cref{lem:contract_spine}, so
    executing \refline{line:contractspine} outputs an instance $(G', k, M)$ in
    polynomial time where $G'$ is a proper minor of~$G$.

    \begin{case}
        \Cref{case:caterpillar:one,case:caterpillar:two,case:caterpillar:three} do not hold.
    \end{case}
    We show that there is a spine vertex with more than $8(\eta+1)$ pendants
    that are not in $N_G(G-C)$.
    Since \Cref{case:caterpillar:three} does not hold,
    for any subpath of the spine of $C$ on $8(\eta+1)+3$ vertices, either a vertex of the path or a pendant of a vertex of the path is in $N_G(G-C)$.
    Moreover, due to \Cref{case:caterpillar:two} not holding,
    $|N_G(G-C)|$ is bounded by $f_{\ref{lem:delete_edge_to_caterpillar}}(|N_G(C)|) = (8(\eta+1))^3 (|N_G(C)| + 5)^8$.
    Thus, the spine of $C$ has at most $(8(\eta+1))^5 (|N_G(C)| + 5)^8$ vertices.
    As such, because $|V(C)| > \BoundCaterpillar$,
    the total number of pendants not in $N_G(G-C)$
    is at least
    $|V(C)| - (8(\eta+1))^5 (|N_G(C)| + 5)^8 - |N_G(G-C)| \geq (8(\eta+1))^7 (|N_G(C)| + 5)^8$.
    By the pigeonhole principle, there is a spine vertex $s$ with at least $8(\eta+1) > f_{\ref{thm:reduce_ccs}}(1)$ pendants not in $N_G(G-C)$.
    Thus, the star induced on $s$ and those pendants satisfies the conditions of \Cref{lem:reduce_star}.
    As in \Cref{case:caterpillar:one}, this concludes this case.

    Since all conditions in the algorithm are simple to check in polynomial time, this finishes the proof.
\end{proof}  
\section{Reducing the Size of Connected Components}
\label{sec:kernel_reducing_cc_size}
We proceed to the final step of our algorithm, reducing the size of connected components.
We achieve this goal in two different stages.
First, a lifting step ensures that each remaining connected component has a neighborhood of constant size,
while possibly increasing the size of the modulator.
In the second step, our already established reduction rules bound the size of these connected components.

\subsection{Creating Protrusions}
Although our techniques \emph{do not} rely on the large hammer of protrusion replacement in the usual sense, e.g.\ \cite{DBLP:journals/jacm/BodlaenderFLPST16,DBLP:journals/siamcomp/FominLST20,fominPlanarFdeletionApproximation2012}, we will first employ a known technique (see, e.g.,~\cite{DBLP:conf/isaac/LochetS24,fominHittingForbiddenMinors2016,donkersPreprocessingOuterplanarVertex2022,kimLinearKernelsSingleexponential2016}) to guarantee that all connected components of $G - M$ are protrusions.
We recall from \cref{def:protrusion} that an induced subgraph $C$ of $G$ is a protrusion if $|N_G(C)| \leq 2(\eta+1)$ and $\pw(C)\leq\eta$.

Set $\BoundProtrusionModulator[x,y] = x + x \cdot \degreeBound[x,y] \cdot (\eta + 1)$\phantomsection\label{bound:protrusion_modulator}.
\begin{theorem}
    \label{thm:create_protrusions}
    There is a polynomial-time algorithm \protrusionAlgo{} that takes an instance $(G,k,M)$ of \pwEtaDeletionEta{} as input, where $\sum_{v \in M} |N_{G}(v) \setminus M| \leq |M| \cdot \degreeBound$, and outputs an equivalent instance $(G,k,M')$ such that
    \begin{bracketenumerate}
        \item each connected component $C$ of $G - M'$ is a protrusion of $G$, and
        \item $|M'| \leq \BoundProtrusionModulator$, and
        \item $M \subseteq M'$.
    \end{bracketenumerate}
\end{theorem}
\begin{proof}
    Our algorithm proceeds as follows.
    Compute a path decomposition $X_1,\dots,X_t$ of width at most $\eta$ of graph $G - M$ in polynomial time using \cref{thm:path_decomposition_algorithm}.
    For each $m \in M$, and each vertex $v \in N_G(m) \setminus M$, mark a bag of $X_1,\dots,X_t$ that contains $v$.
    Let $\XX$ be the set of all marked bags, and set $M' = M \cup \bigcup_{X \in \XX} X$.
    Finally, output $(G,k,M')$.

    By construction, $M \subseteq M'$.
    Because $\sum_{v \in M} |N_{G}(v) \setminus M| \leq |M| \cdot \degreeBound$ we have $|M'| \leq \BoundProtrusionModulator$.
    Finally, for property~(1), consider a connected component $C$ of $G - M'$.
    Then, $N_G(C) \cap M = \emptyset$, and therefore all vertices of $N_G(C)$ stem from at most two bags of $X_1,\dots,X_t$.
    Hence, $|N_G(C)| \leq 2(\eta + 1)$.
    Since additionally $\pw(C) \leq \eta$, $C$ is a protrusion of~$G$.
\end{proof}

\subsection{Reducing the Size of Protrusions}
Given \cref{thm:create_protrusions}, and the fact that we can already bound the degree of vertices in $M$ (\cref{thm:reduce_modulator_degree}), we could apply \emph{protrusion replacers} (see \cite{DBLP:journals/jacm/BodlaenderFLPST16,DBLP:journals/siamcomp/FominLST20,fominPlanarFdeletionApproximation2012}) to immediately finish the kernel.
Here, we instead opt for an approach that completely avoids this ``somewhat opaque'' approach.
We again define some bounds which are important for this step.

\begin{definition}[Important Bounds for Protrusion Replacement]
    \label{def:important_bounds_protrusion_replacement}
    Set
    \begin{align*}
        \BoundProtrusionCaterpillarNeighborhood & = 2(\eta + 1) + \beta,                                                                                                \\
        \BoundChildrenInForest                  & = f_{\ref{thm:reduce_ccs}}( \BoundProtrusionCaterpillarNeighborhood )                                                 \\
        \BoundProtrusionSize                    & = (\BoundChildrenInForest)^{\beta+1} \cdot f_{\ref{thm:caterpillar_case}}( \BoundProtrusionCaterpillarNeighborhood ).
    \end{align*}
\end{definition}

Let us consider a connected protrusion $C$, and a nice $\pwEtaBoundClass[1]$-elimination tree $\FF = (F,\{Y_n\}_{n \in V(F)})$ of $C$ of depth at most $\beta$.
\begin{itemize}
    \item
          Then, $\BoundProtrusionCaterpillarNeighborhood$ bounds the size of the neighborhood of $G[V_n^\FF]$ for each node $n \in V(F)$.
    \item
          Further, $\BoundChildrenInForest$ bounds the number of children of a node $n$ in $F$,
          unless we can reduce the size of the graph by use of \cref{thm:reduce_ccs}.
          To see this, note that $V_c^\FF$, for children $c$ of $n$, are connected components in the graph induced by $V_n^\FF \setminus \{n\}$.
    \item
          Finally, $\BoundProtrusionSize$ is the size bound on the size of $C$ that we eventually obtain.
\end{itemize}
We now describe an algorithm \nameref{rule:reduce_protrusions} that can reduce the size of large protrusions, and prove that it functions correctly.

\begin{algorithm}[t]
    \caption{\algofont{ReduceProtrusions}}
    \label{rule:reduce_protrusions}
    \Input{Instance $(G,k,M)$ of \pwEtaDelElimGeneral{}, connected induced subgraph $C \subseteq G - M$ that is a protrusion of $G$ with $|V(C)| > \BoundProtrusionSize$}
    $\FF = (F,\{Y_n\}_{n \in V(F)})\leftarrow$ nice $\pwEtaBoundClass[1]$-elimination tree of $C$ of depth at most $\beta$\;
    delete all nodes $n \in V(F)$ with $Y_n = \emptyset$ from $\mathcal{F}$\;
    \uIf{some $n \in V(F)$ has more than $\BoundChildrenInForest$ children in $F$}{
        $M' = (V(G) \setminus V^{\FF}_n) \cup \{n\}$\;
        $(G',k,M) \leftarrow$ \nameref{rule:reduce_ccs}$(G,k,M')$\;
        \Return{$(G',k,M)$}
    }\Else{
    $n \leftarrow$ leaf node of $F$ such that $|Y_n| > f_{\ref{thm:caterpillar_case}}( \BoundProtrusionCaterpillarNeighborhood )$\;
    \Return{\nameref{rule:degree_reduction_caterpillar}$((G,k,M),G[Y_n])$}
    }
\end{algorithm}

\begin{lemma}
    \label{thm:reduce_protrusions}
    \namecref{rule:reduce_protrusions} \nameref{rule:reduce_protrusions} runs in polynomial-time.
    Given instance $(G,k,M)$ of \pwEtaDelElimGeneral{} and connected induced subgraph $C \subseteq G - M$ that is a protrusion of $G$ as input, and $|V(C)| > \BoundProtrusionSize$, the algorithm outputs an equivalent instance $(G',k,M)$ where $G'$ is a proper minor of $G$.
\end{lemma}
\begin{proof}
    At the start of the algorithm a nice $\pwEtaBoundClass[1]$-elimination forest $\FF = (F,\{Y_n\}_{n \in V(F)})$ of $C$ is computed, and this is possible in polynomial time using \cref{thm:compute_elimination_forests}.
    Then, the algorithm deletes all nodes $n \in V(F)$ where $Y_n = \emptyset$, note that such a node must be a leaf node by the definition of elimination forests.
    We note that $\FF$ remains a nice $\pwEtaBoundClass[1]$-elimination tree after these deletions.
    This is the case because, even if the parent $p$ of such a node $n$ now becomes a leaf, we have $Y_p = \{p\}$,
    and then $V_p^\FF$ still induces a graph of pathwidth at most one.
    Next, the algorithm distinguishes two cases.

    In the first case, there is a node $n$ of $F$ that has more than $\BoundChildrenInForest$ children in $F$.
    Then, the algorithm sets $M'$ to be the set of all vertices of $G$ apart from those vertices that are in bags of descendants of $n$ that are not $n$ itself.
    Note that each child of $n$ in $F$ corresponds to a connected component of $G - M'$ since $Y_{n'} \neq \emptyset$ for all $n' \in V(F)$.
    Therefore, $G - M'$ has more than $\BoundChildrenInForest$ connected components.
    At the same time, each connected component $W$ of $G - M'$
    has a neighborhood that is a subset of $N_G(C) \cup \anc_F(n)$.
    Since the depth of $\FF$ is at most~$\beta$,
    we have $|N_G(W)| \leq 2(\eta+1) + \beta =\BoundProtrusionCaterpillarNeighborhood$.
    Recall that $\BoundChildrenInForest = f_{\ref{thm:reduce_ccs}}( \BoundProtrusionCaterpillarNeighborhood)$.
    Then, by \cref{thm:reduce_ccs}, the output $(G',k,M)$ of \nameref{rule:reduce_ccs}$(G,k,M')$ is an instance that is equivalent to $(G,k,M)$ and $G'$ is a proper subgraph of $G$.
    Therefore, in this case the algorithm outputs an equivalent instance where $G'$ is a proper minor of $G$.

    Otherwise, all nodes of $F$ have at most $\BoundChildrenInForest$ children.
    This implies $|V(F)| \leq (\BoundChildrenInForest)^{\beta+1}$, because $F$ is a rooted tree of depth at most $\beta$.
    By assumption, $|V(C)| > \BoundProtrusionSize  = (\BoundChildrenInForest)^{\beta+1} \cdot f_{\ref{thm:caterpillar_case}}( \BoundProtrusionCaterpillarNeighborhood )$,
    so, by the pigeonhole principle, there exists some leaf node $n \in V(F)$ such that $|Y_n| > f_{\ref{thm:caterpillar_case}}( \BoundProtrusionCaterpillarNeighborhood )$.
    Set $H = G[Y_n]$.
    Then, we clearly have $|N_G(H)| \leq 2(\eta + 1) + \beta = \BoundProtrusionCaterpillarNeighborhood$ and $|V(H)| > f_{\ref{thm:caterpillar_case}}( \BoundProtrusionCaterpillarNeighborhood )$.
    It then follows from \cref{thm:caterpillar_case} that \nameref{rule:degree_reduction_caterpillar} applied on input $((G,k,M),H)$ outputs an equivalent instance $(G',k,M)$ where $G'$ is a proper minor of~$G$.
\end{proof}

\subsection{The Uniform Kernel}
Now, we have all ingredients at hand to prove our main result.
\begin{algorithm}[t]
    \caption{\algofont{UniformKernel}}
    \label{kernel}
    \Input{Instance $(G,k,M)$ of \pwEtaDelElimGeneral{}; $\eta \geq 1$}
    \If{$|M| \leq k$}{
        \Return{the yes-instance $((M,\emptyset),0,M)$}
    }
    \If{some $m \in M$ has $|N_G(m) \setminus M| > \degreeBound$}{
        \Return{\nameref{rule:reduce_modulator_degree}}$(G,k,M)$ \label{kernel:reduce_mod_degree}\;
    }
    $(G,k,M_2) \leftarrow$ \protrusionAlgo{}$(G,k,M)$ \label{kernel:compute_protrusions} \;
    $\CC \leftarrow$ connected components of $G - M_2$ \label{kernel:set_ccs} \;
    \If{$|\CC| > f_{\ref{thm:reduce_ccs}}( \BoundProtrusionModulator ) $}{
        $(G',k,M_2) \leftarrow$ {\nameref{rule:reduce_ccs}}$(G,k,M_2)$ \label{kernel:reduce_ccs}\;
        \Return{$(G',k,M)$}
    }
    \If{there is some $C \in \CC$ with $|V(C)| > \BoundProtrusionSize$}{
        \Return{\nameref{rule:reduce_protrusions}}$((G,k,M),C)$ \label{kernel:reduce_protrusion_size} \;
    }
    \Return{$(G,k,M)$}
\end{algorithm}

\mThmGeneral*{}
\begin{proof}
    Recall that we fixed $\eta \geq 1, \beta \geq 0$ earlier.
    Since the theorem also states a result for $\eta = 0$, let us remark that \pwEtaDelElimGeneral{} for $\eta = 0$ is exactly the same as \textsc{Vertex Cover} parameterized by the size of a given vertex cover.
    Then, the result follows from \cref{thm:uniform_kernel_vc_parameterization}.

    We prove that \nameref{kernel} yields the desired kernel.
    More precisely, we show that if \nameref{kernel} is applied on an instance $(G,k,M)$ that does not yet have bounded size, then it outputs a strictly smaller instance $(G',k',M)$ where $G'$ is a proper minor of $G$.
    Then, the kernel follows by just applying \nameref{kernel} until the algorithm no longer reduces the size of the instance.

    If $|M| \leq k$, then clearly $(G,k,M)$ is a yes-instance, and therefore we can output the trivial yes-instance $((M,\emptyset),0,M)$ (this is one of the cases where we output the empty graph if $M = \emptyset$).
    Note that this is the only case in which our algorithm may output an instance with a different solution size compared to the input.
    If there is some $m \in M$ with $|N_G(m) \setminus M| > \degreeBound$ the algorithm applies \nameref{rule:reduce_modulator_degree}.
    By \cref{thm:reduce_modulator_degree} the output of this algorithm is an equivalent instance $(G',k,M)$ where $G'$ is a proper minor of $G$.

    If all vertices in $M$ have bounded degree, the algorithm proceeds to \cref{kernel:compute_protrusions}, which just applies the algorithm \protrusionAlgo{} of \cref{thm:create_protrusions}.
    The output of \protrusionAlgo{} is an equivalent instance $(G,k,M_2)$ where $|M_2| \leq \BoundProtrusionModulator$ and $M \subseteq M_2$.
    Then, \nameref{kernel} sets $\CC$ to be the connected components of $G - M_2$.

    If $|\CC|$ is larger than $f_{\ref{thm:reduce_ccs}}( \BoundProtrusionModulator)$, it sets $(G',k,M_2)$ to the output of applying \nameref{rule:reduce_ccs} on input $(G,k,M_2)$.
    By \cref{thm:reduce_ccs}, if that is the case, then $G'$ is a proper subgraph of $G$, and the instance $(G',k,M_2)$ is equivalent to the instance $(G,k,M_2)$.
    Recall that $M \subseteq M_2$.
    Therefore, the kernel can output $(G',k,M)$, which is equivalent to $(G,k,M)$ but strictly smaller.

    If $\CC$ is also of bounded size, then there may still be some connected component $C \in \CC$ that is large, that is, that has size larger than $\BoundProtrusionSize$.
    In this case the algorithm outputs the result of applying \nameref{rule:reduce_protrusions} on instance $(G,k,M)$ and such a graph $C$.
    Since $C$ is an induced connected subgraph of $G - M$, and $C$ is a protrusion of $G$ by \cref{thm:create_protrusions}, the algorithm outputs an equivalent instance $(G',k,M)$ where $G'$ is a proper minor of $G$ by \cref{thm:reduce_protrusions}.

    Otherwise, we have that $|M_2|$ has size at most $\BoundProtrusionModulator$, $|\CC|$ has size at most $f_{\ref{thm:reduce_ccs}}(\BoundProtrusionModulator)$, and each $C \in \CC$ has at most $\BoundProtrusionSize$ vertices.
    Therefore, the graph $G$ has at most $\BoundProtrusionModulator + f_{\ref{thm:reduce_ccs}}(\BoundProtrusionModulator) \cdot \BoundProtrusionSize$ vertices.
    Recall that $\BoundProtrusionModulator = \Oh(|M|^{17} \cdot k^3) = \Oh(|M|^{20})$ and that $\BoundProtrusionSize = \Oh(1)$.
    Moreover, $f_{\ref{thm:reduce_ccs}}(x) \in O(x^3)$.
    Hence, we get that the kernel size is $\Oh(|M|^{60})$.
    Finally, note that in all cases, the instance we output does not change the modulator $M$, and $k'$ is only different from $k$ if the input instance is a trivial yes-instance.
\end{proof} 
\section{Consequences}
\label{sec:consequences}
\subparagraph*{Applications.}
As discussed in \cref{sec:intro}, \cref{main_thm:uniform_kernel_general} implies multiple results for standard graph classes.

\corollaryUniformKernelsNaturalClasses*{}
\begin{proof}
  First, let us consider \pwEtaDeletion{} parameterized by $k + |M|$, where $M$ is a modulator to $\GG_{\td \leq \gamma}$ given as part of the input.
  We first use an approximation algorithm by Gupta et al.~\cite{guptaLosingTreewidthSeparating2019} to compute a modulator $M'$ to $\GG_{\pw \leq \eta}$ in polynomial-time.
  Either $|M'| \in O(k)$, or the instance is a no-instance, and we can output a trivial no-instance.
  Now, if we do not output a no-instance, note that $\hat M = M' \cup M$ is a modulator to $\GG_{\pw \leq \eta}^{\gamma}$ with $|\hat M| \in O(k + |M|)$, and therefore we can apply our uniform kernel from \cref{main_thm:uniform_kernel_general} on this instance $(G,k,\hat M)$ to obtain the desired kernel.

  The result for \pwEtaDeletionDistToG[{\pwEtaBoundClass[1]}] follows from \cref{main_thm:uniform_kernel_general} by setting $\beta = 0$.

  The result for \pwEtaDeletionDistToG[\GG_{\td \leq \eta + 1}]{} follows from the fact that any graph with treedepth at most $\eta + 1$ has pathwidth at most $\eta$.
  Hence, $\GG_{\td \leq \eta + 1} \subseteq \GG_{\pw \leq \eta}^{\eta + 1}$, and the kernel exists by \cref{main_thm:uniform_kernel_general}.
  It is also essentially implied by (1).

  The result for \pwEtaDeletionDistToG[\GG^{\eta - 1}] follows from \cref{main_thm:uniform_kernel_general} by setting $\beta = \eta - 1$ and the fact that all graphs in $\GG^{\eta - 1}$ have pathwidth at most $\eta$.

  For the result for \pwEtaDeletionDistToG[{\GG_{\pw \leq \eta}^{\pl \leq \gamma}}]{}, note that a graph that does not contain a path on $\gamma + 1$ vertices as a subgraph has treedepth at most $\gamma$~\cite[Proposition 6.1]{nesetrilSparsityGraphsStructures2012}.
  In fact, by a result of Hatzel et al.~\cite[Theorem 3]{hatzelTightBoundTreedepth2024} a graph with pathwidth less than $\eta+1$ that contains no path on $2^a$ vertices has treedepth at most $10(\eta+1) \cdot a$.
  Therefore, the graph class $\GG_{\pw \leq \eta}^{\pl \leq \gamma}$ only consists of graphs with pathwidth less than $\eta+1$, and treedepth at most $\beta$ for a constant $\beta$.
  This yields $\GG_{\pw \leq \eta}^{\pl \leq \gamma} \subseteq \GG_{\pw \leq \eta}^{\beta}$, and the kernel exists by \cref{main_thm:uniform_kernel_general}.
\end{proof}

We also remark that result (2) of \cref{corollary:uniform_kernels_natural_classes} can also be extended to include the case $\eta = 0$ by using the kernel for \textsc{Vertex Cover} parameterized by the distance to a forest of Jansen and Bodlaender~\cite{JansenB2013}.

\subparagraph*{Graph-theoretical consequences.}
Recall the notions of $k$-apices and minor-minimal obstructions.

\defApicesObstructions*{}

As mentioned in the introduction, we can use our kernelization algorithm to obtain uniform bounds on their size (\cref{main_thm:obstruction_bound}, restated below).

\mainThmObstructionBound*{}
\begin{proof}
  Let $\eta, \beta \geq 0$ be arbitrary integers (we allow $\eta = 0$), and let $H$ be a minor-minimal obstruction to the class of $k$-apices of $\pwEtaBoundClass[\eta]$.
  Moreover, let $M \subseteq V(H)$ be a minimum-size set such that $\pwElimDist{H - M} \leq \beta$ and $H - M \in \pwEtaBoundClass[\eta]$.
  By \cref{main_thm:uniform_kernel_general}, when we input instance $(H,k,M)$ of \pwEtaDelElimGeneral{} into our kernelization algorithm, we obtain an equivalent instance $(H',k',M)$ as output.
  Moreover, $H'$ is a minor of $H$ with $\Oh(|M|^{60})$ vertices, and if $(H,k,M)$ is a no-instance, then $k = k'$.
  Since $H$ is not a $k$-apex of $\pwEtaBoundClass[\eta]$, instance $(H,k,M)$ is a no-instance, and therefore also instance $(H',k,M)$ is a no-instance, and $k = k'$.
  Hence, $H'$ is not a $k$-apex of $\pwEtaBoundClass[\eta]$.
  Since $H$ was minor-minimal, we must have $H' = H$, showing that the size of $H$ is bounded by $f(\eta,\beta) \cdot |M|^{60}$ for an appropriately chosen function $f$.
\end{proof}

By using an analogous proof and the kernel of \cref{thm:uniform_kernel_vc_parameterization} we obtain the following result.

\thmObstructionBoundVertexCover*{}

\cref{thm:obstruction_bound_vertex_cover} is significant because the same statement is \emph{not true} for treewidth, as the following result shows.
Note that this result was informally stated by Giannopoulou et al.~\cite{giannopoulouUniformKernelizationComplexity2017}
as it indeed follows quite easily from a construction they use in a kernelization lower bound.
For the sake of completeness, we provide a proof here.

\thmTreewidthObstructionsNonUniform*{}
\begin{proof}
  Let $H_0$ be a graph consisting of a clique $C$ on $\eta + 1 + k$ vertices ($\eta = 0$ is allowed).
  Moreover, for each $X \in \binom{V(C)}{\eta + 1}$, introduce a vertex $v_X$ with neighborhood $X$.

  First, we argue that $H_0$ is not a $k$-apex of $\GG_{\tw \leq \eta}$.
  Consider any set $S \subseteq V(H_0)$ of size~$k$.
  If $|S \cap V(C)| < k$, then $C - S$ contains a clique of size at least $\eta + 2$, and therefore $\tw(H_0 - S) > \eta$.
  If $|S \cap V(C)| = k$, then $S$ only contains vertices of $C$.
  Set $X = V(C) \setminus S$.
  Since $|V(C)| = k + \eta + 1$, we have $|X| = \eta + 1$.
  Then, the vertices $\{v_X\} \cup X$ form a clique of size $\eta + 2$ in $H_0 - S$, which again proves that $\tw(H_0 - S) > \eta$.

  We have established that $H_0$ is not a $k$-apex, however, it is not clear that $H_0$ is a minor-minimal obstruction.
  We process $H_0$ further to obtain the graph $H$, for which we can show this property.
  Concretely, let $E' \subseteq \{uv \in E(H_0) \mid u,v \in V(C)\}$ be a maximal set of edges between vertices of $C$ such that $H_0 - E'$ is still not a $k$-apex of $\GG_{\tw \leq \eta}$.
  Set $H = H_0 - E'$, clearly $H$ is not a $k$-apex of $\GG_{\tw \leq \eta}$.

  Let us now argue that $\vc(H) = k + \eta + 1$.
  Observe that $C$ is a vertex cover of $H$ of size $k + \eta + 1$, and hence $\vc(H) \leq k + \eta + 1$.
  Towards a contradiction, assume that there is a vertex cover $M$ of $H$ of size at most $k + \eta$.
  Then, $\tw(H- M) = 0$.
  So, we must have $|M| \geq k$, otherwise $M$ would certify that $H$ is a $k$-apex of $\GG_{\tw \leq \eta}$, a contradiction.
  Hence, for any $M' \subseteq M$ of size exactly $k$, we have that $H - M'$ has a vertex cover of size at most $\eta$.
  Therefore, $\tw(H - M') \leq \eta$, contradicting that $H$ is not a $k$-apex.

  Moreover, we have $|V(H)| = k + \eta + 1 + \binom{k + \eta + 1}{\eta + 1}$ by construction.
  Therefore, $|V(H)| = k + \eta + 1 + \binom{\vc(H)}{\eta + 1}$, as desired.

  It remains to show that $H$ is a minor-minimal obstruction to the class of $k$-apices of $\GG_{\tw \leq \eta}$.
  By construction graph $H$ is not a $k$-apex.
  We will show that any proper minor of $H$ is a $k$-apex to conclude the proof.
  For this purpose, consider a minor $H'$ of $H$ that is obtained by deleting an edge of $H$, deleting a vertex of $H$, or contracting an edge of $H$.
  By proving that $H'$ is a $k$-apex, we show that any proper minor of $H$ is a $k$-apex.
  We consider three cases depending on how $H'$ is obtained from $H$.
  \begin{description}
    \item[$H'$ is created from $H$ by deleting an edge $uv$:]
          If $u,v \in V(C)$, then $uv \notin E'$.
          By our choice of $E'$, the graph $H'  = H - uv$ is a $k$-apex.

          Therefore, let us consider that $v = v_X$ for some $X \in \binom{V(C)}{\eta + 1}$, and $u \in V(C)$.
          Set $S = V(C) \setminus X$ and observe that $|S| = k$.
          We now argue that $\tw(H' - S) \leq \eta$.

          For the tree decomposition, let $T$ be a star graph with center node $c$ and a leaf node $w_Y$ for each $Y \in \binom{V(C)}{\eta + 1}$.
          The graph $T$ will form the tree of the tree decomposition of $H' - S$.
          For the bags, we set $B_c = X$, and for each node $w_Y$ we set $B_{w_Y} = N_{H' - S}[v_Y]$.
          It is easy to see that $(T,\{B_n\}_{n \in V(T)})$ is a tree decomposition of $H' - S$.
          So, it remains to argue that it has width at most $\eta$.
          Clearly $|B_c| = \eta + 1$, which is sufficiently small.
          For a node $w_Y$ with $Y \neq X$, note that at least one vertex of $Y = N_H(v_Y)$ is part of $S$, and that $|Y| = \eta + 1$.
          Therefore, $|B_{w_Y}| \leq \eta + 1$.
          Finally, for node $w_X$ observe that $|N_{H' - S}(v_X)| \leq \eta$, because the edge $v_X u$ is not part of $H'$, and therefore $|B_{w_X}| \leq \eta + 1$.
          Thus, $H'$ is a $k$-apex.
    \item[$H'$ is created from $H$ by deleting a vertex $v$:] Note that, because $H$ has no isolated vertices, deleting $v$ also deletes an edge incident on $v$.
          Therefore, $H'$ is a proper minor of some graph $H''$ that is created by deleting an edge of $H$. Hence, by the previous case, $H''$ is a $k$-apex. Thus, also $H'$ is a $k$-apex.
    \item[$H'$ is created from $H$ by contracting an edge $uv$:]
          Since edge $uv$ is contracted, $H'$ contains some vertex $z$ that is not part of $H$ and $N_{H'}(z) = (N_H(v) \cup N_H(u)) \setminus \{u,v\}$.
          Set $C' = (V(C) \setminus \{v,u\}) \cup \{z\}$.

          Now, consider the case that $u,v \in V(C)$.
          Then, let $S$ be an arbitrary subset of $V(C')$ of size $k$.
          We will show that $\tw(H' - S) \leq \eta$ by providing a tree decomposition of $H' - S$.
          Let $T$ be a star graph with center vertex $c$ and a leaf $w_Y$ for each $Y \in \binom{V(C)}{\eta + 1}$.
          For the bags of the decomposition, set $B_c = C' \setminus S$, and for each $w_Y$ set $B_{w_Y} = N_{H' - S}[v_Y]$.
          It is clear that $(T,\{B_n\}_{n \in V(T)})$ is a tree decomposition of $H' - S$, so it remains to argue that each bag contains at most $\eta + 1$ vertices.
          Since $|C'| = k + \eta$, we have $|B_c| = \eta$.
          For a node $w_Y$, notice that $N_{H' - S}(v_Y) \subseteq C' \setminus S$, and therefore $|B_{w_{Y}}| \leq \eta + 1$.
          Hence, $H'$ is a $k$-apex.

          \quad Finally, consider that $v = v_X$, for some $X \in \binom{V(C)}{\eta + 1}$, and $u \in V(C)$.
          Let $X' \coloneq (X \setminus \{u\}) \cup \{z\}$, let $S = C' \setminus X'$, and observe that $|S| = k$ and $|X'| = \eta + 1$.
          We will again show that $\tw(H' - S) \leq \eta$ by providing a tree decomposition of $H' - S$.
          Let $T$ be a star graph with center vertex $c$ and a leaf $w_Y$ for each $Y \in \binom{V(C)}{\eta + 1} \setminus \{X\}$.
          We set $B_c = X'$, and for each $Y \in \binom{V(C)}{\eta + 1} \setminus \{X\}$ we set $B_{w_Y} = N_{H' - S}[v_Y]$.
          Again, it is clear that $(T,\{B_n\}_{n \in V(T)})$ is a tree decomposition of $H' - S$, so it remains to bound the width of the decomposition.
          We have $|B_c| = \eta + 1$.
          For each vertex $v_Y$ of $H'$, notice that $N_{H' - S}(v_Y) \neq C' \setminus S = X'$, because vertex $v_X$, the vertex with $N_H(v_X) = X$ is not part of $H'$.
          Therefore, we have $|B_{w_Y}| \leq \eta + 1$ for each node $w_Y$.
  \end{description}

  We have thus established that $H$ is a minor-minimal obstruction to the class of $k$-apices of $\GG_{\tw \leq \eta}$.
\end{proof}

Note that \cref{thm:treewidth_obstructions_non_uniform} illustrates that a result analogous to \cref{thm:obstruction_bound_vertex_cover} is \emph{impossible} for treewidth: the exponent in the size of obstructions in terms of their vertex cover size must always depend on $\eta$ in a non-uniform manner.
 
\newpage
\bibliography{bib}

\begin{thebibliography}{10}

\bibitem{abu-khzamCrownStructuresVertex2007}
Faisal~N. Abu{-}Khzam, Michael~R. Fellows, Michael~A. Langston, and W.~Henry Suters.
\newblock Crown structures for vertex cover kernelization.
\newblock {\em Theory Comput. Syst.}, 41(3):411--430, 2007.
\newblock \href {https://doi.org/10.1007/S00224-007-1328-0} {\path{doi:10.1007/S00224-007-1328-0}}.

\bibitem{arnborgMonadicSecondOrder1990}
Stefan Arnborg, Andrzej Proskurowski, and Detlef Seese.
\newblock Monadic second order logic, tree automata and forbidden minors.
\newblock In Egon B{\"{o}}rger, Hans~Kleine B{\"{u}}ning, Michael~M. Richter, and Wolfgang Sch{\"{o}}nfeld, editors, {\em Computer Science Logic, 4th Workshop, {CSL} '90, Heidelberg, Germany, October 1-5, 1990, Proceedings}, volume 533 of {\em Lecture Notes in Computer Science}, pages 1--16. Springer, 1990.
\newblock \href {https://doi.org/10.1007/3-540-54487-9_49} {\path{doi:10.1007/3-540-54487-9_49}}.

\bibitem{belmonteGrundyDistinguishesTreewidth2022}
R{\'{e}}my Belmonte, Eun~Jung Kim, Michael Lampis, Valia Mitsou, and Yota Otachi.
\newblock Grundy distinguishes treewidth from pathwidth.
\newblock {\em {SIAM} J. Discret. Math.}, 36(3):1761--1787, 2022.
\newblock \href {https://doi.org/10.1137/20M1385779} {\path{doi:10.1137/20M1385779}}.

\bibitem{BIENSTOCK1995481}
Daniel Bienstock and Michael~A. Langston.
\newblock Algorithmic implications of the graph minor theorem.
\newblock In {\em Network Models}, volume~7 of {\em Handbooks in Operations Research and Management Science}, pages 481--502. Elsevier, 1995.
\newblock \href {https://doi.org/10.1016/S0927-0507(05)80125-2} {\path{doi:10.1016/S0927-0507(05)80125-2}}.

\bibitem{bodlaenderLineartimeAlgorithmFinding1996}
Hans~L. Bodlaender.
\newblock A linear-time algorithm for finding tree-decompositions of small treewidth.
\newblock {\em {SIAM} J. Comput.}, 25(6):1305--1317, 1996.
\newblock \href {https://doi.org/10.1137/S0097539793251219} {\path{doi:10.1137/S0097539793251219}}.

\bibitem{DBLP:journals/jacm/BodlaenderFLPST16}
Hans~L. Bodlaender, Fedor~V. Fomin, Daniel Lokshtanov, Eelko Penninkx, Saket Saurabh, and Dimitrios~M. Thilikos.
\newblock {(Meta)} kernelization.
\newblock {\em J. {ACM}}, 63(5):44:1--44:69, 2016.
\newblock \href {https://doi.org/10.1145/2973749} {\path{doi:10.1145/2973749}}.

\bibitem{bodlaenderKernelBoundsStructural2012}
Hans~L. Bodlaender, Bart M.~P. Jansen, and Stefan Kratsch.
\newblock Kernel bounds for structural parameterizations of pathwidth.
\newblock In Fedor~V. Fomin and Petteri Kaski, editors, {\em Algorithm Theory - {SWAT} 2012 - 13th Scandinavian Symposium and Workshops, Helsinki, Finland, July 4-6, 2012. Proceedings}, volume 7357 of {\em Lecture Notes in Computer Science}, pages 352--363. Springer, 2012.
\newblock \href {https://doi.org/10.1007/978-3-642-31155-0_31} {\path{doi:10.1007/978-3-642-31155-0_31}}.

\bibitem{bodlaenderPreprocessingTreewidthCombinatorial2013}
Hans~L. Bodlaender, Bart M.~P. Jansen, and Stefan Kratsch.
\newblock Preprocessing for treewidth: {A} combinatorial analysis through kernelization.
\newblock {\em {SIAM} J. Discret. Math.}, 27(4):2108--2142, 2013.
\newblock \href {https://doi.org/10.1137/120903518} {\path{doi:10.1137/120903518}}.

\bibitem{bodlaenderPreprocessingRulesTriangulation2005}
Hans~L. Bodlaender, Arie M. C.~A. Koster, and Frank van~den Eijkhof.
\newblock Preprocessing rules for triangulation of probabilistic networks.
\newblock {\em Comput. Intell.}, 21(3):286--305, 2005.
\newblock \href {https://doi.org/10.1111/J.1467-8640.2005.00274.X} {\path{doi:10.1111/J.1467-8640.2005.00274.X}}.

\bibitem{bodlaenderPathwidthTreewidthCographs1993a}
Hans~L. Bodlaender and Rolf~H. M{\"{o}}hring.
\newblock The pathwidth and treewidth of cographs.
\newblock {\em {SIAM} J. Discret. Math.}, 6(2):181--188, 1993.
\newblock \href {https://doi.org/10.1137/0406014} {\path{doi:10.1137/0406014}}.

\bibitem{bougeretKernelizationDichotomiesHitting2025a}
Marin Bougeret, Eric Brandwein, and Ignasi Sau.
\newblock Kernelization dichotomies for hitting minors under structural parameterizations.
\newblock In Meena Mahajan, Florin Manea, Annabelle McIver, and Kim~Thang Nguyen, editors, {\em 43rd International Symposium on Theoretical Aspects of Computer Science, {STACS} 2026, Grenoble, France, March 9-13, 2026}, volume 364 of {\em LIPIcs}, pages 17:1--17:19. Schloss Dagstuhl - Leibniz-Zentrum f{\"{u}}r Informatik, 2026.
\newblock \href {https://doi.org/10.4230/LIPICS.STACS.2026.17} {\path{doi:10.4230/LIPICS.STACS.2026.17}}.

\bibitem{bryant1987finding}
R.~L. Bryant, Michael~R. Fellows, N.~G. Kinnersley, and Michael~A. Langston.
\newblock On finding obstruction sets and polynomial-time algorithms for gate matrix layout.
\newblock In {\em Proc. 25th Allerton Conf. on Communication, Control and Computing}, pages 397--398, 1987.

\bibitem{bulianGraphIsomorphismParameterized2016}
Jannis Bulian and Anuj Dawar.
\newblock Graph isomorphism parameterized by elimination distance to bounded degree.
\newblock {\em Algorithmica}, 75(2):363--382, 2016.
\newblock \href {https://doi.org/10.1007/S00453-015-0045-3} {\path{doi:10.1007/S00453-015-0045-3}}.

\bibitem{bulianFixedparameterTractableDistances2017}
Jannis Bulian and Anuj Dawar.
\newblock Fixed-parameter tractable distances to sparse graph classes.
\newblock {\em Algorithmica}, 79(1):139--158, 2017.
\newblock \href {https://doi.org/10.1007/S00453-016-0235-7} {\path{doi:10.1007/S00453-016-0235-7}}.

\bibitem{bussNondeterminism1993}
Jonathan~F. Buss and Judy Goldsmith.
\newblock Nondeterminism within {P}.
\newblock {\em {SIAM} J. Comput.}, 22(3):560--572, 1993.
\newblock \href {https://doi.org/10.1137/0222038} {\path{doi:10.1137/0222038}}.

\bibitem{DBLP:journals/jacm/ChekuriC16}
Chandra Chekuri and Julia Chuzhoy.
\newblock Polynomial bounds for the grid-minor theorem.
\newblock {\em J. {ACM}}, 63(5):40:1--40:65, 2016.
\newblock \href {https://doi.org/10.1145/2820609} {\path{doi:10.1145/2820609}}.

\bibitem{chenVertexCoverFurther2001}
Jianer Chen, Iyad~A. Kanj, and Weijia Jia.
\newblock Vertex cover: Further observations and further improvements.
\newblock {\em J. Algorithms}, 41(2):280--301, 2001.
\newblock \href {https://doi.org/10.1006/JAGM.2001.1186} {\path{doi:10.1006/JAGM.2001.1186}}.

\bibitem{DBLP:conf/wg/ChorFJ04}
Benny Chor, Mike Fellows, and David~W. Juedes.
\newblock Linear kernels in linear time, or how to save $k$ colors in {$O(n^2)$} steps.
\newblock In Juraj Hromkovic, Manfred Nagl, and Bernhard Westfechtel, editors, {\em Graph-Theoretic Concepts in Computer Science, 30th International Workshop, {WG} 2004, Bad Honnef, Germany, June 21-23, 2004, Revised Papers}, volume 3353 of {\em Lecture Notes in Computer Science}, pages 257--269. Springer, 2004.
\newblock \href {https://doi.org/10.1007/978-3-540-30559-0_22} {\path{doi:10.1007/978-3-540-30559-0_22}}.

\bibitem{clautiauxNewLowerUpper2003}
Fran{\c{c}}ois Clautiaux, Jacques Carlier, Aziz Moukrim, and St{\'{e}}phane N{\`{e}}gre.
\newblock New lower and upper bounds for graph treewidth.
\newblock In Klaus Jansen, Marian Margraf, Monaldo Mastrolilli, and Jos{\'{e}} D.~P. Rolim, editors, {\em Experimental and Efficient Algorithms, Second International Workshop, {WEA} 2003, Ascona, Switzerland, May 26-28, 2003, Proceedings}, volume 2647 of {\em Lecture Notes in Computer Science}, pages 70--80. Springer, 2003.
\newblock \href {https://doi.org/10.1007/3-540-44867-5_6} {\path{doi:10.1007/3-540-44867-5_6}}.

\bibitem{cyganParameterizedAlgorithms2015}
Marek Cygan, Fedor~V. Fomin, Lukasz Kowalik, Daniel Lokshtanov, D{\'{a}}niel Marx, Marcin Pilipczuk, Michal Pilipczuk, and Saket Saurabh.
\newblock {\em Parameterized Algorithms}.
\newblock Springer, 2015.
\newblock \href {https://doi.org/10.1007/978-3-319-21275-3} {\path{doi:10.1007/978-3-319-21275-3}}.

\bibitem{cyganHardnessLosingWidth2014}
Marek Cygan, Daniel Lokshtanov, Marcin Pilipczuk, Michal Pilipczuk, and Saket Saurabh.
\newblock On the hardness of losing width.
\newblock {\em Theory Comput. Syst.}, 54(1):73--82, 2014.
\newblock \href {https://doi.org/10.1007/S00224-013-9480-1} {\path{doi:10.1007/S00224-013-9480-1}}.

\bibitem{cyganImprovedFPTAlgorithm2012}
Marek Cygan, Marcin Pilipczuk, Michal Pilipczuk, and Jakub~Onufry Wojtaszczyk.
\newblock An improved {FPT} algorithm and a quadratic kernel for pathwidth one vertex deletion.
\newblock {\em Algorithmica}, 64(1):170--188, 2012.
\newblock \href {https://doi.org/10.1007/S00453-011-9578-2} {\path{doi:10.1007/S00453-011-9578-2}}.

\bibitem{DBLP:conf/soda/DellM12}
Holger Dell and D{\'{a}}niel Marx.
\newblock Kernelization of packing problems.
\newblock In Yuval Rabani, editor, {\em Proceedings of the Twenty-Third Annual {ACM-SIAM} Symposium on Discrete Algorithms, {SODA} 2012, Kyoto, Japan, January 17-19, 2012}, pages 68--81. {SIAM}, 2012.
\newblock \href {https://doi.org/10.1137/1.9781611973099.6} {\path{doi:10.1137/1.9781611973099.6}}.

\bibitem{DBLP:journals/corr/abs-1812-03155}
Holger Dell and D{\'{a}}niel Marx.
\newblock Kernelization of packing problems.
\newblock {\em CoRR}, abs/1812.03155, 2018.
\newblock \href {https://arxiv.org/abs/1812.03155} {\path{arXiv:1812.03155}}.

\bibitem{DellM2014}
Holger Dell and Dieter van Melkebeek.
\newblock Satisfiability allows no nontrivial sparsification unless the polynomial-time hierarchy collapses.
\newblock {\em J. {ACM}}, 61(4):23:1--23:27, 2014.
\newblock \href {https://doi.org/10.1145/2629620} {\path{doi:10.1145/2629620}}.

\bibitem{donkersPreprocessingOuterplanarVertex2022}
Huib Donkers, Bart M.~P. Jansen, and Michal Wlodarczyk.
\newblock Preprocessing for outerplanar vertex deletion: An elementary kernel of quartic size.
\newblock {\em Algorithmica}, 84(11):3407--3458, 2022.
\newblock \href {https://doi.org/10.1007/S00453-022-00984-2} {\path{doi:10.1007/S00453-022-00984-2}}.

\bibitem{fellowsWhatKnownVertex2018}
Michael~R. Fellows, Lars Jaffke, Aliz~Izabella Kir{\'{a}}ly, Frances~A. Rosamond, and Mathias Weller.
\newblock What is known about vertex cover kernelization?
\newblock In Hans{-}Joachim B{\"{o}}ckenhauer, Dennis Komm, and Walter Unger, editors, {\em Adventures Between Lower Bounds and Higher Altitudes - Essays Dedicated to Juraj Hromkovi{\v{c}} on the Occasion of His 60th Birthday}, volume 11011 of {\em Lecture Notes in Computer Science}, pages 330--356. Springer, 2018.
\newblock \href {https://doi.org/10.1007/978-3-319-98355-4_19} {\path{doi:10.1007/978-3-319-98355-4_19}}.

\bibitem{DBLP:journals/algorithmica/FellowsKNRRSTW08}
Michael~R. Fellows, Christian Knauer, Naomi Nishimura, Prabhakar Ragde, Frances~A. Rosamond, Ulrike Stege, Dimitrios~M. Thilikos, and Sue Whitesides.
\newblock Faster fixed-parameter tractable algorithms for matching and packing problems.
\newblock {\em Algorithmica}, 52(2):167--176, 2008.
\newblock \href {https://doi.org/10.1007/S00453-007-9146-Y} {\path{doi:10.1007/S00453-007-9146-Y}}.

\bibitem{DBLP:journals/jcss/FominJP14}
Fedor~V. Fomin, Bart M.~P. Jansen, and Michal Pilipczuk.
\newblock Preprocessing subgraph and minor problems: When does a small vertex cover help?
\newblock {\em J. Comput. Syst. Sci.}, 80(2):468--495, 2014.
\newblock \href {https://doi.org/10.1016/J.JCSS.2013.09.004} {\path{doi:10.1016/J.JCSS.2013.09.004}}.

\bibitem{fominHittingForbiddenMinors2016}
Fedor~V. Fomin, Daniel Lokshtanov, Neeldhara Misra, Geevarghese Philip, and Saket Saurabh.
\newblock Hitting forbidden minors: Approximation and kernelization.
\newblock {\em {SIAM} J. Discret. Math.}, 30(1):383--410, 2016.
\newblock \href {https://doi.org/10.1137/140997889} {\path{doi:10.1137/140997889}}.

\bibitem{fominPlanarFdeletionApproximation2012}
Fedor~V. Fomin, Daniel Lokshtanov, Neeldhara Misra, and Saket Saurabh.
\newblock Planar {F}-deletion: Approximation, kernelization and optimal {FPT} algorithms.
\newblock In {\em 53rd Annual {IEEE} Symposium on Foundations of Computer Science, {FOCS} 2012, New Brunswick, NJ, USA, October 20-23, 2012}, pages 470--479. {IEEE} Computer Society, 2012.
\newblock \href {https://doi.org/10.1109/FOCS.2012.62} {\path{doi:10.1109/FOCS.2012.62}}.

\bibitem{DBLP:journals/siamcomp/FominLST20}
Fedor~V. Fomin, Daniel Lokshtanov, Saket Saurabh, and Dimitrios~M. Thilikos.
\newblock Bidimensionality and kernels.
\newblock {\em {SIAM} J. Comput.}, 49(6):1397--1422, 2020.
\newblock \href {https://doi.org/10.1137/16M1080264} {\path{doi:10.1137/16M1080264}}.

\bibitem{giannopoulouUniformKernelizationComplexity2017}
Archontia~C. Giannopoulou, Bart M.~P. Jansen, Daniel Lokshtanov, and Saket Saurabh.
\newblock Uniform kernelization complexity of hitting forbidden minors.
\newblock {\em {ACM} Trans. Algorithms}, 13(3):35:1--35:35, 2017.
\newblock \href {https://doi.org/10.1145/3029051} {\path{doi:10.1145/3029051}}.

\bibitem{guptaLosingTreewidthSeparating2019}
Anupam Gupta, Euiwoong Lee, Jason Li, Pasin Manurangsi, and Michal Wlodarczyk.
\newblock Losing treewidth by separating subsets.
\newblock In Timothy~M. Chan, editor, {\em Proceedings of the Thirtieth Annual {ACM-SIAM} Symposium on Discrete Algorithms, {SODA} 2019, San Diego, California, USA, January 6-9, 2019}, pages 1731--1749. {SIAM}, 2019.
\newblock \href {https://doi.org/10.1137/1.9781611975482.104} {\path{doi:10.1137/1.9781611975482.104}}.

\bibitem{hatzelTightBoundTreedepth2024}
Meike Hatzel, Gwena{\"{e}}l Joret, Piotr Micek, Marcin Pilipczuk, Torsten Ueckerdt, and Bartosz Walczak.
\newblock Tight bound on treedepth in terms of pathwidth and longest path.
\newblock {\em Comb.}, 44(2):417--427, 2024.
\newblock \href {https://doi.org/10.1007/S00493-023-00077-W} {\path{doi:10.1007/S00493-023-00077-W}}.

\bibitem{DBLP:journals/siamdm/HeggernesHLP13}
Pinar Heggernes, Pim van~'t Hof, Daniel Lokshtanov, and Christophe Paul.
\newblock Obtaining a bipartite graph by contracting few edges.
\newblock {\em {SIAM} J. Discret. Math.}, 27(4):2143--2156, 2013.
\newblock \href {https://doi.org/10.1137/130907392} {\path{doi:10.1137/130907392}}.

\bibitem{hols2022elimination}
Eva{-}Maria~C. Hols, Stefan Kratsch, and Astrid Pieterse.
\newblock Elimination distances, blocking sets, and kernels for vertex cover.
\newblock {\em {SIAM} J. Discret. Math.}, 36(3):1955--1990, 2022.
\newblock \href {https://doi.org/10.1137/20M1335285} {\path{doi:10.1137/20M1335285}}.

\bibitem{hornThreeResultsTrees1972}
W.~A. Horn.
\newblock Three results for trees, using mathematical induction.
\newblock {\em Journal of Research of the National Bureau of Standards, Section B: Mathematical Sciences}, 76B(1-2):39--43, 1972.
\newblock \href {https://doi.org/10.6028/JRES.076B.002} {\path{doi:10.6028/JRES.076B.002}}.

\bibitem{JansenB2013}
Bart M.~P. Jansen and Hans~L. Bodlaender.
\newblock Vertex cover kernelization revisited - upper and lower bounds for a refined parameter.
\newblock {\em Theory Comput. Syst.}, 53(2):263--299, 2013.
\newblock \href {https://doi.org/10.1007/S00224-012-9393-4} {\path{doi:10.1007/S00224-012-9393-4}}.

\bibitem{jansen2020polynomial}
Bart M.~P. Jansen and Astrid Pieterse.
\newblock Polynomial kernels for hitting forbidden minors under structural parameterizations.
\newblock {\em Theor. Comput. Sci.}, 841:124--166, 2020.
\newblock \href {https://doi.org/10.1016/J.TCS.2020.07.009} {\path{doi:10.1016/J.TCS.2020.07.009}}.

\bibitem{jansenLossyPlanarizationConstantfactor2025}
Bart M.~P. Jansen and Michal Wlodarczyk.
\newblock Lossy planarization: {A} constant-factor approximate kernelization for planar vertex deletion.
\newblock {\em {SIAM} J. Comput.}, 54(1):1--91, 2025.
\newblock \href {https://doi.org/10.1137/22M152058X} {\path{doi:10.1137/22M152058X}}.

\bibitem{kimLinearKernelsSingleexponential2016}
Eun~Jung Kim, Alexander Langer, Christophe Paul, Felix Reidl, Peter Rossmanith, Ignasi Sau, and Somnath Sikdar.
\newblock Linear kernels and single-exponential algorithms via protrusion decompositions.
\newblock {\em {ACM} Trans. Algorithms}, 12(2):21:1--21:41, 2016.
\newblock \href {https://doi.org/10.1145/2797140} {\path{doi:10.1145/2797140}}.

\bibitem{DBLP:journals/dam/KinnersleyL94}
Nancy~G. Kinnersley and Michael~A. Langston.
\newblock Obstruction set isolation for the gate matrix layout problem.
\newblock {\em Discret. Appl. Math.}, 54(2-3):169--213, 1994.
\newblock \href {https://doi.org/10.1016/0166-218X(94)90021-3} {\path{doi:10.1016/0166-218X(94)90021-3}}.

\bibitem{DBLP:conf/focs/KorhonenPS24}
Tuukka Korhonen, Michal Pilipczuk, and Giannos Stamoulis.
\newblock Minor containment and disjoint paths in almost-linear time.
\newblock In {\em 65th {IEEE} Annual Symposium on Foundations of Computer Science, {FOCS} 2024, Chicago, IL, USA, October 27-30, 2024}, pages 53--61. {IEEE}, 2024.
\newblock \href {https://doi.org/10.1109/FOCS61266.2024.00014} {\path{doi:10.1109/FOCS61266.2024.00014}}.

\bibitem{DBLP:journals/jacm/KratschW20}
Stefan Kratsch and Magnus Wahlstr{\"{o}}m.
\newblock Representative sets and irrelevant vertices: New tools for kernelization.
\newblock {\em J. {ACM}}, 67(3):16:1--16:50, 2020.
\newblock \href {https://doi.org/10.1145/3390887} {\path{doi:10.1145/3390887}}.

\bibitem{kumar2lkKernelLComponent2016}
Mithilesh Kumar and Daniel Lokshtanov.
\newblock A $2\ell k$ kernel for $\ell$-component order connectivity.
\newblock In Jiong Guo and Danny Hermelin, editors, {\em 11th International Symposium on Parameterized and Exact Computation, {IPEC} 2016, Aarhus, Denmark, August 24-26, 2016}, volume~63 of {\em LIPIcs}, pages 20:1--20:14. Schloss Dagstuhl - Leibniz-Zentrum f{\"{u}}r Informatik, 2016.
\newblock \href {https://doi.org/10.4230/LIPICS.IPEC.2016.20} {\path{doi:10.4230/LIPICS.IPEC.2016.20}}.

\bibitem{lampisKernelOrder22011}
Michael Lampis.
\newblock A kernel of order $2 k-c \log k$ for vertex cover.
\newblock {\em Inf. Process. Lett.}, 111(23-24):1089--1091, 2011.
\newblock \href {https://doi.org/10.1016/J.IPL.2011.09.003} {\path{doi:10.1016/J.IPL.2011.09.003}}.

\bibitem{DBLP:conf/isaac/LochetS24}
William Lochet and Roohani Sharma.
\newblock Uniform polynomial kernel for deletion to {$K_{2, p}$} minor-free graphs.
\newblock In Juli{\'{a}}n Mestre and Anthony Wirth, editors, {\em 35th International Symposium on Algorithms and Computation, {ISAAC} 2024, Sydney, Australia, December 8-11, 2024}, volume 322 of {\em LIPIcs}, pages 46:1--46:14. Schloss Dagstuhl - Leibniz-Zentrum f{\"{u}}r Informatik, 2024.
\newblock \href {https://doi.org/10.4230/LIPICS.ISAAC.2024.46} {\path{doi:10.4230/LIPICS.ISAAC.2024.46}}.

\bibitem{nesetrilSparsityGraphsStructures2012}
Jaroslav Nesetril and Patrice~Ossona de~Mendez.
\newblock {\em Sparsity - Graphs, Structures, and Algorithms}, volume~28 of {\em Algorithms and {Combinatorics}}.
\newblock Springer, 2012.
\newblock \href {https://doi.org/10.1007/978-3-642-27875-4} {\path{doi:10.1007/978-3-642-27875-4}}.

\bibitem{philipQuarticKernelPathwidthone2010b}
Geevarghese Philip, Venkatesh Raman, and Yngve Villanger.
\newblock A quartic kernel for pathwidth-one vertex deletion.
\newblock In Dimitrios~M. Thilikos, editor, {\em Graph Theoretic Concepts in Computer Science - 36th International Workshop, {WG} 2010, Zar{\'{o}}s, Crete, Greece, June 28-30, 2010, Revised Papers}, volume 6410 of {\em Lecture Notes in Computer Science}, pages 196--207. Springer, 2010.
\newblock \href {https://doi.org/10.1007/978-3-642-16926-7_19} {\path{doi:10.1007/978-3-642-16926-7_19}}.

\bibitem{reidlFasterParameterizedAlgorithm2014}
Felix Reidl, Peter Rossmanith, Fernando~S{\'{a}}nchez Villaamil, and Somnath Sikdar.
\newblock A faster parameterized algorithm for treedepth.
\newblock In Javier Esparza, Pierre Fraigniaud, Thore Husfeldt, and Elias Koutsoupias, editors, {\em Automata, Languages, and Programming - 41st International Colloquium, {ICALP} 2014, Copenhagen, Denmark, July 8-11, 2014, Proceedings, Part {I}}, volume 8572 of {\em Lecture Notes in Computer Science}, pages 931--942. Springer, 2014.
\newblock \href {https://doi.org/10.1007/978-3-662-43948-7_77} {\path{doi:10.1007/978-3-662-43948-7_77}}.

\bibitem{DBLP:journals/jct/RobertsonS83}
Neil Robertson and Paul~D. Seymour.
\newblock Graph minors. {I. Excluding} a forest.
\newblock {\em J. Comb. Theory {B}}, 35(1):39--61, 1983.
\newblock \href {https://doi.org/10.1016/0095-8956(83)90079-5} {\path{doi:10.1016/0095-8956(83)90079-5}}.

\bibitem{DBLP:journals/jct/RobertsonS86}
Neil Robertson and Paul~D. Seymour.
\newblock Graph minors. {V. Excluding} a planar graph.
\newblock {\em J. Comb. Theory {B}}, 41(1):92--114, 1986.
\newblock \href {https://doi.org/10.1016/0095-8956(86)90030-4} {\path{doi:10.1016/0095-8956(86)90030-4}}.

\bibitem{DBLP:journals/jct/RobertsonS95b}
Neil Robertson and Paul~D. Seymour.
\newblock Graph minors. {XIII. The} disjoint paths problem.
\newblock {\em J. Comb. Theory {B}}, 63(1):65--110, 1995.
\newblock \href {https://doi.org/10.1006/JCTB.1995.1006} {\path{doi:10.1006/JCTB.1995.1006}}.

\bibitem{DBLP:journals/jct/RobertsonS04}
Neil Robertson and Paul~D. Seymour.
\newblock Graph minors. {XX. Wagner's} conjecture.
\newblock {\em J. Comb. Theory {B}}, 92(2):325--357, 2004.
\newblock \href {https://doi.org/10.1016/J.JCTB.2004.08.001} {\path{doi:10.1016/J.JCTB.2004.08.001}}.

\bibitem{DBLP:conf/icalp/SauST20}
Ignasi Sau, Giannos Stamoulis, and Dimitrios~M. Thilikos.
\newblock An {FPT}-algorithm for recognizing $k$-apices of minor-closed graph classes.
\newblock In Artur Czumaj, Anuj Dawar, and Emanuela Merelli, editors, {\em 47th International Colloquium on Automata, Languages, and Programming, {ICALP} 2020, Saarbr{\"{u}}cken, Germany (Virtual Conference), July 8-11, 2020}, volume 168 of {\em LIPIcs}, pages 95:1--95:20. Schloss Dagstuhl - Leibniz-Zentrum f{\"{u}}r Informatik, 2020.
\newblock \href {https://doi.org/10.4230/LIPICS.ICALP.2020.95} {\path{doi:10.4230/LIPICS.ICALP.2020.95}}.

\bibitem{DBLP:journals/jctb/SauST23}
Ignasi Sau, Giannos Stamoulis, and Dimitrios~M. Thilikos.
\newblock $k$-apices of minor-closed graph classes. {I. Bounding} the obstructions.
\newblock {\em J. Comb. Theory {B}}, 161:180--227, 2023.
\newblock \href {https://doi.org/10.1016/J.JCTB.2023.02.012} {\path{doi:10.1016/J.JCTB.2023.02.012}}.

\bibitem{soleimanfallahKernelOrder2kc2011}
Arezou Soleimanfallah and Anders Yeo.
\newblock A kernel of order $2k-c$ for vertex cover.
\newblock {\em Discret. Math.}, 311(10-11):892--895, 2011.
\newblock \href {https://doi.org/10.1016/J.DISC.2011.02.014} {\path{doi:10.1016/J.DISC.2011.02.014}}.

\bibitem{DBLP:journals/talg/Thomasse10}
St{\'{e}}phan Thomass{\'{e}}.
\newblock A $4k^2$ kernel for feedback vertex set.
\newblock {\em {ACM} Trans. Algorithms}, 6(2):32:1--32:8, 2010.
\newblock \href {https://doi.org/10.1145/1721837.1721848} {\path{doi:10.1145/1721837.1721848}}.

\bibitem{tsurSmallerKernelsTwo2024}
Dekel Tsur.
\newblock Smaller kernels for two vertex deletion problems.
\newblock {\em Inf. Process. Lett.}, 186:106493, 2024.
\newblock \href {https://doi.org/10.1016/J.IPL.2024.106493} {\path{doi:10.1016/J.IPL.2024.106493}}.

\bibitem{xiaoLinearKernelsSeparating2017}
Mingyu Xiao.
\newblock Linear kernels for separating a graph into components of bounded size.
\newblock {\em J. Comput. Syst. Sci.}, 88:260--270, 2017.
\newblock \href {https://doi.org/10.1016/J.JCSS.2017.04.004} {\path{doi:10.1016/J.JCSS.2017.04.004}}.

\end{thebibliography}

\appendix
\section{List of Constants and Functions}
\label{sec:appendix_constants}
This appendix provides a consolidated list of constants and functions used in the analysis and design of our reduction rules.
These values are typically dependent on the constants $\eta$ (pathwidth bound) and $\beta$ (elimination distance bound), and the size of the modulator $|M|$ or other derived quantities.

\subsection{General Functions}

\begin{itemize}
    \item $f_{\ref{thm:find_stable_graph}}(x)$:  A function used in \cref{thm:find_stable_graph} to determine the minimum required number of 
    vertex-disjoint subgraphs to guarantee that one of them is stable.  It is also used in \cref{def:important_bounds} to define other constants. It is defined as:
    \[ f_{\ref{thm:find_stable_graph}}(x) = 2x \cdot (\eta + 1) + 1 = \Oh(x). \]

    \item $f_{\ref{thm:reduce_ccs}}(x)$: A function used in \cref{thm:reduce_ccs} to bound the number of connected components. It is defined as:
    \[ f_{\ref{thm:reduce_ccs}}(x) = x^2 (2x (\eta+2) +1) = \Oh(x^3). \]

    \item $f_{\ref{thm:tree_lemma_to_find_siblings}}(c,x)$: A function used in \cref{thm:tree_lemma_to_find_siblings} to determine the minimum required size of a vertex set $X$ at a certain depth in a rooted tree
    to apply the lemma. It is defined as:
    \[ f_{\ref{thm:tree_lemma_to_find_siblings}}(c,x) = (\beta+1) \cdot c^{\beta+1} \cdot x = \Oh(c^{\beta+1} \cdot x). \]

    \item $f_{\ref{thm:caterpillar_case}}(x)$: A function used in \cref{thm:caterpillar_case} to bound the overall size of a caterpillar ($|V(C)|$). It is defined as:
    \[ f_{\ref{thm:caterpillar_case}}(x) = (8(\eta+1)(x+5))^8 = \Oh(x^8). \]

    \item $f_{\ref{lem:delete_edge_to_caterpillar}}(x)$: A function used in \cref{lem:delete_edge_to_caterpillar} to bound the number of boundary vertices in a caterpillar. It is defined as:
    \[ f_{\ref{lem:delete_edge_to_caterpillar}}(x) = (8(\eta+1))^3 (x + 5)^8 = \Oh(x^8). \]
\end{itemize}
\subsection{Bounds for the Modulator Degree Reduction}
\label{sub:modulator_degree_bounds}

These bounds are defined in \cref{def:important_bounds_degree_reduction} and are used within the \nameref{rule:reduce_modulator_degree} reduction rule.

\begin{itemize}
    \item $\polishingSetSize[x,y]$: The maximum size of a polishing set $Y$ is $\polishingSetSize$.
    \[ \polishingSetSize[x,y] = x^2 \cdot (y + \eta + 2) \cdot (\eta + 1) = \Oh(x^2 \cdot y). \]
    \item $\degreeBoundCaterpillarCase[x]$: the threshold on the number of neighbors a vertex $m$ has in a single leaf bag that triggers the caterpillar reduction rule (\cref{rule:reduce_modulator_degree}). Note that $\degreeBoundCaterpillarCase[x] = f_{\ref{thm:caterpillar_case}}(x + 2(\eta + 1) + \beta) + 1$.
    \[ \degreeBoundCaterpillarCase[x] = \BoundCaterpillarOverall[x] + 1 = \Oh(x^{8}). \]
    \item $\degreeBoundCC[x]$: $\degreeBoundCC$ is a bound on the degree of a modulator vertex within a connected component of $G - (M \cup Y)$.
    \[ \degreeBoundCC[x] = \degreeBoundTreedepthCase(x + 2(\eta + 1)) \cdot \degreeBoundCaterpillarCase[x] = \Oh(x^{10}). \]
    \item $\ccBound[x,y]$: $\ccBound$ is a bound on the number of connected components of $G - (M \cup Y)$.
    \[ \ccBound[x,y] = (\polishingSetSize[x,y] + x)^2 \cdot (2(\polishingSetSize[x,y] + x)(\eta+2) + 1) = \Oh(x^6 y^3). \] \item $\degreeBound[x,y]$: The overall bound on the number of neighbors of a modulator vertex in $G - M$ is $\degreeBound$.
    \[ \degreeBound[x,y] = \degreeBoundCC[x] \cdot \ccBound[x,y] + \polishingSetSize[x,y] = \Oh(x^{16} \cdot y^3 ). \] \end{itemize}

\subsection{Bounds for the Degree Reduction in Elimination Forests}
\label{sub:elimination_forest_bounds}

These bounds are defined in \cref{def:important_bounds} and are used within the \nameref{rule:degree_reduction_elimination_forest} reduction rule.

\begin{itemize}
    \item $\boundOnS$: An upper bound on the size of the intersection of a solution $S$ with a near-protrusion~$C$.
    \[ \boundOnS = 3(\eta+1)  = \Oh(1). \]
    \item $\numbChildren$: The minimum required number of children of a node in an elimination tree needed to guarantee the existence of a stable subgraph (\cref{thm:find_stable_graph}).
    \[ \numbChildren = f_{\ref{thm:find_stable_graph}}(\boundOnS+\beta) + \boundOnS + 1 = \Oh(1). \]
    \item $\markedNumbChildren$: A derived bound on the number of children to guarantee the existence of siblings with the same ancestor-type.
    \[ \markedNumbChildren = 2^\beta \cdot \numbChildren = \Oh(1). \]
    \item $\markedNumb$: The number of vertices marked during a step in the marking scheme of \nameref{rule:degree_reduction_elimination_forest}.
    \[ \markedNumb = f_{\ref{thm:tree_lemma_to_find_siblings}}(\markedNumbChildren,1) = (\beta + 1) \cdot \left( \markedNumbChildren \right)^{\beta +1} = \Oh(1). \]
    \item $\numNeighborsInDepth(x)$: A function to calculate the number of neighbors at a specific depth in an elimination forest to apply \Cref{thm:tree_lemma_to_find_siblings}.
    \[ \numNeighborsInDepth(x) = f_{\ref{thm:tree_lemma_to_find_siblings}} \big(\markedNumb, \tbinom{x}{\leq 2}\big) = (\beta + 1) \cdot \left( \markedNumb \right)^{\beta + 1} \cdot \tbinom{x}{\leq 2} = \Oh(x^2). \]
    \item $\degreeBoundTreedepthCase(x)$: The overall degree bound used to apply \Cref{thm:treedepth_case}.
    \[ \degreeBoundTreedepthCase(x) = (\beta + 1) \cdot \numNeighborsInDepth(x) = \Oh(x^2). \]
\end{itemize}
\subsection{Bounds for Protrusion Creation}
\label{sub:protrusion_creation_bounds}

This bound is used in \cref{thm:create_protrusions} to bound the new modulator size when creating protrusions.

\begin{itemize}
    \item $\BoundProtrusionModulator[x,y]$: The maximum size of the new modulator after the protrusion creation step is $\BoundProtrusionModulator$.
    \[ \BoundProtrusionModulator[x,y] = x + x \cdot \degreeBound[x,y] \cdot (\eta + 1) = \Oh(x^{17} \cdot y^3). \] \end{itemize}

\subsection{Bounds for Protrusion Size Reduction}
\label{sub:protrusion_size_reduction_bounds}

These bounds are defined in \cref{def:important_bounds_protrusion_replacement} and are used within the \nameref{rule:reduce_protrusions} reduction rule.

\begin{itemize}
    \item $\BoundProtrusionCaterpillarNeighborhood$: A bound on the size of the neighborhood of a connected component induced by a subtree of an elimination tree, used in protrusion size reduction.
    \[ \BoundProtrusionCaterpillarNeighborhood = 2(\eta + 1) + \beta = \Oh(1). \]
    \item $\BoundChildrenInForest$: A bound on the number of children a node in an elimination tree can have before a reduction rule is applied.
    \[ \BoundChildrenInForest = f_{\ref{thm:reduce_ccs}}( \BoundProtrusionCaterpillarNeighborhood ). = \Oh(1). \]
    \item $\BoundProtrusionSize$: The overall bound on the size of a protrusion guaranteed by \Cref{thm:reduce_protrusions}.
    \[ \BoundProtrusionSize = (\BoundChildrenInForest)^{\beta+1} \cdot f_{\ref{thm:caterpillar_case}}( \BoundProtrusionCaterpillarNeighborhood ) = \Oh(1). \]
\end{itemize} 
\end{document}